\documentclass[envcountsame,11pt]{llncs}

\makeatletter
\DeclareFontFamily{OMX}{MnSymbolE}{}
\DeclareSymbolFont{MnLargeSymbols}{OMX}{MnSymbolE}{m}{n}
\SetSymbolFont{MnLargeSymbols}{bold}{OMX}{MnSymbolE}{b}{n}
\DeclareFontShape{OMX}{MnSymbolE}{m}{n}{
    <-6>  MnSymbolE5
   <6-7>  MnSymbolE6
   <7-8>  MnSymbolE7
   <8-9>  MnSymbolE8
   <9-10> MnSymbolE9
  <10-12> MnSymbolE10
  <12->   MnSymbolE12
}{}
\DeclareFontShape{OMX}{MnSymbolE}{b}{n}{
    <-6>  MnSymbolE-Bold5
   <6-7>  MnSymbolE-Bold6
   <7-8>  MnSymbolE-Bold7
   <8-9>  MnSymbolE-Bold8
   <9-10> MnSymbolE-Bold9
  <10-12> MnSymbolE-Bold10
  <12->   MnSymbolE-Bold12
}{}

\let\llangle\@undefined
\let\rrangle\@undefined
\DeclareMathDelimiter{\llangle}{\mathopen}%
                     {MnLargeSymbols}{'164}{MnLargeSymbols}{'164}
\DeclareMathDelimiter{\rrangle}{\mathclose}%
                     {MnLargeSymbols}{'171}{MnLargeSymbols}{'171}
\makeatother

\usepackage{xspace}
\usepackage[dvipsnames]{xcolor}
\usepackage{amsmath,amssymb}
\usepackage{enumerate}
\usepackage{enumitem}
\usepackage{extarrows}
\usepackage{booktabs}
\usepackage{pifont}
\usepackage{cite}
\usepackage{graphicx}
\usepackage{environ}
\usepackage{url}
\usepackage{mleftright}
\usepackage{multirow}
\usepackage[top=1in, left=1.2in, right=1.2in, bottom=1in]{geometry}

\newcommand{\pconst}{k}
\newcommand{\pconsts}{k_\sys}

\newcommand{\pconstse}{k_\both}
\newcommand{\pconsta}{k_\env}
\newcommand{\pconstb}{k_\both}
\newcommand{\cut}{\hat\nproc}

\newcommand{\mconf}{\mathsf{m}}

\newcommand{\FClass}{\mathfrak{F}}
\newcommand{\Synthesis}[4]{\textsc{Synth}(#1,#2,#3,#4)}
\newcommand{\Gameproblem}[3]{\textsc{Game}(#1,#2,#3)}

\newcommand{\confset}{\C}

\newcommand{\myparagraph}[1]{\noindent\emph{#1.}~}

\usepackage{stmaryrd}
\usepackage{qtree}
\DeclareMathAlphabet{\mymathbb}{U}{BOONDOX-ds}{m}{n}

\DeclareSymbolFont{greeksymbols}{U}{pxmia}{m}{it}

\usepackage[colorlinks]{hyperref}
\usepackage{cleveref}

\newcommand\myshade{85}
\colorlet{mylinkcolor}{violet}
\colorlet{mycitecolor}{YellowOrange}
\colorlet{myurlcolor}{Aquamarine}

\hypersetup{
  linkcolor  = mylinkcolor!\myshade!black,
  citecolor  = mycitecolor!\myshade!black,
  urlcolor   = myurlcolor!\myshade!black,
  colorlinks = true,
}

\newcommand{\C}{\mathbb{C}}
\newcommand{\G}{\mathcal{G}}

\newcommand{\Zero}{\mymathbb{0}}

\newcommand{\N}{\mathbb{N}}

\newcommand{\X}{\mathcal{X}}

\newcommand{\Win}[3]{\mathit{Win}(#3)}
\newcommand{\NWin}[3]{\mathit{Win}_{\mathsf{norm}}(#3)}

\renewcommand{\phi}{\varphi}

\newcommand{\conff}{\conf}
\newcommand{\confp}{\conf'}
\newcommand{\confpp}{\conf''}

\newcommand{\stratf}{f}

\newcommand{\stratpp}{f''}

\newcommand{\updp}{\upd'}
\newcommand{\updpp}{\upd''}

\newcommand{\hatconf}{D}
\newcommand{\hatconff}{D}
\newcommand{\hatconfp}{D'}
\newcommand{\hatconfpp}{D''}

\newcommand{\hatstrat}{g}
\newcommand{\hatstratf}{g}

\newcommand{\hatstratpp}{g''}

\newcommand{\hatupd}{\eta}

\newcommand{\hatupdp}{\eta'}
\newcommand{\hatupdpp}{\eta''}

\newcommand{\stratnorm}{\strat_N}
\newcommand{\strataux}{\strat_\textup{aux}}

\newcommand{\df}{=}

\usepackage{ifpdf}
\ifpdf
\usepackage{tikz}
\usetikzlibrary{arrows,automata}
\usetikzlibrary{arrows,automata,positioning,calc}
\usetikzlibrary{shapes,snakes}
\usepackage[colorinlistoftodos,prependcaption]{todonotes}
\fi

\newcommand{\exend}{\hfill \ensuremath{\lhd}}

\newcommand{\ttrue}{\mathit{true}}
\newcommand{\ffalse}{\mathit{false}}
\newcommand{\ftrans}[3]{\llbracket {#1} \rrbracket_{#2,#3}}

\newcommand{\bound}{B}
\newcommand{\mem}{\text{mem}}

\newcommand{\tvec}{\vec{k}}

\newcommand{\locAcc}{\mathfrak{C}}

\newcommand{\game}{\mathcal{G}}
\newcommand{\Acc}{\mathcal{F}}

\newcommand{\Loc}{L}
\newcommand{\loc}{\ell}
\newcommand{\upd}{\tau}
\newcommand{\conf}{C}

\newcommand{\play}{\pi}

\newcommand{\AllConf}{\mathit{Conf}}

\newcommand{\Structure}{\mathcal{S}}
\newcommand{\Updates}{T}
\newcommand{\sysUpdates}{\Updates_\sys}
\newcommand{\envUpdates}{\Updates_\env}
\newcommand{\cutoff}{\vec{k_0}}

\newcommand{\outm}[1]{\mathit{out}_{#1}}
\newcommand{\inm}[1]{\mathit{in}_{#1}}
\newcommand{\Plays}{\mathit{Plays}}

\newcommand{\cPlays}[1]{\Plays_{#1}}

\newcommand{\dist}{d}
\newcommand{\ConstF}{K}
\newcommand{\maxdist}{\mathit{Max}}

\newcommand{\dec}[2]{#1[#2\scalebox{0.8}{\textup{--\,--}}]}
\newcommand{\inc}[2]{#1[#2\scalebox{0.8}{\textup{++}]}}

\newcommand{\locsetgezero}[1]{\left(#1\right) \ge 0}
\newcommand{\wlocsetgezero}[1]{#1 \ge 0}

\newcommand{\tcm}{M}
\newcommand{\tcmQ}{Q}
\newcommand{\tcmq}{q}
\newcommand{\tcminit}{\tcmq_0}
\newcommand{\tcmF}{F}
\newcommand{\tcmfinal}{\tcmq_h}
\newcommand{\tcmcounter}{\mathsf{c}}
\newcommand{\tcmT}{\Delta}
\newcommand{\tcmt}{t}
\newcommand{\Locok}{\Loc_\checkmark}

\newcommand{\loclet}[1]{\llangle#1\rrangle}
\newcommand{\accCond}[1]{\Acc_{#1}}
\newcommand{\accConde}[1]{\Acc_{#1}^{\,\env}}

\newcommand{\Alpha}{A}
\newcommand{\Events}{\Sigma}
\newcommand{\sEvents}{\Sigma_\sys}
\newcommand{\eEvents}{\Sigma_\env}

\newcommand{\Procs}{\mathbb{P}}

\newcommand{\DS}{\Events}

\newcommand{\strat}{f}

\newcommand{\fo}{\textup{FO}}
\newcommand{\fod}{\textup{FO}^2}
\renewcommand{\succ}{\text{+1}}
\newcommand{\dataeq}{\sim}
\newcommand{\dataord}{\lesssim}
\newcommand{\dataordinv}{\gtrsim}

\newcommand{\env}{\mathsf{e}}
\newcommand{\sys}{\mathsf{s}}
\newcommand{\both}{{\sys\env}}

\newcommand{\nproc}{N}

\newcommand{\Pos}[1]{\mathit{Pos}(#1)}

\newcommand{\synNpos}{\mathbb{N}}

\newcommand{\proc}{p}

\newcommand{\wordLetter}[2]{#1[#2]}

\newcommand{\event}{\sigma}
\newcommand{\ev}{\sigma}

\newcommand{\sProcs}{\mathbb{P}_{\!\sys}}
\newcommand{\eProcs}{\mathbb{P}_{\!\env}}

\newcommand{\seProcs}{\mathbb{P}_{\!\both}}

\newcommand{\sN}{\mathcal{N}_\sys}
\newcommand{\eN}{\mathcal{N}_\env}
\newcommand{\seN}{\mathcal{N}_\both}

\newcommand{\Types}{\mathbb{T}}

\makeatletter
\newcommand{\xRightarrow}[2][]{\ext@arrow 0359\Rightarrowfill@{#1}{#2}}
\makeatother

\makeatletter
\newcommand{\problemtitle}[1]{\gdef\@problemtitle{#1}}
\newcommand{\probleminput}[1]{\gdef\@probleminput{#1}}
\newcommand{\problemquestion}[1]{\gdef\@problemquestion{#1}}
\NewEnviron{decproblem}{
  \problemtitle{}\probleminput{}\problemquestion{}
  \BODY
  \par\addvspace{.2\baselineskip}
  \noindent
  \fbox{\centering
  \begin{tabular}{ll}
    \multicolumn{2}{l}{\normalsize \@problemtitle} \\
    \hline
    \rule{0pt}{3ex}
    \normalsize \textup{\textbf{Input:}} & \normalsize \@probleminput \\[0.5ex]
    \normalsize \textup{\textbf{Question:}} & \normalsize \@problemquestion
  \end{tabular}}
  \par\addvspace{.2\baselineskip}
}

\newif\iflong
\longfalse

\begin{document}

\title{Parameterized Synthesis for\\Fragments of First-Order Logic over Data Words\thanks{Partly supported by ANR FREDDA (ANR-17-CE40-0013).}}

\author{B{\'e}atrice B{\'e}rard\inst{1} \and Benedikt Bollig\inst{2} \and Mathieu Lehaut\inst{1} \and Nathalie Sznajder\inst{1}}

\institute{
Sorbonne Universit{\'e}, CNRS, LIP6, F-75005 Paris, France
\and
CNRS, LSV \& ENS Paris-Saclay, Universit{\'e} Paris-Saclay, France
}

\maketitle

\begin{abstract}
We study the synthesis problem for systems with a parameterized number of processes. As in the classical case due to Church, the system selects actions depending on the program run so far, with the aim of fulfilling a given specification. The difficulty is that, at the same time, the environment executes actions that the system cannot control. In contrast to the case of fixed, finite alphabets, here we consider the case of parameterized alphabets. An alphabet reflects the number of processes that are static but unknown. The synthesis problem then asks whether there is a finite number of processes for which the system can satisfy the specification. This variant is already undecidable for very limited logics. Therefore, we consider a first-order logic without the order on word positions. We show that even in this restricted case synthesis is undecidable if both the system and the environment have access to all processes. On the other hand, we prove that the problem is decidable if the environment only has access to a bounded number of processes. In that case, there is even a cutoff meaning that it is enough to examine a bounded number of process architectures to solve the synthesis problem.
\end{abstract}

\section{Introduction}

Synthesis deals with the problem of automatically generating a program that satisfies a given specification. The problem goes back to Church \cite{sisl1957-Chu}, who formulated it as follows: The environment and the system alternately select an input symbol and an output symbol from a finite alphabet, respectively, and in this way generate an infinite sequence. The question now is whether the system has a \emph{winning strategy}, which guarantees that the resulting infinite run is contained in a given ($\omega$)-regular language representing the specification, no matter how the environment behaves. This problem is decidable and very well understood \cite{TAMS138-BL,AIOCT1972-Rab}, and it has been extended in several different ways (e.g., \cite{JacobsTZ18,HornTW015,PnueliR90,JenkinsORW11,VelnerR11}).

In this paper, we consider a variant of the synthesis problem that allows us to model programs with a variable number of processes. As we then deal with an unbounded number of process identifiers, a fixed finite alphabet is not suitable anymore. It is more appropriate to use an infinite alphabet, in which every letter contains a process identifier and a program action. One can distinguish two cases here. In \cite{FigueiraP18}, a potentially infinite number of data values are involved in an infinite program run (e.g. by dynamic process generation). In a \emph{parameterized} system \cite{Esparza14,Bloem:2015}, on the other hand, one has an unknown but \emph{static} number of processes so that, along each run, the number of processes is finite. In this paper, we are interested in the latter, i.e., parameterized case. Parameterized programs are ubiquitous and occur, e.g., in distributed algorithms, ad-hoc networks, telecommunication protocols, cache-coherence protocols, swarms robotics, and biological systems. The synthesis question asks whether the system has a winning strategy for some number of processes (existential version) or no matter how many processes there are (universal version).

Over infinite alphabets, there are a variety of different specification languages (e.g., \cite{Bojanczy06,DL-tocl08,Kaminski1994,DemriDG12,LibkinTV15,FrenkelGS19,SchroderKMW17}). Unlike in the case of finite alphabets, there is no canonical definition of regular languages. In fact, the synthesis problem has been studied for N-memory automata \cite{BrutschT16}, the Logic of Repeating Values \cite{FigueiraP18}, and register automata \cite{khalimov_et_al:concur:2019,exibard_et_al:concur:2019,KhalimovMB18}. Though there is no agreement on a ``regular'' automata model, first-order (FO) logic over data words can be considered as a canonical logic, and this is the specification language we consider here. In addition to classical FO logic on words over finite alphabets, it provides a predicate $x \sim y$ to express that two events $x$ and $y$ are triggered by the same process. Its two-variable fragment FO$^2$ has a decidable emptiness and universality problem \cite{Bojanczy06} and is, therefore, a promising candidate for the synthesis problem.

Previous generalizations of Church's synthesis problem to infinite alphabets were generally \emph{synchronous} in the sense that the system and the environment perform their actions in strictly alternating order. This assumption was made, e.g., in the above-mentioned recent papers \cite{BrutschT16,FigueiraP18,khalimov_et_al:concur:2019,exibard_et_al:concur:2019,KhalimovMB18}. If there are several processes, however, it is realistic to relax this condition, which leads us to an \emph{asynchronous} setting in which the system has no influence on when the environment acts. Like in \cite{GS-tocl12}, where the asynchronous case for a fixed number of processes was considered, we only make the reasonable fairness assumption that the system is not blocked forever.

In summary, the synthesis problem over infinite alphabets can be classified as $(i)$ parameterized vs. dynamic, $(ii)$ synchronous vs. asynchronous, and $(iii)$ according to the specification language (register automata, Logic of Repeating Values, FO logic, etc.). As explained above, we consider here the \emph{parameterized asynchronous case for specifications written in FO logic}. To the best of our knowledge, this combination has not been considered before. For flexible modeling, we also distinguish between three types of processes: those that can only be controlled by the system; those that can only be controlled by the environment; and finally those that can be triggered by both. A partition into system and environment processes is also made in \cite{FinkbeinerO17,beutner_et_al:concur:2019}, but for a fixed number of processes and in the presence of an arena in terms of a Petri net.

Let us briefly describe our results. We show that the general case of the synthesis problem is undecidable for FO$^2$ logic. This follows from an adaptation of an undecidability result from \cite{figueira2018playing,FigueiraP18} for a fragment of the Logic of Repeating Values \cite{DemriDG12}. We therefore concentrate on an orthogonal logic, namely FO without the order on the word positions. First, we show that this logic can essentially count processes and actions of a given process up to some threshold. Though it has limited expressive power (albeit orthogonal to that of FO$^2$), it leads to intricate behaviors in the presence of an uncontrollable environment. In fact, we show that the synthesis problem is still undecidable. Due to the lack of the order relation, the proof requires a subtle reduction from the reachability problem in 2-counter Minsky machines. However, it turns out that the synthesis problem is decidable if the number of processes that are controllable by the environment is bounded, while the number of system processes remains unbounded. In this case, there is even a cutoff $k$, an important measure for parameterized systems (cf.\ \cite{Bloem:2015} for an overview): If the system has a winning strategy for $k$ processes, then it has one for any number of processes greater than $k$, and the same applies to the environment. The proofs of both main results rely on a reduction of the synthesis problem to \emph{parameterized vector games}, certain turn-based games in which, similar to Petri nets, tokens corresponding to the processes are moved around between states.

The paper is structured as follows. In Section~\ref{sec;preliminaries}, we define FO logic (especially FO without word order), and in Section~\ref{sec:synthesis}, we present the parameterized synthesis problem. In Section~\ref{sec:games}, we transform a given formula into a normal form and finally into a parameterized vector game. Based on this reduction, we investigate cutoff properties and show our (un)decidability results in Section~\ref{sec:results}. We conclude in Section~\ref{sec:conclusion}. Missing proof details are available in the appendix.


\section{Preliminaries}\label{sec;preliminaries}

For a finite or infinite alphabet $\Events$, let $\Sigma^\ast$ and
$\Sigma^\omega$ denote the sets of finite and, respectively, infinite words over $\Sigma$.
The empty word is $\varepsilon$.
Given $w \in \Events^\ast \cup \Events^\omega$,
let $|w|$ denote the length of $w$ and $\Pos{w}$ its set of positions:
$|w| = n$ and $\Pos{w} = \{1,\ldots,n\}$ if $w = \event_1\event_2 \ldots \event_n \in \Events^\ast$, and
$|w| = \omega$ and $\Pos{w} = \{1,2,\ldots\}$ if $w \in \Events^\omega$.
Let $\wordLetter{w}{i}$ be the $i$-th letter of $w$ for all $i \in \Pos{w}$.

\paragraph{\textup{\textbf{Executions.}}}

We consider programs involving a finite (but not fixed) number of processes.
Processes are controlled by antagonistic protagonists, System and Environment.
Accordingly, each process has a \emph{type} among $\Types = \{\sys,\env,\both\}$,
and we let $\sProcs$, $\eProcs$, and $\seProcs$ denote the
pairwise disjoint finite sets of processes controlled by System,
by Environment, and by both System and Environment, respectively.
We let $\Procs$ denote the triple $(\sProcs,\eProcs,\seProcs)$.
Abusing notation, we sometimes refer to $\Procs$ as the disjoint
union $\sProcs \cup \eProcs \cup \seProcs$.

For a set $S$, vectors $s \in S^\Types$ are usually referred to as triples
$s = (s_\sys, s_\env, s_\both)$. Moreover, for $s,s' \in \N^\Types$, we write
$s \le s'$ if $s_\theta \le s'_\theta$ for all $\theta \in \Types$.
Finally, let $s + s' = (s_\sys + s_\sys',\, s_\env+s_\env',\, s_\both+s_\both')$.

Processes can execute actions from a finite alphabet $\Alpha$.
Whenever an action is executed, we would like to know whether
it was triggered by System or by Environment. Therefore,
$\Alpha$ is partitioned into
$\Alpha = \Alpha_\sys \uplus \Alpha_\env$.
Let $\sEvents = \Alpha_\sys \times (\sProcs \cup \seProcs)$ and
$\eEvents = \Alpha_\env \times (\eProcs \cup \seProcs)$.
Their union $\Events = \sEvents \cup \eEvents$ is the set of \emph{events}.
A word $w \in \Events^\ast \cup \Events^\omega$ is called a $\Procs$-\emph{execution}.

\paragraph{\textup{\textbf{Logic.}}}

\newcommand{\FO}{\fo_{\!\Alpha}}
\newcommand{\FOfull}{\fo_{\!\Alpha}[\sim,<,+1]}
\newcommand{\FOtwo}{\fo_{\!\Alpha}^2[\sim,<,+1]}
\newcommand{\FOtwoorder}{\fo_{\!\Alpha}^2[\sim,<]}
\newcommand{\FOdata}{\fo_{\!\Alpha}[\sim]}
\newcommand{\FOtwodata}{\fo_{\!\Alpha}^2[\sim]}

\newcommand{\FOtwoR}{\fo_{\!\Alpha}^2[R]}
\newcommand{\mFOtwoR}{\fo_{\!\Alpha}_-^2[R]}

\newcommand{\mFOtwo}{\fo_-^2[<,+1]}
\newcommand{\mFOtwoorder}{\fo_-^2[<]}
\newcommand{\mFOdata}{\fo[\,]}
\newcommand{\mFOtwodata}{\fo_-^2[\,]}

\newcommand{\wFO}{\fo}
\newcommand{\wFOfull}{\fo[\sim,<,+1]}
\newcommand{\wFOtwo}{\fo^2[\sim,<,+1]}
\newcommand{\wFOtwoorder}{\fo^2[\sim,<]}
\newcommand{\wFOdata}{\fo[\sim]}
\newcommand{\wFOtwodata}{\fo^2[\sim]}

\newcommand{\wFOtwoR}{\fo^2[R]}

\newcommand{\ProcVar}{\mathcal{V}_\mathsf{proc}}
\newcommand{\PosVar}{\mathcal{V}}
\newcommand{\pv}{p}
\newcommand{\pva}{p}
\newcommand{\pvb}{q}
\newcommand{\modelspi}[4]{(#1,#2)\models_#3 #4}
\newcommand{\modelsp}[3]{(#1,#2)\models #3}
\newcommand{\succrel}[2]{{+}1(#1,#2)}
\newcommand{\succrelwp}{{+}1}
\newcommand{\Int}{\mathcal{I}}
\newcommand{\Free}{\mathit{Free}}
\newcommand{\procform}{\mathit{proc}}
\newcommand{\UR}[1]{R_{#1}}

\newcommand{\Inter}{I}
\newcommand{\dom}{\textup{dom}}
\newcommand{\subf}{\xi}

Formulas of our logic are evaluated over
$\Procs$-executions. We fix an infinite supply
$\PosVar = \{x,y,z,\ldots\}$
of variables, which are interpreted as processes
from $\Procs$ or positions of the execution.
The logic $\FOfull$ is given by the grammar
\begin{align*}
\varphi ~::=~ \theta(x)~|~a(x)~|~ x = y~|~x \sim y~|~x < y~|~\succrel{x}{y}~|~\neg \varphi~|~\varphi \lor \varphi~|~\exists x. \varphi
 \end{align*}
where $x, y \in \PosVar$, $\theta \in \Types$, and $a \in \Alpha$.
Conjunction ($\land$), universal quantification ($\forall$),
implication (${\Longrightarrow}$), $\ttrue$, and $\ffalse$ are obtained as
abbreviations as usual.

Let $\phi \in \FOfull$. By $\Free(\phi) \subseteq \PosVar$,
we denote the set of variables that occur free in $\phi$.
If $\Free(\phi) = \emptyset$, then we call $\phi$ a \emph{sentence}.
We sometimes write $\phi(x_1,\ldots,x_n)$ to emphasize the fact
that
$\Free(\phi) \subseteq \{x_1,\ldots,x_n\}$.

To evaluate $\varphi$ over a $\Procs$-execution $w = (a_1,p_1)(a_2,p_2) \ldots$,
we consider $(\Procs,w)$ as a structure
$\Structure_{(\Procs,w)} = (\Procs \mathrel{\uplus} \Pos{w},\sProcs, \eProcs, \seProcs,(\UR{a})_{a \in \Alpha},\sim,<,\succrelwp)$
where
$\Procs \uplus \Pos{w}$ is the universe,
$\sProcs$ $\eProcs$, and $\seProcs$ are interpreted as unary relations,
$\UR{a}$ is the unary relation $\{i \in \Pos{w} \mid a_i = a\}$,
${<} = \{(i,j) \in \Pos{w} \times \Pos{w} \mid i < j\}$,
${\succrelwp} = \{(i,i+1) \mid 1 \le i < |w|\}$, and
${\sim}$ is the least equivalence relation over $\Procs \mathrel{\uplus} \Pos{w}$ containing
\begin{itemize}[topsep=0.6ex]
\item $(p,i)$ for all $p \in \Procs$ and $i \in \Pos{w}$ such that $p = p_i$, and
\item $(i,j)$ for all $(i,j) \in \Pos{w} \times \Pos{w}$ such that $p_i = p_j$.
\end{itemize}
An equivalence class of $\sim$ is often simply referred to as a \emph{class}.

\begin{example}\label{ex:execution}
Suppose $\Alpha_\sys = \{a,b\}$ and $\Alpha_\env = \{c,d\}$.
Let the set of processes $\Procs$ be given by $\sProcs = \{1,2,3\}$, $\eProcs = \{4,5\}$, and $\seProcs = \{6,7,8\}$.
Moreover, let \[w = (a,1)(b,8)(d,7)(c,4)(a,6)(c,6)(a,7)(d,6)(b,2)(d,7)(a,7) \in \Events^\ast\,.\]
Figure \ref{fig:execution} illustrates $\Structure_{(\Procs,w)}$. The edge relation represents $\succrelwp$, its transitive closure is $<$.
\exend
\end{example}

\begin{figure}[t]
\centering
\includegraphics[width=0.92\textwidth]{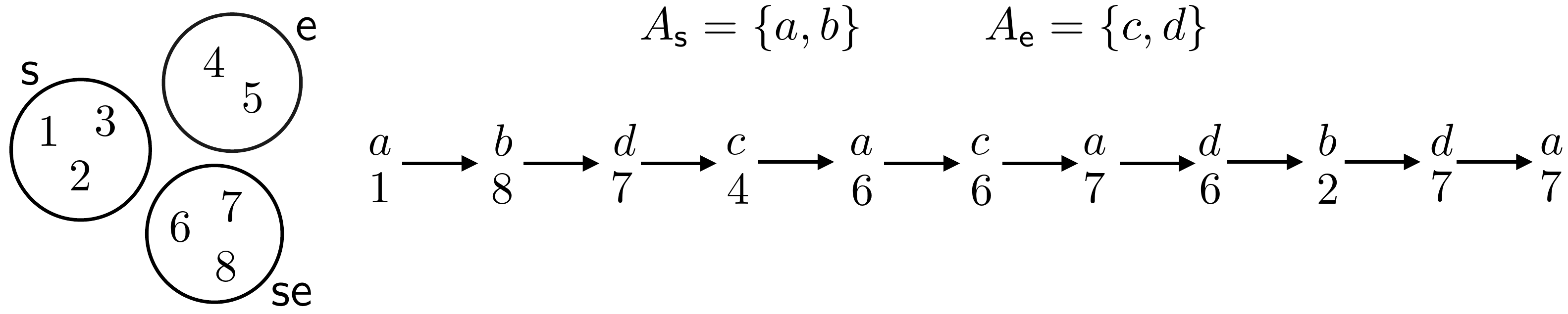}
\caption{Representation of $\Procs$-execution as a mathematical structure\label{fig:execution}}
\end{figure}

\newcommand{\cbound}{m}

An \emph{interpretation} for $(\Procs,w)$ is a partial mapping
$\Inter: \PosVar \to \Procs \cup \Pos{w}$.
Suppose $\phi \in \FOfull$ such that $\Free(\phi) \subseteq \dom(\Inter)$.
The satisfaction relation
$(\Procs,w),\Inter \models \phi$ is then defined as expected,
based on the structure $\Structure_{(\Procs,w)}$ and
interpreting free variables according to $\Inter$.
For example, let $w = (a_1,p_1)(a_2,p_2) \ldots$ and $i \in \Pos{w}$.
Then, for $\Inter(x) = i$, we have $(\Procs,w),\Inter \models a(x)$ if $a_i = a$.

We identify some fragments of $\FOfull$.
For $R \subseteq \{\sim,<,{+}1\}$, let
$\FO[R]$ denote the set of formulas that do not use symbols in $\{\sim,<,{+}1\} \setminus R$.
Moreover, $\FOtwoR$ denotes the fragment of $\FO[R]$ that uses only two (reusable) variables.

Let $\phi(x_1,\ldots,x_n,y) \in \FOfull$ and $\cbound \in \N$.
We use $\exists^{\ge \cbound}y.
\phi(x_1,\ldots,x_n,y)$ as an abbreviation for
\[
\exists y_1 \ldots \exists y_\cbound.
\bigwedge_{1 \le i < j \le \cbound} \neg{(y_i = y_j)} \wedge \bigwedge_{1 \le i \le \cbound} \phi(x_1,\ldots,x_n,y_i)\,,
\]
if $\cbound > 0$, and $\exists^{\ge 0}y. \phi(x_1, \ldots, x_n, y) = \ttrue$.
Thus, $\exists^{\ge \cbound}y.\phi$ says that there are at least $\cbound$
distinct elements that verify $\phi$. 
We also use
$\exists^{= \cbound} y. \phi$ as an abbreviation for
$\exists^{\ge \cbound} y. \phi \wedge \neg \exists^{\ge \cbound+1} y. \phi$.
Note that $\phi \in \FO[R]$ implies that $\exists^{\ge \cbound} y. \phi \in \FO[R]$
and $\exists^{= \cbound} y. \phi \in \FO[R]$.

\begin{example}\label{ex:formulas}
Let $\Alpha$, $\Procs$, and $w$ be like in Example~\ref{ex:execution} and Figure~\ref{fig:execution}.
\begin{itemize}\itemsep=0.5ex
\item $\varphi_1 = \forall x.\bigl((\sys(x) \vee \both(x)) \implies \exists y. (x \sim y \wedge (a(y) \vee b(y)))\bigr)$ says that each process that System can control executes at least one system action. We have $\phi_1 \in \FOtwodata$ and $(\Procs,w) \not\models \varphi_1$, as process $3$ is idle.

\item $\varphi_2 = \forall x.\bigl(d(x) \implies \exists y.(x \sim y \wedge a(y))\bigr)$ says that, for every $d$, there is an $a$ on the same process. We have $\phi_2 \in \FOtwodata$ and $(\Procs,w) \models \varphi_2$.

\item $\varphi_3 = \forall x.\bigl(d(x) \implies \exists y.(x \sim y \wedge x < y \wedge a(y))\bigr)$ says that every $d$ is \emph{eventually} followed by an $a$ executed by the same process. We have $\varphi_3 \in \FOtwoorder$ and $(\Procs,w) \not\models \varphi_3$: The event $(d,6)$ is not followed by some $(a,6)$.

\item $\varphi_4 = \forall x. \bigl(\bigl(\exists^{= 2} y.(x \sim y \wedge a(y))\bigr)
\Longleftrightarrow \bigl(\exists^{= 2} y.(x \sim y \wedge d(y))\bigr)\bigr)$ says that
each class contains exactly two occurrences of $a$
iff it contains exactly two occurrences of $d$.
Moreover, $\varphi_4 \in \FOdata$ and $(\Procs,w) \models \varphi_4$.
Note that $\varphi_4 \not\in \FOtwodata$, as
$\exists^{= 2} y$ requires the use of three different variable names.
\exend
\end{itemize}
\end{example}

\newcommand{\Partitions}[1]{\textup{Part}(#1)}
\newcommand{\Abstractions}[1]{\mathfrak{C}_{\!#1}}
\newcommand{\Class}{\mathcal{C}}
\newcommand{\eclass}{E}
\newcommand{\class}{X}
\newcommand{\Part}{P}
\newcommand{\vect}{\mathbf{v}}
\newcommand{\pequiv}[1]{E_{#1}}
\renewcommand{\phi}{\varphi}
\newcommand{\ctype}[1]{\psi_{#1}}

\section{Parameterized Synthesis Problem}\label{sec:synthesis}

We define an asynchronous synthesis problem.
A $\Procs$-\emph{strategy} (for System) is a mapping $\strat: \Events^\ast \to \sEvents \cup \{\varepsilon\}$.
A $\Procs$-execution $w = \ev_1\ev_2\ldots \in \DS^\ast \cup \DS^\omega$ is $\strat$-\emph{compatible} if, for all $i \in \Pos{w}$ such that $\ev_i \in \DS_\sys$, we have $f(\ev_1 \ldots \ev_{i-1}) = \ev_i$.
We call $w$ $\strat$-\emph{fair} if the following hold:
$(i)$ If $w$ is finite, then $f(w) = \varepsilon$, and
$(ii)$ if $w$ is infinite and $f(\ev_1 \ldots \ev_{i-1}) \neq \varepsilon$
for infinitely many $i \ge 1$, then $\ev_j \in \DS_\sys$ for infinitely many $j \ge 1$.

Let $\varphi \in \FOfull$ be a sentence. We say that $f$ is $\Procs$-\emph{winning} for $\varphi$
if, for every $\Procs$-execution $w$ that is $\strat$-compatible and $\strat$-fair, we have
$\modelsp{\Procs}{w}{\varphi}$.

The existence of a $\Procs$-strategy that is $\Procs$-winning for a given formula
does not depend on the concrete process identities but only on
the cardinality of the sets $\sProcs$, $\eProcs$, and $\seProcs$.
This motivates the following definition of winning triples for a formula.
Given $\varphi$, let
$\Win{\Alpha_\sys}{\Alpha_\env}{\varphi}$ be the set of
triples $(\pconsts,\pconsta,\pconstb) \in \N^\Types$
for which there is $\Procs = (\sProcs,\eProcs,\seProcs)$
such that $|\Procs_\theta| = \pconst_\theta$ for all $\theta \in \Types$
and there is a $\Procs$-strategy that is $\Procs$-winning for $\varphi$.

Let $\Zero = \{0\}$ and $\pconsta,\pconstb \in \N$.
In this paper, we focus on the intersection of $\Win{}{}{\varphi}$ with the sets
$\synNpos \times \Zero \times \Zero$ (which corresponds to the usual satisfiability problem);
$\synNpos \times \{\pconsta\} \times \{\pconstb\}$ (there is a constant number of environment and mixed processes);
$\synNpos \times \synNpos \times \{\pconstse\}$ (there is a constant number of mixed processes);
$\Zero \times \Zero \times \synNpos$ (each process is controlled by both System and Environment).

\begin{definition}[synthesis problem]\label{def:synthesis}
For fixed $\FClass \in \{\textup{FO},\textup{FO}^2\}$, $R \subseteq \{\sim,<,{+}1\}$, and
sets $\sN,\eN,\seN \subseteq \N$, the (parameterized) synthesis problem is given as follows:
\begin{center}
\begin{decproblem}
  \problemtitle{$\Synthesis{\FClass[R]}{\sN}{\eN}{\seN}$}
  \probleminput{$\Alpha = \Alpha_\sys \uplus \Alpha_\env$ and a sentence $\varphi \in \FClass_A[R]$}
  \problemquestion{$\Win{}{}{\varphi} \cap (\sN \times \eN \times \seN) \neq \emptyset$\,\textup{?}}  
\end{decproblem}
\end{center}
The \emph{satisfiability problem} for $\FClass[R]$ is defined as $\Synthesis{\FClass[R]}{\N}{\Zero}{\Zero}$.
\end{definition}

\newcommand{\enqueue}[2]{#1 \odot #2}

\begin{example}\label{ex:synthesis}
Suppose $\Alpha_\sys = \{a,b\}$ and $\Alpha_\env = \{c,d\}$, and consider
the formulas $\phi_1$--$\phi_4$ from Example~\ref{ex:formulas}.

First, we have $\Win{}{}{\varphi_1} = \N^\Types$. Given an arbitrary
$\Procs$ and any total order ${\sqsubseteq}$ over $\sProcs \cup
\seProcs$, a possible $\Procs$-strategy $\strat$ that is
$\Procs$-winning for $\varphi_1$ maps $w \in \Sigma^\ast$ to $(a,p)$
if $p$ is the smallest process from $\sProcs \cup \seProcs$ wrt.\
$\sqsubseteq$ that does not occur in $w$, and
that returns $\varepsilon$ for $w$ if all
processes from $\sProcs \cup \seProcs$ already occur in $w$.

For the three formulas $\phi_2$, $\phi_3$, and $\phi_4$, observe that,
since $d$ is an environment action, if there is at least one process
that is exclusively controlled by Environment, then there is no
winning strategy. Hence we must have $\eProcs = \emptyset$.  In fact,
this condition is sufficient in the three cases and the strategies
described below show that all three sets $\Win{}{}{\varphi_2}$,
$\Win{}{}{\varphi_3}$, and $\Win{}{}{\varphi_4}$ are equal to $\N
\times \Zero \times \N$.

\begin{itemize}\itemsep=0.5ex
\item For $\phi_2$, the very same strategy as for $\phi_1$ also works
  in this case, producing an $a$ for every process in $\sProcs \cup
  \seProcs$, whether there is a $d$ or not.
\item For $\phi_3$, a winning strategy $\strat$ will apply the
  previous mechanism iteratively, performing $(a,p)$ for $p \in
  \seProcs = \{p_0,\ldots,p_{n-1}\}$ over and over again: $f(w) =
  (a,p_i)$ where $i$ is the number of occurrences of letters from
  $\Sigma_\sys$ modulo $n$.  By the fairness assumption, this
  guarantees satisfaction of $\phi_3$.  A more ``economical'' winning
  strategy $\strat$ may organize pending requests in terms of $d$ in a
  queue and acknowledge them successively. More precisely, given $u
  \in \Procs^\ast$ and $\sigma \in \Sigma$, we define another word
  $\enqueue{u}{\sigma} \in \Procs^\ast$ by $\enqueue{u}{(d,p)} = u
  \cdot p$ (inserting $p$ in the queue) and $\enqueue{(p \cdot
    u)}{(a,p)} = u$ (deleting it).  In all other cases,
  $\enqueue{u}{\event} = u$.  Let $w = \event_1 \ldots \event_n \in
  \Sigma^\ast$, with queue
  $\enqueue{(\enqueue{(\enqueue{\varepsilon}{\event_1})}{\event_2}
    \ldots)}{\event_n} = p_1 \ldots p_k$.  We let $\strat(w) = \varepsilon$
    if $k = 0$, and $\strat(w) = (a,p_1)$ if $k \ge 1$.
\item For $\phi_4$, the same strategy as for $\phi_3$ ensures that
  every $d$ has a \emph{corresponding} $a$ so that, in the long run,
  there are as many $a$'s as $d$'s in every class.
\exend
\end{itemize}
\end{example}

Another interesting question is whether System (or Environment) has
a winning strategy as soon as the number of processes is big enough.
This leads to the notion of a cutoff (cf.\ \cite{Bloem:2015} for an overview):
Let $\sN,\eN,\seN \subseteq \N$ and $W \subseteq \N^\Types$. 
We call $\cutoff \in \N^\Types$ a \emph{cutoff} of $W$
wrt.\ $(\sN,\eN,\seN)$ if $\cutoff \in \sN \times \eN \times \seN$ and either
\begin{itemize}\itemsep=0.5ex
\item for all $\tvec \in \sN \times \eN \times \seN$ such that $\tvec \ge \cutoff$, we have $\tvec \in W$, or
\item for all $\tvec \in \sN \times \eN \times \seN$ such that $\tvec \ge \cutoff$, we have  $\tvec \not\in W$.
\end{itemize}

Let $\FClass \in \{\textup{FO},\textup{FO}^2\}$ and $R \subseteq \{\sim,<,{+}1\}$.
If, for every alphabet $\Alpha = \Alpha_\sys \uplus \Alpha_\env$ and every sentence
$\varphi \in \FClass_A[R]$, the set $\Win{}{}{\varphi}$ has a computable cutoff wrt.\ $(\sN,\eN,\seN)$,
then we know that $\Synthesis{\FClass[R]}{\sN}{\eN}{\seN}$ is decidable, as it can be reduced to
a finite number of simple synthesis problems over a finite alphabet.
This is how we will show decidability of
$\Synthesis{\wFOdata}{\N}{\{\pconsta\}}{\{\pconstb\}}$ for all $\pconsta, \pconstb \in \N$.

\medskip

Our contributions are summarized in Table~\ref{table:summary}.
Note that known satisfiability results for data logic
apply to our logic, as processes can be simulated by treating
every $\theta \in \Types$ as an ordinary letter.
Let us first state undecidability of the general synthesis problem,
which motivates the study of other FO fragments.

\newcommand{\status}[2]{\begin{tabular}{c}#1\end{tabular}}
\newcommand{\qstatus}[2]{\begin{tabular}{c}\textcolor{orange}{#1}\\\textcolor{orange}{#2}\end{tabular}}
\newcommand{\cwidth}[1]{\parbox[b]{6em}{\centering #1}}
\newcommand{\ourresult}[1]{\textcolor{blue}{#1}}

\begin{table}[t]
\caption{Summary of results. Our contributions are highlighted in \ourresult{blue}.\label{table:summary}}
\centering
\begin{tabular}{lcccc}
\toprule
Synthesis~~~~~~~~~ & \cwidth{$(\synNpos, \Zero, \Zero)$} & \cwidth{$(\synNpos,\{\pconsta\},\{\pconstb\})$}  & \cwidth{$(\synNpos, \synNpos, \Zero)$} & \cwidth{$(\Zero, \Zero, \synNpos)$}\\
 \midrule
$\wFOtwo$ & \status{decidable \cite{Bojanczy06}}{no cutoff} & \status{?}{\textcolor{gray}{no cutoff}} & \status{?}{\textcolor{gray}{no cutoff}} & \status{\ourresult{undecidable}}{\textcolor{gray}{no cutoff}} \\
 \midrule
$\wFOtwoorder$ & \status{NEXPTIME-c. \cite{Bojanczy06}}{cutoff} & \status{?}{?} & \status{?}{?} & \status{?}{?}\\
 \midrule
$\wFOdata$ & \status{\ourresult{decidable}}{\ourresult{cutoff}}  & \status{\ourresult{decidable}}{\ourresult{cutoff}} & \status{?$^*$\!\!\!}{\ourresult{no cutoff}} & \status{\ourresult{undecidable}}{\ourresult{no cutoff}}\\
 \bottomrule
\end{tabular}\\[0.5ex]
\scalebox{0.85}{$^*$We show, however, that there is no cutoff.}
\end{table}

\begin{theorem}\label{thm:undecFOtwo}
The problem $\Synthesis{\wFOtwo}{\Zero}{\Zero}{\N}$ is undecidable.
\end{theorem}

\newcommand{\Op}{\mathsf{Op}}
\newcommand{\op}{\mathsf{op}}
\newcommand{\incc}[1]{\tcmcounter_{#1}{\scalebox{0.8}{\textup{++}}}}
\newcommand{\decc}[1]{\tcmcounter_{#1}{\scalebox{0.8}{\textup{--\,--}}}}
\newcommand{\zeroc}[1]{\tcmcounter_{#1}{\scalebox{0.8}{==}}\hspace{0.06em}0}

\begin{proof}[sketch]
We adapt the proof from \cite{figueira2018playing,FigueiraP18} reducing the halting problem for 2-counter machines. 
We show that their encoding can be expressed in our logic, even if we restrict it to two variables, and can also be adapted to the
asynchronous setting. Details are given in Appendix~\ref{app:undecFOtwo}.
\qed
\end{proof}


\newcommand{\AccConf}{\AllConf_{\!\Acc}}
\newcommand{\RejConf}{\AllConf_{\!\neg \Acc}}

\section{$\wFOdata$ and Parameterized Vector Games}\label{sec:games}

Due to the undecidability result of Theorem~\ref{thm:undecFOtwo},
one has to switch to other fragments of first-order logic.
We will henceforth focus on the logic $\wFOdata$ and establish
some important properties, such as a normal form, that
will allow us to deduce a couple of results, both
positive and negative.

\subsection{Satisfiability and Normal Form for $\wFOdata$}

We first show that $\wFOdata$ logic essentially allows
one to count letters in a class up to some threshold,
and to count such classes up to some other threshold.
For $\bound \in \N$ and
$\loc \in \{0,\ldots,\bound\}^\Alpha$, we first
define an $\FOdata$-formula $\ctype{\bound,\loc}(y)$
verifying that, in the class defined by $y$, the number of occurrences of each letter
$a \in A$, counted up to $\bound$, is $\loc(a)$:
\[
\ctype{\bound,\loc}(y) ~=~
\bigwedge_{\substack{a \in \Alpha\\[0.3ex]\loc(a) < \bound}} \exists^{=\loc(a)} z.\bigl(y \sim z \wedge a(z)\bigr)
\wedge
\bigwedge_{\substack{a \in \Alpha\\[0.3ex]\loc(a) = \bound}} \exists^{\ge \loc(a)} z.\bigl(y \sim z \wedge a(z)\bigr)
\]

\begin{theorem}[normal form for $\boldsymbol{\wFOdata}$]\label{thm:normalform}
Let $\phi \in \FOdata$ be a sentence. There is $\bound \in \N$ such that
$\phi$ is effectively equivalent to a
disjunction of conjunctions of
formulas of the form
$\exists^{\bowtie \cbound} y. \bigl(\theta(y) \wedge \ctype{\bound,\loc}(y)\bigr)$
where ${\bowtie} \in \{{\ge}, {=}\}$, $\cbound \in \N$, $\theta \in \Types$,
and $\loc \in \{0,\ldots,\bound\}^\Alpha$.
\end{theorem}

\begin{proof}[sketch]
We use two known normal-form constructions for general FO logic. Due to Schwentick and Barthelmann \cite{schwentick1998local}, any $\FOdata$ formula is effectively
equivalent to a formula of the form $\exists x_1 \dots \exists x_n \forall y. \varphi(x_1, \dots, x_n,y)$ where, in $\phi(x_1,\ldots,x_n,y)$, quantification is always of the form
$\exists z. (z \sim y \mathrel{\wedge} \ldots)$ or $\forall z. (z \sim y \implies \ldots)$.
By guessing the exact relation between the variables $x_1,\ldots,x_n$, one can eliminate these ending up with formulas that only talk about the class of a given event $y$. Those formulas are then evaluated over multi-sets over the alphabet $\Types \cup \Alpha$. According to Hanf's theorem \cite{Hanf1965,BolligK12}, they are effectively equivalent to statements counting elements up to some threshold. This finally leads to the desired normal form.
The details can be found in Appendix~\ref{app:normalform}.
\qed
\end{proof}

\begin{example}\label{ex:normalform}
Recall the formula $\varphi_4 = \forall x. \bigl(\bigl(\exists^{= 2} y.(x \sim y \wedge a(y))\bigr) \Longleftrightarrow \bigl(\exists^{= 2} y.(x \sim y \wedge d(y))\bigr)\bigr) \in \FOdata$ from Example~\ref{ex:formulas}, over $\Alpha_\sys = \{a,b\}$ and $\Alpha_\env = \{c,d\}$.
An equivalent formula in normal form is
$
\varphi_4' = \bigwedge_{{\theta \in \Types\text{, }\loc \in Z}}
\exists^{= 0} y. \bigl(\theta(y) \wedge \ctype{3,\loc}(y)\bigr)
$
where $Z$ is the set of vectors $\loc \in \{0,\ldots,3\}^\Alpha$ such that
$\loc(a) = 2 \neq \loc(d)$ or
$\loc(d) = 2 \neq \loc(a)$.
The formula indeed says that there is no class with ${=}2$ occurrences of $a$ and ${\neq}2$ occurrences of $d$ or vice versa, which is equivalent to $\varphi_4$.
\exend
\end{example}

\begin{corollary}
The satisfiability problem for $\wFOdata$ is decidable.
Moreover, if an $\FOdata$ formula has an infinite model, then it also has a finite one.
\end{corollary}

\begin{proof}[idea]
Take a formula in normal form with its associated threshold $\bound$. A~formula of the form
$\varphi_{\theta,\loc}^{\bowtie \cbound} = \exists^{\bowtie \cbound} y. \bigl(\theta(y) \wedge \ctype{\bound,\loc}(y)\bigr)$
is satisfied by the execution
\[\prod_{a \in \Alpha} \prod_{j\in\{1,\ldots,\loc(a)\}} (a, \proc_1) \ldots (a, \proc_n)\]
where $\proc_1, \dots, \proc_n \in \Procs_\theta$ are pairwise distinct and $n \in \N$ is such that $n \bowtie \cbound$.
As long as there are no two inconsistent formulas for the same pair $(\theta, \loc)$ such as $\varphi_{\theta,\loc}^{=k_1} \land \varphi_{\theta,\loc}^{=k_2}$ with $k_1 \neq k_2$ or $\varphi_{\theta,\loc}^{=k_1} \land \varphi_{\theta,\loc}^{\ge k_2}$ with $k_1 < k_2$, any conjunction of such formulas can also be satisfied by concatenating one satisfying execution for each pair $(\theta, \loc)$, which gives a finite model.
\qed
\end{proof}

Note that satisfiability for $\wFOtwodata$ is already NEXPTIME-hard, which
even holds in the presence of unary relations only \cite{Furer83,GradelKV97}.
It is NEXPTIME-complete due to the upper bound for $\wFOtwoorder$ \cite{Bojanczy06}.


\newcommand{\bcounter}{\textup{BinaryCounter}}
\newcommand{\init}{\text{Init}}
\newcommand{\eqform}[4]{\textup{Eq}^{#1,#2}(#3,#4)}
\newcommand{\firstzero}{\textup{FirstZero}}
\newcommand{\successor}{\textup{Successor}}
\newcommand{\rec}{\textup{Rec}}
\newcommand{\unicity}{\textup{Unicity}}
\newcommand{\circuit}{\textup{Circuit}}
\newcommand{\gate}{\textup{Gate}}
\renewcommand{\output}{\textup{Output}}
\newcommand{\match}{\textup{MatchVal}}
\newcommand{\sat}{\textup{Sat}}

\newcommand{\succsat}{\textsc{Succinct-3-SAT}\xspace}
\newcommand{\NEXPTIME}{\textup{NEXPTIME}\xspace}
\newcommand{\literal}{l}
\newcommand{\propform}{\Psi}
\newcommand{\variable}{V}
\newcommand{\inp}{\mathit{in}}
\newcommand{\outp}{\mathit{out}}
\newcommand{\gatel}{\mathit{gate}}
\newcommand{\varl}{\mathit{var}}

\newcommand{\IAlpha}{\mathit{In}}
\newcommand{\GAlpha}{\mathit{Gates}}
\newcommand{\OAlpha}{\mathit{Out}}
\newcommand{\VAlpha}{\mathit{Var}}
\newcommand{\cntform}[1]{\bcounter[#1]}
\newcommand{\ngates}{r}

\subsection{From Synthesis to Parameterized Vector Games}

\newcommand{\accfunction}{\kappa}
\newcommand{\memG}{\mem}

Exploiting the normal form for $\FOdata$,
we now present a reduction of the synthesis problem
to a strictly turn-based two-player game.
This game is conceptually simpler and easier to reason about.
The reduction works in both directions, which will allow us
to derive both decidability and undecidability results.

Note that, given a formula $\varphi \in \FOdata$
(which we suppose to be in normal form with threshold $\bound$),
the order of letters in an
execution does not matter.
Thus, given some $\Procs$, a reasonable strategy for Environment would be to just ``wait and see''.
More precisely, it does not put Environment into a worse position if,
given the current execution $w \in \Events^\ast$, it
lets the System execute as many actions as it wants in terms of a word $u \in \sEvents^\ast$.
Due to the fairness assumption, System would be able to execute all the letters
from $u$ anyway. Environment can even require System to play a word $u$ such that $(\Procs,wu) \models \varphi$.
If System is not able to produce such a word, Environment can just sit back and do nothing.
Conversely, upon $wu$ satisfying $\varphi$,
Environment has to be able to come up with a word $v \in \eEvents^\ast$
such that $(\Procs,wuv) \not\models \varphi$. This leads to a turn-based
game in which System and Environment play in strictly alternate order and have to provide
a satisfying and, respectively, falsifying execution.

In a second step, we can get rid of process identifiers:
According to our normal form, all we are interested in is the \emph{number}
of processes that agree on their letters counted up to
threshold $\bound$. That is, a finite execution
can be abstracted as a \emph{configuration} $\conf: \Loc \to \N^\Types$
where $\Loc = \{0,\ldots,\bound\}^\Alpha$.
For $\loc \in \Loc$ and $\conf(\loc) = (n_\sys,n_\env,n_\both)$,
$n_\theta$ is the number of processes of
type $\theta$ whose letter count up to threshold $\bound$
corresponds to $\loc$.
We can also say that $\loc$ contains $n_\theta$ tokens
of type $\theta$. If it is System's turn, it will pick some pairs
$(\loc,\loc')$ and move some tokens of type $\theta \in \{\sys,\both\}$
from $\loc$ to $\loc'$, provided
$\loc(a) \le \loc'(a)$ for all $a \in \Alpha_\sys$
and $\loc(a) = \loc'(a)$ for all $a \in \Alpha_\env$.
This actually corresponds to adding more system letters in the corresponding
processes. The Environment proceeds analogously.

Finally, the formula $\phi$ naturally translates to an acceptance condition
$\Acc \subseteq \locAcc^L$ over configurations, where
$\locAcc$ is the set of \emph{local acceptance conditions}, which are of the form
$({\bowtie_\sys} n_\sys\,,\, {\bowtie_\env} n_\env\,,\, {\bowtie_\both} n_\both)$ where
${\bowtie_\sys}, {\bowtie_\env}, {\bowtie_\both} \in \{=,\ge\}$ and $n_\sys,n_\env,n_\both \in \N$.

We end up with a turn-based game in which,
similarly to a VASS game \cite{BrazdilJK10,Jancar2015,AbdullaMSS13,RaskinSB05,CourtoisS14},
System and Environment move tokens along vectors from $\Loc$. Note that, however, our games
have a very particular structure so that undecidability for VASS games does not carry over
to our setting.
Moreover, existing decidability results do not allow us to infer our cutoff results below.

In the following, we will formalize \emph{parameterized vector games}.

\begin{definition}\label{def:games}
A \emph{parameterized vector game} (or simply \emph{game})
is given by a triple $\game=(\Alpha,\bound,\Acc)$ where
$\Alpha = \Alpha_\sys \uplus \Alpha_\env$ is the finite alphabet,
$\bound \in \N$ is a bound, and, letting $\Loc = \{0,\ldots,\bound\}^\Alpha$,
$\Acc \subseteq \locAcc^L$ is a finite set called \emph{acceptance condition}.
\end{definition}

\myparagraph{Locations}
Let $\loc_0$ be the location such that $\loc_0(a) = 0$ for all $a \in \Alpha$.
For $\loc \in \Loc$ and $a\in \Alpha$, we define $\loc + a$ by $(\loc + a)(b)= \loc(b)$ for $b \neq a$ and $(\loc + a)(b)= \max\{\loc(a)+1,\bound\}$ otherwise. This is extended for all $u \in \Alpha^*$ and $a \in \Alpha$ by $\loc + \varepsilon= \loc$ and $\loc + ua = (\loc+u) + a$.
By $\loclet{w}$, we denote the location $\loc_0 + w$.

\medskip

\myparagraph{Configurations}
As explained above, a \emph{configuration} of $\game$ is
a mapping $\conf: \Loc \to \N^\Types$.
Suppose that, for $\loc \in \Loc$ and $\theta \in \Types$,
we have $\conf(\loc) = (n_\sys,n_\env,n_\both)$.
Then, we let $\conf(\loc,\theta)$ refer to $n_\theta$.
By $\AllConf$, we denote the set of all configurations.

\medskip

\myparagraph{Transitions}
A \emph{system transition} (respectively \emph{environment transition}) is a mapping
$\upd: \Loc \times \Loc \to (\N \times \{0\} \times \N)$
(respectively $\upd: \Loc \times \Loc \to (\{0\} \times \N \times \N)$)
such that, for all $(\loc,\loc') \in \Loc \times \Loc$
with $\upd(\loc,\loc') \neq (0,0,0)$,
there is a word $w \in \Alpha_\sys^\ast$ (respectively $w \in \Alpha_\env^\ast$) such that $\loc' = \loc + w$.
Let $\sysUpdates$ denote the set of system transitions, $\envUpdates$ the set of environment transitions, and $\Updates = \sysUpdates \cup \envUpdates$ the set of all transitions.

For $\upd \in \Updates$, let the mappings $\outm{\upd},\inm{\upd}: \Loc \to \N^\Types$ be defined by $\outm{\upd}(\loc) \df \sum_{\loc' \in \Loc} \upd(\loc,\loc')$ and
$\inm{\upd}(\loc) \df \sum_{\loc' \in \Loc} \upd(\loc',\loc)$
(recall that sum is component-wise).
We say that $\upd \in \Updates$ is \emph{applicable} at $\conf \in \AllConf$ if, for
all $\loc \in \Loc$, we have $\outm{\upd}(\loc) \le \conf(\loc)$ (component-wise).
Abusing notation, we let $\upd(\conf)$ denote the configuration $\conf'$ defined by
$\conf'(\loc) \df \conf(\loc) - \outm{\upd}(\loc) + \inm{\upd}(\loc)$
for all $\loc \in \Loc$.
Moreover, for $\upd(\loc,\loc') = (n_\sys,n_\env,n_\both)$
and $\theta \in \Types$, we let $\upd(\loc,\loc',\theta)$ refer to $n_\theta$.

\medskip

\myparagraph{Plays}
Let $\conf \in \AllConf$.
We write $\conf \models \Acc$ if there is $\accfunction \in \Acc$
such that, for all $\loc \in \Loc$, we have $\conf(\ell) \models \accfunction(\ell)$ (in the expected manner).
A $\conf$-\emph{play}, or simply \emph{play}, is a finite sequence
$
\pi = \conf_0\upd_1\conf_1\upd_2\conf_2 \ldots \upd_n\conf_n
$
alternating between configurations and transitions
(with $n \ge 0$) such that $\conf_0 = \conf$ and, for all $i \in \{1,\ldots,n\}$,
$\conf_{i} = \upd_i(\conf_{i-1})$ and
\begin{itemize}\itemsep=0.5ex
\item if $i$ is odd, then $\upd_i \in \sysUpdates$ and $\conf_i \models \Acc$ (System's move),

\item if $i$ is even, then $\upd_i \in \envUpdates$ and $\conf_i \not\models \Acc$ (Environment's move).
\end{itemize}
The set of all $\conf$-plays is denoted by $\cPlays{\conf}$.

\medskip

\myparagraph{Strategies}
A $\conf$-\emph{strategy} for System is a partial mapping $\strat: \cPlays{\conf} \to \sysUpdates$ such that
$\strat(\conf)$ is defined and, for all $\pi = \conf_0\upd_1\conf_1 \ldots \upd_i\conf_i \in \cPlays{\conf}$ with $\upd = \strat(\pi)$ defined, we have that $\upd$ is applicable at $\conf_i$ and $\tau(\conf_i) \models \Acc$. Play $\pi = \conf_0\upd_1\conf_1 \ldots \upd_n\conf_n$ is
\begin{itemize}\itemsep=0.5ex
\item $\strat$-\emph{compatible} if, for all odd $i \in \{1,\ldots,n\}$, $\tau_{i} = \strat(\conf_0\upd_1\conf_1 \ldots \upd_{i-1}\conf_{i-1})$,
\item $\strat$-\emph{maximal} if it is not the strict prefix of an $\strat$-compatible play,
\item \emph{winning} if $\conf_n \models \Acc$.
\end{itemize}
We say that $\strat$ is \emph{winning} for System (from $\conf$) if all
$\strat$-compatible $\strat$-maximal $\conf$-plays are winning.
Finally, $\conf$ is \emph{winning} if there is a $\conf$-strategy
that is winning.
Note that, given an initial configuration $\conf$,
we deal with an acyclic finite reachability game so that,
if there is a winning $\conf$-strategy, then there is a positional
one, which only depends on the last configuration.

For $\tvec \in \N^\Types$, let $\conf_{\tvec}$ denote the configuration
that maps $\loc_0$ to $\tvec$ and all other locations to $(0,0,0)$.
We set
$
\Win{\Alpha_\sys}{\Alpha_\env}{\game} = \{\tvec \in \N^\Types \mid
\conf_{\tvec} \textup{ is winning for System}\}
$.

\begin{definition}[game problem]\label{def:gameproblem}
For sets $\sN,\eN,\seN \subseteq \N$, the game problem is given as follows:
\begin{center}
\begin{decproblem}
  \problemtitle{$\Gameproblem{\sN}{\eN}{\seN}$}
  \probleminput{Parameterized vector game $\game$}
  \problemquestion{$\Win{}{}{\game} \cap (\sN \times \eN \times \seN) \neq \emptyset$\,\textup{?}}
\end{decproblem}
\end{center}
\end{definition}

We can show that parameterized vector games are equivalent to the synthesis problem in the following sense
(the proof can be found in Appendix~\ref{app:equifogame}):

\begin{lemma}\label{lemma:equifogame}
For every sentence $\varphi \in \FOdata$, there is a parameterized vector game $\G=(\Alpha,\bound,\Acc)$ such that
$\Win{}{}{\varphi} = \Win{}{}{\game}$. Conversely,
for every parameterized vector game $\game=(\Alpha,\bound,\Acc)$, there is a sentence $\varphi \in \FOdata$
such that $\Win{}{}{\game} = \Win{}{}{\varphi}$.
Both directions are effective.
\end{lemma}

\begin{figure}[t]
\centering
\includegraphics[width=0.95\textwidth]{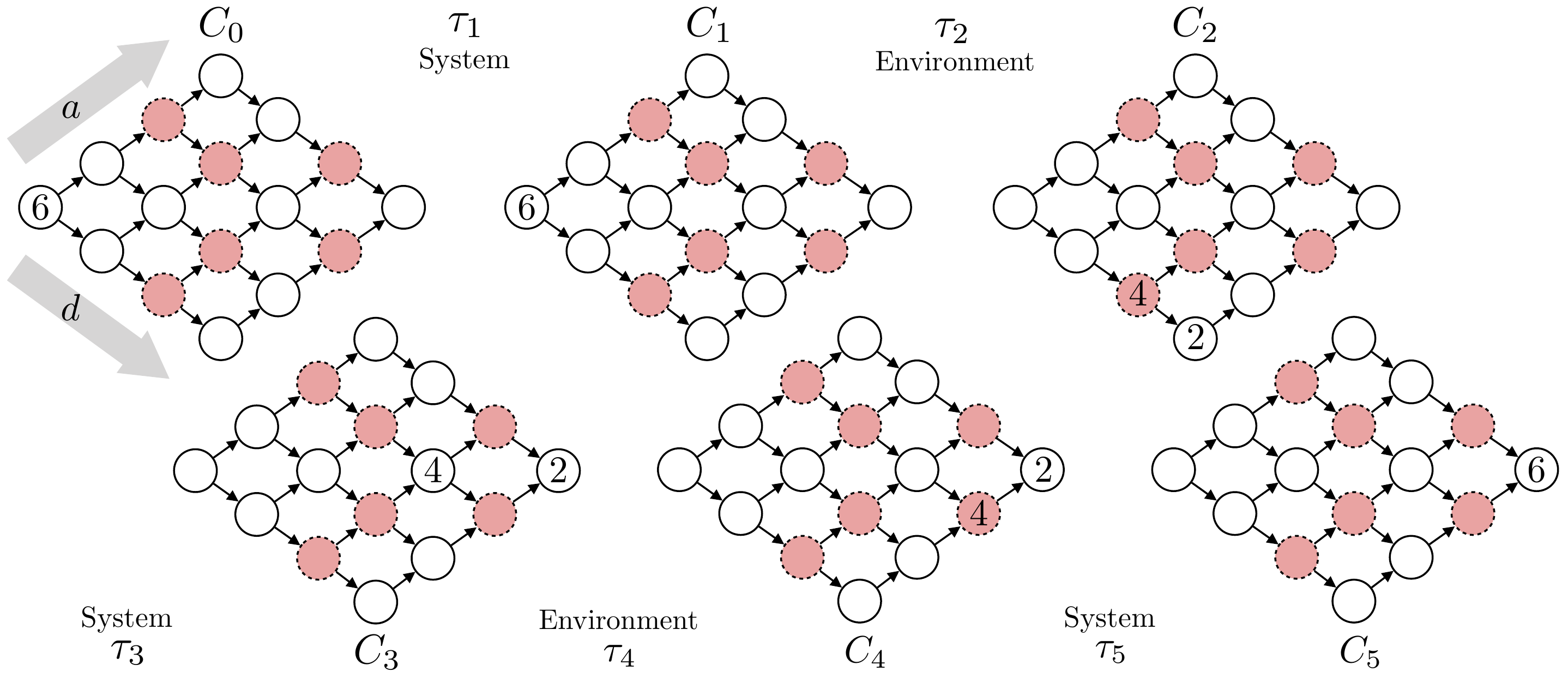}
\caption{A play of a parameterized vector game\label{fig:play}}
\end{figure}

\begin{example}\label{ex:game}
To illustrate parameterized vector games and the reduction from
the synthesis problem,
consider the formula $\varphi_4'=\bigwedge_{{\theta \in \Types\text{, }\loc \in Z}}
\exists^{= 0} y. \bigl(\theta(y) \wedge \ctype{3,\loc}(y)\bigr)$
in normal form from Example~\ref{ex:normalform}.
For simplicity, we assume that $\Alpha_\sys = \{a\}$ and $\Alpha_\env = \{d\}$.
That is, $Z$ is the set of vectors $\loclet{a^id^j} \in \Loc = \{0,\ldots,3\}^{\{a,d\}}$
such that $i = 2 \neq j$ or $j = 2 \neq i$.
Figure~\ref{fig:play} illustrates a couple of configurations $\conf_0,\ldots,\conf_5: \Loc \to \N^\Types$.
The leftmost location in a configuration is $\loc_0$, the rightmost
location $\loclet{a^3d^3}$, the topmost one
$\loclet{a^3}$, and the one at the bottom $\loclet{d^3}$.
Self-loops have been omitted, and locations from $Z$ have red background and a dashed border.

Towards an equivalent game $\game = (\Alpha,3,\Acc)$, it remains to
determine the acceptance condition $\Acc$.
Recall that $\phi_4'$ says that every class contains two occurrences of $a$
iff it contains two occurrences of $d$.
This is reflected by the acceptance condition
$\Acc = \{\accfunction\}$ where $\accfunction(\loc) = ({=}0\,,{=}0\,,{=}0)$ for all $\loc \in Z$ and
$\accfunction(\loc) = ({\ge}0\,,{\ge}0\,,{\ge}0)$ for all $\loc \in \Loc \setminus Z$.
With this, a configuration is accepting iff no
token is on a location from $Z$ (a red location).

We can verify that $\Win{}{}{\game} = \Win{}{}{\varphi_4'} = \N \times \Zero \times \N$.
In $\game$, a uniform winning strategy $\strat$ for System that works for all $\Procs$ with $\eProcs = \emptyset$
proceeds as follows: System first awaits an Environment's move and then
moves each token upwards as many locations as Environment has moved it downwards.
Figure~\ref{fig:play} illustrates an $\strat$-maximal $\conf_{(6,0,0)}$-play that is winning for System.
We note that $\strat$ is a ``compressed'' version of the winning strategy presented in
Example~\ref{ex:synthesis}, as System makes her moves only when really needed.
\exend
\end{example}

\section{Results for $\wFOdata$ via Parameterized Vector Games}\label{sec:results}
In this section, we present our results for the synthesis problem
for $\wFOdata$, which we obtain showing corresponding results for parameterized vector games.
In particular, we show that
$(\wFOdata,\Zero,\Zero,\N)$ and $(\wFOdata,\N,\N,\Zero)$
do not have a cutoff, whereas \mbox{$(\wFOdata,\N,{\{\pconsta\}},{\{\pconstb\}})$}
has a cutoff for all $\pconsta, \pconstb \in \N$. Finally, we prove that 
$\Synthesis{\wFOdata}{\Zero}{\Zero}{\N}$ is, in fact, undecidable.

\begin{lemma}\label{lem:cutoffgame00N}
There is a game $\game=(\Alpha,\bound,\Acc)$ such that $\Win{}{}{\game}$ does not have a cutoff wrt.\ $(\Zero,\Zero,\N)$.
\end{lemma}

\newcommand{\Accfunction}{K}
\newcommand{\kequal}{{^=k}}
\newcommand{\kplus}{{^{\ge}k}}
\newcommand{\kibowtie}[1]{^{\bowtie_{#1}}k_{#1}}
\newcommand{\zeroequal}{{^=0}}
\newcommand{\zeroplus}{{^{\ge}0}}
\newcommand{\oneequal}{{^=1}}
\newcommand{\oneplus}{{^{\ge}1}}
\newcommand{\twoequal}{{^=2}}
\newcommand{\twoplus}{{^{\ge}2}}
\newcommand{\locenv}{\Loc_\env}
\newcommand{\acckappa}[5]{[#1\,,\hspace{0em}#2\,,\hspace{0em}#3\,,\hspace{0em}#4\,,\hspace{0em}#5]}

\begin{proof}
We let
$\Alpha_\sys = \{a\}$ and $\Alpha_\env = \{b\}$, as well as
$\bound = 2$.
For $k \in \{0,1,2\}$, define the local acceptance conditions
$\kequal = ({=}0\,, {=}0\,, {=}k)$ and $\kplus = ({=}0\,, {=}0\,, {\ge}k)$.
Set $\loc_1 = \loclet{a}, \loc_2 = \loclet{ab}, \loc_3 = \loclet{a^2b},$ and $\loc_4 = \loclet{a^2b^2}$. For $k_0,\ldots,k_4 \in \{0,1,2\}$ and ${\bowtie_0},\ldots,{\bowtie_4} \in \{=,\ge\}$, let
$\acckappa{\kibowtie{0}}{\kibowtie{1}}{\kibowtie{2}}{\kibowtie{3}}{\kibowtie{4}}$ denote $\accfunction \in \locAcc^\Loc$ where
$\accfunction(\loc_i) = ({\kibowtie{i}})$ for all $i \in \{0,\ldots,4\}$
and $\accfunction(\loc') = (\zeroequal)$ for $\loc' \notin \{\loc_0, \ldots, \loc_4\}$.
Finally,
\[
\Acc =
\left\{
\begin{array}{lclcl}
\acckappa{\zeroplus}{\twoequal}{\zeroequal}{\zeroequal}{\zeroplus} &\;&
	\acckappa{\zeroplus}{\zeroequal}{\zeroequal}{\twoequal}{\zeroplus} &\;&
		\acckappa{\zeroequal}{\zeroequal}{\zeroequal}{\zeroequal}{\twoplus}\\[0.3ex]
\acckappa{\zeroplus}{\oneequal}{\oneequal}{\zeroequal}{\zeroplus} &\;&
	\acckappa{\zeroplus}{\zeroequal}{\zeroequal}{\oneequal}{\oneplus}
\end{array}
\right\} \cup \Accfunction_\env\\
\]
where
$\Accfunction_\env = \{\accfunction_\loc \mid \loc \in \Loc$ such that $\loc(b) > \loc(a)\}$ with $\accfunction_\loc(\loc') = ({\oneplus})$ if $\loc' = \loc$, and $\accfunction_\loc(\loc') = (\zeroplus)$ otherwise.
This is illustrated in Figure~\ref{fig:nocutoff1}.

\begin{figure}[t]
\centering
\includegraphics[width=0.85\textwidth]{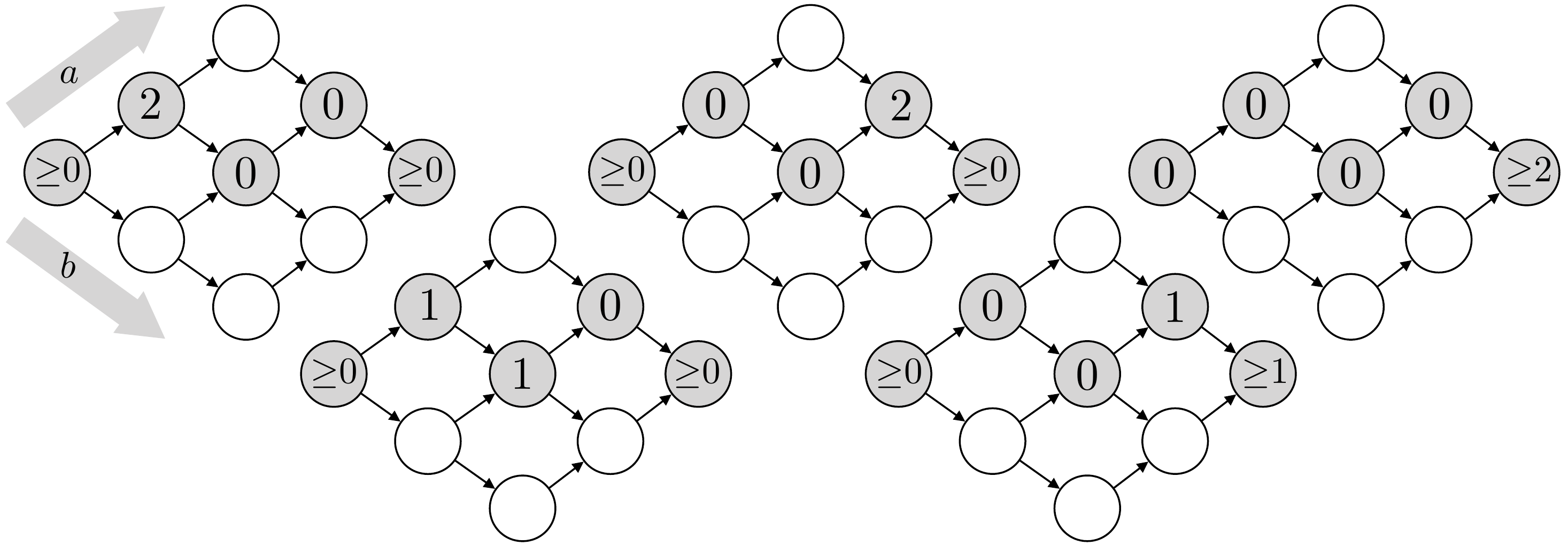}
\caption{Acceptance conditions for a game with no cutoff wrt.\ $(\Zero,\Zero,\N)$\label{fig:nocutoff1}}
\end{figure}

There is a winning strategy for System from any initial configuration of size $2n$: Move two tokens from $\loc_0$ to $\loc_1$, wait until Environment sends them both to $\loc_2$, then move them to $\loc_3$, wait until they are moved to $\loc_4$, then repeat with two new tokens from $\loc_0$ until all the tokens are removed
from $\loc_0$, and Environment cannot escape $\Acc$ anymore.
However, one can check that there is no winning strategy for initial configurations of odd size.
\qed
\end{proof}

\begin{lemma}\label{lem:cutoffgameNN0}
There is a game $\game=(\Alpha,\bound,\Acc)$ such that $\Win{}{}{\game}$ does not have a cutoff wrt.\ $(\N,\N,\Zero)$.
\end{lemma}

\begin{proof}
We define $\game$ such that System wins only if she has at least as many processes as Environment. Let $\Alpha_\sys = \{a\}$, $\Alpha_\env = \{b\}$, and $\bound = 2$.
As there are no shared processes, we can safely ignore locations with a letter from both System and Environment.
We set $\Acc = \{\accfunction_1, \accfunction_2, \accfunction_3, \accfunction_4\}$ where
\[
\begin{array}{lclcl}
\accfunction_1(\loclet{a}) = ({=}1\,, {=}0\,, {=}0) & \;~ &
	\accfunction_2(\loclet{a}) = ({=}1\,, {=}0\,, {=}0) & ~ &
		\accfunction_3(\loclet{a}) = ({=}0\,, {=}0\,, {=}0)\\[0.3ex]
\accfunction_1(\loclet{b}) = ({=}0\,, {=}0\,, {=}0) & \;~ &
	\accfunction_2(\loclet{b}) = ({=}0\,, {\ge}2\,, {=}0) & ~ &
		\accfunction_3(\loclet{b}) = ({=}0\,, {\ge}1\,, {=}0)\,,
\end{array}
\]
$\accfunction_4(\loc_0) = ({=}0\,, {=}0\,, {=}0)$, and $\accfunction_i(\loc') = ({\ge}0\,, {\ge}0\,, {=}0)$ for all other $\loc' \in \Loc$ and $i \in \{1,2,3,4\}$.
The details are available in Appendix~\ref{app:cutoffgameNN0}.~\qed
\end{proof}

We now turn to the case where the number of processes that can be triggered by Environment is bounded.
Note that similar restrictions are imposed in other settings to get
decidability, such as limiting the environment to a finite (Boolean) domain \cite{FigueiraP18} or
restricting to one environment process \cite{FinkbeinerO17,beutner_et_al:concur:2019}.
We obtain decidability of the synthesis problem via a cutoff construction:

\begin{theorem}\label{thm:cutoffgameNcc}
Given $\pconsta, \pconstb \in \N$,
every game $\game=(\Alpha,\bound,\Acc)$ has a cutoff wrt.\ $(\N,\{\pconsta\},\{\pconstb\})$.
More precisely:
Let $\ConstF$ be the largest constant that occurs in $\Acc$.
Moreover, let
$\maxdist \df (\pconsta + \pconstb) \cdot |A_\env| \cdot \bound$ and
$\cut = |\Loc|^{\maxdist+1} \cdot \ConstF$.
Then, $(\cut,\pconsta,\pconstb)$ is a cutoff of $\Win{\Alpha_\sys}{\Alpha_\env}{\game}$ wrt.\ $({\N},\{\pconsta\},\{\pconstb\})$.
\end{theorem}

\newcommand{\dConf}[1]{\AllConf_{\!#1}}

\begin{proof}
We will show that, for all $\nproc \ge \cut$,
\[(\nproc,\pconsta,\pconstb) \in \Win{\Alpha_\sys}{\Alpha_\env}{\game} ~\Longleftrightarrow~ (\nproc+1,\pconsta,\pconstb) \in \Win{\Alpha_\sys}{\Alpha_\env}{\game}\,.\]

The main observation is that, when $\conf$ contains more than $K$ tokens in a given $\loc\in\Loc$, adding more tokens in $\loc$ will not change whether $\conf\models\Acc$. 
Given $\conf,\conf' \in \AllConf$,
we write $\conf <_\env \conf$ if $\conf \neq \conf'$ and there is $\upd \in \Updates_\env$
such that $\upd(\conf) = \conf'$.
Note that the length
$d$ of a chain $\conf_0 <_\env \conf_1 <_\env \ldots <_\env \conf_d$
is bounded by $\maxdist$. In other words,
$\maxdist$ is the maximal number of
transitions that Environment can do in a play.
For all $\dist \in \{0,\ldots,\maxdist\}$,
let $\dConf{d}$ be the set of configurations $\conf \in \AllConf$
such that the longest chain in $(\AllConf,<_\env)$ starting from
$\conf$ has length~$d$.

\begin{claim}
Suppose that $\conff \in \dConf{d}$ and $\loc \in \Loc$ such that $\conff(\loc)  = (\nproc,n_\env,n_\both)$ with $\nproc \ge |\Loc|^{d+1} \cdot \ConstF$ and $n_\env,n_\both \in \N$. Set $\hatconf = \conff[\loc \mapsto (\nproc+1,n_\env,n_\both)]$. Then,
\begin{center}
$\conff$ is winning for System $~\Longleftrightarrow~$ $\hatconff$ is winning for System.
\end{center}
\end{claim}
\newcommand{\locp}{\loc'}

To show the claim, we proceed by induction
on $d \in \N$, which is illustrated in Figure~\ref{fig:cutoff}.
In each implication, we distinguish the cases $d=0$ and $d \ge 1$.
For the latter, we assume that equivalence holds for all
values strictly smaller than~$d$.

For
$\upd \in \sysUpdates$ and $\loc,\locp \in \Loc$, we let
$\inc{\upd}{(\loc,\locp,\sys)}$ denote the transition
$\hatupd \in \sysUpdates$ given by
$\hatupd(\loc_1,\loc_2,\env) = \upd(\loc_1,\loc_2,\env) = 0$,
$\hatupd(\loc_1,\loc_2,\both) = \upd(\loc_1,\loc_2,\both)$,
$\hatupd(\loc_1,\loc_2,\sys) =
\upd(\loc_1,\loc_2,\sys) + 1$ if $(\loc_1,\loc_2) = (\loc,\locp)$,
and
$\hatupd(\loc_1,\loc_2,\sys) =
\upd(\loc_1,\loc_2,\sys)$ if $(\loc_1,\loc_2) \neq (\loc,\locp)$.
We define $\dec{\upd}{(\loc,\locp,\sys)}$ similarly (provided $\upd(\loc,\locp,\sys) \ge 1$).

\begin{figure}[t]
\begin{center}
\includegraphics[width=0.9\textwidth]{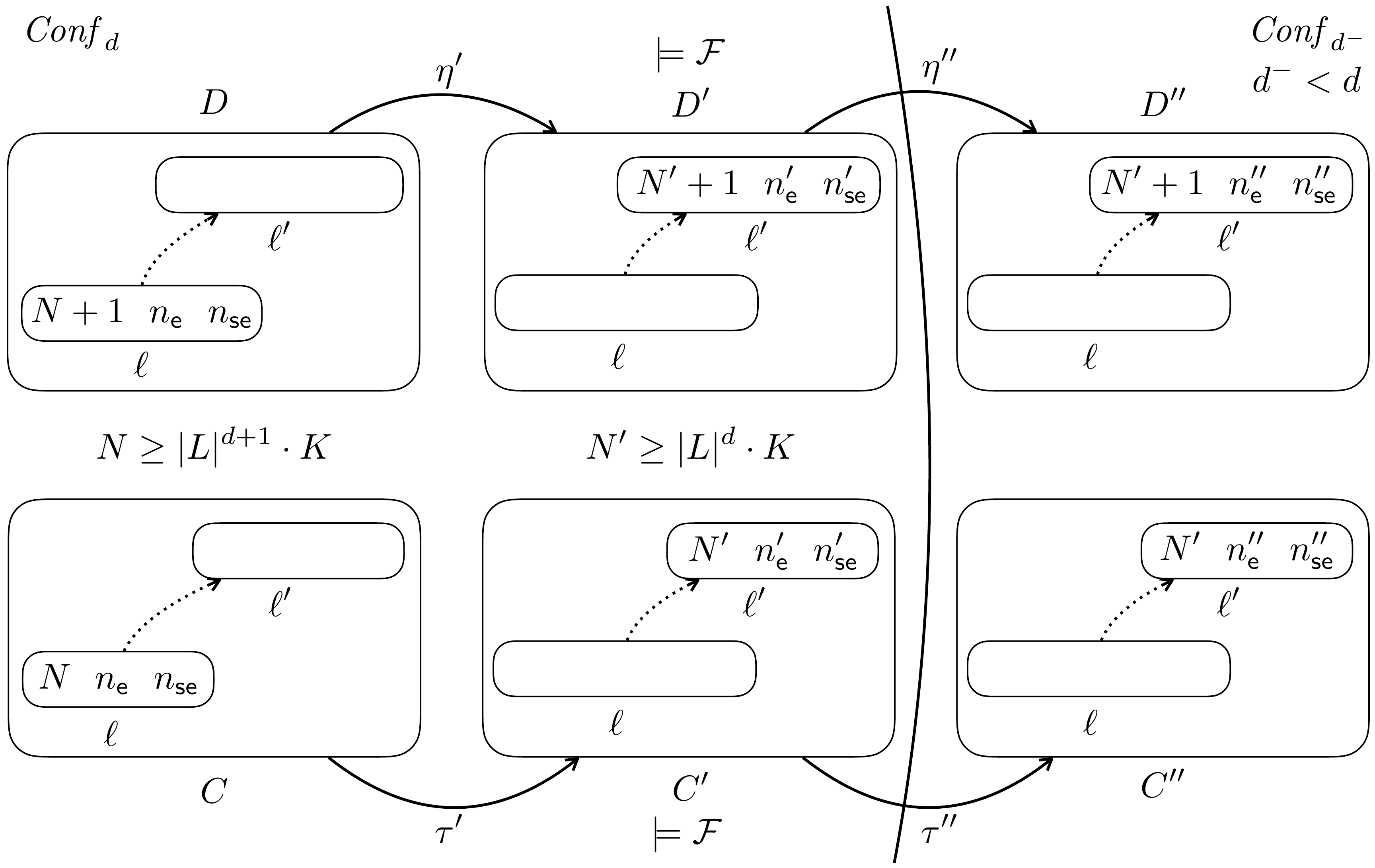}
\vspace{-4ex}
\end{center}
\caption{Induction step in the cutoff construction\label{fig:cutoff}}
\end{figure}

\paragraph{$\Longrightarrow$\textup{:}}

Let $\stratf$ be a winning strategy for System from
$\conff \in \dConf{d}$.
Let $\updp = f(\conff)$ and $\confp = \updp(\conff)$.
Note that $\confp \models \Acc$.
Since $\conff(\loc,\sys) = N \ge |\Loc|^{d+1} \cdot \ConstF$,
there is $\locp \in \Loc$ such that
$\loc + w = \locp$ for some $w \in \Alpha_\sys^\ast$
and $\confp(\locp,\sys) = N' \ge |\Loc|^{d} \cdot \ConstF$.

We show that $\hatconf = \conff[\loc \mapsto (\nproc+1,n_\env,n_\both)]$ is winning for System by exhibiting
a corresponding winning strategy $\hatstratf$ from $\hatconff$ that will carefully control the position of the additional token.
First, set $\hatstratf(\hatconff) = \hatupdp$ where
$\hatupdp = \inc{\updp}{(\loc,\locp,\sys)}$.
Let $\hatconfp = \hatupdp(\hatconff)$.
We obtain $\hatconfp(\locp,\sys) = N'+1$.
Note that, since $N' \ge \ConstF$, the acceptance condition
$\Acc$ cannot distinguish between
$\confp$ and $\hatconf'$.
Thus, we have $\hatconfp \models \Acc$.

\begin{description}\itemsep=1ex
\item[\normalfont Case $d=0$:]

As, for all transitions $\hatupdpp \in \envUpdates$,
we have $\hatupdpp(\hatconfp) = \hatconfp \models \Acc$,
we reached a maximal play that is winning for System.
We deduce that $\hatconff$ is winning for System.

\item[\normalfont Case $d \ge 1$:]

Take any $\hatupdpp \in \envUpdates$ and $\hatconfpp$ such that
$\hatconfpp = \hatupdpp(\hatconfp) \not\models \Acc$.
Let $\updpp = \hatupdpp$ and $\confpp = \updpp(\confp)$.
Note that
$\hatconfpp = \confpp[(\locp,\sys) \mapsto N+1]$,
$\confpp = \hatconfpp[(\locp,\sys) \mapsto N]$,
and $\confpp,\hatconfpp \in \dConf{d^-}$ for some $d^- < d$.
As $\strat$ is a winning strategy for System from $\conff$,
we have that $\confpp$ is winning for System.
By induction hypothesis, 
$\hatconfpp$ is winning for System, say by winning strategy
$\hatstratpp$.
We let
$
\hatstrat(\hatconff\, \hatupdp\, \hatconfp\, \hatupdpp\, \pi) = \hatstratpp(\pi)
$
for all $\hatconfpp$-plays $\pi$.
For all unspecified plays, let $\hatstratf$ return any applicable system
transition. Altogether, for any choice of $\hatupdpp$, we have that $\hatstratpp$ is winning from $\hatconfpp$. Thus, $\hatstratf$ is a winning strategy from $\hatconff$.
\end{description}

\paragraph{$\Longleftarrow$\textup{:}}

Suppose
$\hatstratf$ is a winning strategy for System from $\hatconff$.
Thus, for $\hatupdp = \hatstratf(\hatconff)$ and
$\hatconfp = \hatupdp(\hatconff)$, we have $\hatconfp \models \Acc$.
Recall that $\hatconff(\loc,\sys) \ge (|\Loc|^{d+1} \cdot \ConstF) + 1$.
We distinguish two cases:
\begin{enumerate}\itemsep=0.5ex
\item Suppose there is $\locp \in \Loc$
such that $\loc \neq \locp$,
$\hatconfp(\locp,\sys) = N'+1$
for some $N' \ge |\Loc|^{d} \cdot \ConstF$, and
$\hatupdp(\loc,\locp,\sys) \ge 1$. Then, we set
$\updp = \dec{\hatupdp}{(\loc,\locp,\sys)}$.

\item Otherwise, we have $\hatconfp(\loc,\sys) \ge (|\Loc|^{d} \cdot \ConstF) + 1$,
and we set $\updp = \hatupdp$ (as well as $\locp = \loc$ and $N' = N$).
\end{enumerate}
Let $\confp = \updp(\conff)$.
Since $\hatconfp \models \Acc$, one obtains $\confp \models \Acc$.

\begin{description}\itemsep=1ex
\item[\normalfont Case $d=0$:]
For all transitions $\updpp \in \envUpdates$,
we have $\updpp(\confp) = \confp \models \Acc$.
Thus, we reached a maximal play that is winning for System.
We deduce that $\conff$ is winning for System.

\item[\normalfont Case $d \ge 1$:]
Take any $\updpp \in \envUpdates$ such that
$\conf'' = \updpp(\confp) \not\models \Acc$.
Let $\hatupdpp = \updpp$ and $\hatconfpp = \hatupdpp(\hatconfp)$.
We have $\confpp = \hatconfpp[(\locp,\sys) \mapsto N']$,
$\hatconfpp = \confpp[(\locp,\sys) \mapsto N'+1]$,
and $\confpp,\hatconfpp \in \dConf{d^-}$ for some $d^- < d$.
As $\hatconfpp$ is winning for System,
by induction hypothesis, 
$\confpp$ is winning for System, say by winning strategy
$\stratpp$.
We let
$
\stratf(\conff\, \updp\, \confp\, \updpp\, \pi) = \stratpp(\pi)
$
for all $\confpp$-plays $\pi$.
For all unspecified plays, let $\stratf$ return an arbitrary applicable system
transition. Again, for any choice of $\updpp$,
$\stratpp$ is winning from $\confpp$. Thus, $\stratf$
is a winning strategy from $\conff$.
\end{description}

This concludes the proof of the claim and, therefore, of Theorem~\ref{thm:cutoffgameNcc}.
\qed
\end{proof}

\begin{corollary}\label{cor:decgameNcc}
Let $\pconsta, \pconstb \in \N$ be the number of environment and the number of mixed processes, respectively.
The problems $\Gameproblem{\N}{\{\pconsta\}}{\{\pconstb\}}$ and
$\Synthesis{\wFOdata}{\N}{\{\pconsta\}}{\{\pconstb\}}$ are decidable.
\end{corollary}

\newcommand{\first}{\text{First}}
\newcommand{\second}{\text{Second}}
\newcommand{\third}{\text{Third}}
\newcommand{\old}{\text{Old}}
\newcommand{\typeletter}[1]{t_\Alpha^{#1}}
\newcommand{\typeequal}[1]{t_{=}^{#1}}
\newcommand{\typebefore}[1]{t_{\dataord}^{#1}}
\newcommand{\typeafter}[1]{t_{\dataordinv}^{#1}}

\begin{theorem}\label{thm:game00N}
$\Gameproblem{\Zero}{\Zero}{\N}$ and
$\Synthesis{\wFOdata}{\Zero}{\Zero}{\N}$ are undecidable.
\end{theorem}

\begin{proof}
We provide a reduction from the halting problem for 2-counter machines (2CM) to $\Gameproblem{\Zero}{\Zero}{\N}$. A 2CM $\tcm=(\tcmQ,\tcmT,\tcmcounter_1,\tcmcounter_2,\tcminit,\tcmfinal)$ has two counters, $\tcmcounter_1$ and $\tcmcounter_2$, a finite set of states $\tcmQ$, and a set of transitions~$\tcmT \subseteq Q \times \Op \times Q$
where $\Op = \{\incc{i}\,,\,\decc{i}\,,\,\zeroc{i} \mid i \in \{1,2\}\}$. Moreover, we have an initial state $\tcminit\in Q$ and a halting state $\tcmfinal\in Q$. 
A configuration of $\tcm$ is a triple $\gamma = (\tcmq, \nu_1, \nu_2)\in Q\times\N\times\N$ giving the current state and the current respective counter values. The initial configuration is $\gamma_0=(\tcminit, 0, 0)$ and the set of halting configurations is $\tcmF=\{\tcmfinal\}\times\N\times\N$.
For $t \in \tcmT$, configuration $(\tcmq', \nu_1', \nu_2')$ is a ($t$-)successor of $(\tcmq, \nu_1, \nu_2)$,
written $(\tcmq,\nu_1,\nu_2)\vdash_\tcmt (\tcmq',\nu'_1,\nu'_2)$, if
there is $i \in \{1,2\}$ such that $\nu_{3-i}' = \nu_{3-i}$ and
one of the following holds: $(i)$ $t=(\tcmq,\incc{i},\tcmq')$ and $\nu'_i=\nu_i+1$, or $(ii)$ $t=(\tcmq,\decc{i},\tcmq')$ and $\nu'_i=\nu_i-1$, or $(iii)$ $t=(\tcmq,\zeroc{i},\tcmq')$ and $\nu_i=\nu'_i=0$.
A run of $\tcm$ is a (finite or infinite) sequence $\gamma_0 \vdash_{\tcmt_1} \gamma_1 \vdash_{\tcmt_2} \dots$.
The 2CM halting problem asks whether there is a run reaching a configuration in $\tcmF$. It is known to be undecidable~\cite{Minsky67}.

\smallskip

We fix a 2CM $\tcm=(\tcmQ,\tcmT,\tcmcounter_1,\tcmcounter_2,\tcminit,\tcmfinal)$.
Let $\Alpha_\sys = \tcmQ \cup \tcmT \cup \{a_1, a_2\}$ and $\Alpha_\env = \{b\}$ with $a_1$, $a_2$, and $b$ three fresh symbols.
We consider the game $\game=(\Alpha,\bound,\Acc)$ with $\Alpha=\Alpha_\sys\uplus\Alpha_\env$, $\bound = 4$, and $\Acc$ defined below. Let $\Loc = \{0,\ldots,\bound\}^\Alpha$. Since there are only processes shared by System and Environment, we alleviate notation and consider that a configuration is simply a mapping $\conf: \Loc \to \N$.
From now on, to avoid confusion, we refer to configurations of the 2CM $\tcm$ as $\tcm$-configurations, and to configurations of $\game$ as $\game$-configurations.

Intuitively, every valid run of $\tcm$ will be encoded as a play in $\game$, and the acceptance condition will enforce that, if a player in $\game$ deviates from a valid play, then she will lose immediately.
At any point in the play, there will be at most one process with only a letter from $\tcmQ$ played, which will represent the current state in the simulated 2CM run.
Similarly, there will be at most one process with only a letter from $\tcmT$ to represent what transition will be taken next.
Finally, the value of counter $\tcmcounter_i$ will be encoded by the number of processes with exactly two occurrences of $a_i$ and two occurrences of $b$ (i.e., $\conf(\loclet{a_i^2b^2})$).

To increase counter $\tcmcounter_i$, the players will move a new token to $\loclet{a_i^2b^2}$, and to decrease it, they will move, together, a token from $\loclet{a_i^2b^2}$ to $\loclet{a_i^4 b^4}$. Observe that, if $\tcmcounter_i$ has value $0$, then $\conf(\loclet{a_i^2b^2})=0$ in the corresponding configuration
of the game. As expected, it is then impossible to simulate the decrement of $\tcmcounter_i$.
Environment's only role is to acknowledge System's actions by playing its (only) letter when System simulates a valid run. If System tries to cheat, she loses immediately.

\medskip
\myparagraph{Encoding an $\tcm$-configuration}
Let us be more formal. Suppose $\gamma = (\tcmq, \nu_1, \nu_2)$ is an $\tcm$-configuration
and $\conf$ a $\game$-configuration. We say that $\conf$ \emph{encodes} $\gamma$ if
\begin{itemize}\itemsep=0.3ex
\item 
$\conf(\loclet{\tcmq}) = 1$, $\conf(\loclet{a_1^2 b^2}) = \nu_1$, $\conf(\loclet{a_2^2 b^2}) = \nu_2$,
\item 
$\conf(\loc) \ge 0$ for all $\loc \in \{\loc_0\} \cup \{\loclet{\hat\tcmq^2 b^2}, \loclet{\tcmt^2 b^2}, \loclet{a_i^4 b^4} \mid \hat\tcmq \in \tcmQ, \tcmt \in \tcmT, i \in \{1,2\}\}$,
\item 
$\conf(\loc) = 0$ for all other $\loc \in \Loc$.
\end{itemize}
We then write $\gamma= \mconf(\conf)$.
Let $\confset(\gamma)$ be the set of  $\game$-configurations $\conf$ that encode~$\gamma$.
We say that a $\game$-configuration $\conf$ is \emph{valid} if $\conf\in \confset(\gamma)$ for some $\gamma$.

\medskip
\myparagraph{Simulating a transition of $\tcm$}
Let us explain how we go from a $\game$-configuration encoding $\gamma$ to a $\game$-configuration encoding a successor $\tcm$-configuration
$\gamma'$. Observe that System cannot change by herself the $\tcm$-configuration encoded. If, for instance, she tries to change the current state $q$, she might move one process from $\loc_0$ to $\loclet{q'}$, but then the $\game$-configuration is not valid anymore. We need to move the process in $\loclet{q}$ into 
$\loclet{q^2b^2}$ and this requires the cooperation of Environment. 

Assume that the game is in configuration $C$ encoding $\gamma=(q,\nu_1,\nu_2)$. System will pick a transition $t$ starting in state $q$, say, 
$t = (q,\incc{1}, q')$. From configuration $C$, System will go to the configuration $C_1$ defined by $C_1(\loclet{t})=1$, $C_1(\loclet{a_1})=1$, and 
$C_1(\loc)=C(\loc)$ for all other $\loc\in\Loc$.

If the transition $t$ is correctly chosen, Environment will go to a configuration $C_2$
defined by $C_2(\loclet{q})=0$, $C_2(\loclet{qb})=1$, $C_2(\loclet{t})=0$, $C_2(\loclet{tb})=1$, $C_2(\loclet{a_1})=0$, $C_2(\loclet{a_1b})=1$ and, for all other  $\loc\in\Loc$, $C_2(\loc)=C_1(\loc)$.
This means that Environment moves processes in locations $\loclet{t}$, $\loclet{q}$, $\loclet{a_1}$ to locations $\loclet{tb}$, 
$\loclet{qb}$, $\loclet{a_1b}$, respectively.

To finish the transition, System will now move a process to the destination state $q'$ of $t$, and go to configuration $C_3$
defined by $C_3(\loclet{q'})=1$, $C_3(\loclet{tb})=0$, $C_3(\loclet{t^2b})=1$, $C_3(\loclet{qb})=0$, $C_3(\loclet{q^2b})=1$, $C_3(\loclet{a_1b})=0$,
$C_3(\loclet{a_1^2b})=1$, and $C_3(\loc)=C_2(\loc)$ for all other $\loc\in\Loc$.

Finally, Environment moves to configuration $C_4$ given by
$C_4(\loclet{t^2b})=0$, $C_4(\loclet{t^2b^2})=C_3(\loclet{t^2b^2})+1$, $C_4(\loclet{q^2b})=0$, $C_4(\loclet{q^2b^2})=C_3(\loclet{q^2b^2})+1$, $C_4(\loclet{a_1^2b})=0$, $C_4(\loclet{a_1^2b^2})=C_3(\loclet{a_1^2b^2})+1$, and 
$C_4(\loc)=C_3(\loc)$ for all other $\loc\in\Loc$. Observe that $C_4\in \confset((q', \nu_1+1, \nu_2))$.

Other types of transitions will be simulated similarly.
To force System to start the simulation in $\gamma_0$, and not in any $\tcm$-configuration, the configurations $C$ such that 
$C(\loclet{q_0^2b^2})=0$ and $C(\loclet{\tcmq})=1$ for $\tcmq\neq\tcminit$ are not valid, and will be losing for System.

\medskip
\myparagraph{Acceptance condition}
It remains to define $\Acc$ in a way that enforces the above sequence of $\game$-configurations.
Let $\Locok = \{\loc_0\} \cup \{\loclet{a_i^2 b^2}, \loclet{a_i^4 b^4} \mid i \in \{1,2\}\} \cup \{\loclet{\tcmq^2 b^2} \mid \tcmq \in \tcmQ\} \cup \{\loclet{\tcmt^2 b^2} \mid \tcmt \in \tcmT\}$ be the set of
elements in $\Loc$ whose values do not affect the acceptance of the configuration. By $[\loc_1\bowtie_1 n_1, \dots, \loc_k\bowtie_k n_k]$, we denote $\accfunction\in\locAcc^L$ such that
$\accfunction(\loc_i){\,=\,}{({\bowtie_i}n_i)}$ for $i \in \{1,\ldots,k\}$
and $\accfunction(\loc){\,=\,} ({=}0)$
for all $\loc\in\Loc\setminus\{\loc_1,\dots, \loc_k\}$.
Moreover, for a set of locations $\hat\Loc \subseteq \Loc$, we let $\wlocsetgezero{\hat\Loc}$ stand for ``$(\loc\geq 0)$ for all $\loc\in\hat\Loc$''.

\begin{table}
\begin{center}
\caption{\mbox{Acceptance conditions for the game simulating a 2CM}}
\label{tab:game00N}
\makebox[\textwidth]{%
\scalebox{1}{
\begin{tabular}{ l }
\toprule
\textbf{Requirements for System}\\
\midrule
\textbf{(a)}
For all $\tcmt = (q,\op,q') \in \tcmQ$:\\[-2ex]
\parbox{\textwidth}{
\begin{alignat*}{8}
\accCond{(\tcmq, \tcmt)}&=\textstyle\bigcup_{\hat\tcmq\in \tcmQ}\bigl\{[\loclet{q}=1,\;&&\loclet{t}=1, \;&&\loclet{a_i}=1, \;&&\loclet{{\hat\tcmq}^2b^2}\geq 1,
\;&&\locsetgezero{\Locok\setminus \{\loclet{{\hat\tcmq}^2b^2}\}}]\bigr\} && \text{ if } \op = \incc{i}\\
\accCond{(\tcmq, \tcmt)}&=\textstyle\bigcup_{\hat\tcmq\in \tcmQ}\bigl\{[\loclet{q}=1,\;&&\loclet{t}=1, \;&&\loclet{a_i^3b^2}=1, \;&&\loclet{{\hat\tcmq}^2b^2}\geq 1,
\;&&\locsetgezero{\Locok\setminus \{\loclet{{\hat\tcmq}^2b^2}\}} ]\bigr\} && \text{ if } \op = \decc{i} \\
\accCond{(\tcmq, \tcmt)}&=\textstyle\bigcup_{\hat\tcmq\in \tcmQ}\bigl\{[\loclet{q}=1,\;&&\loclet{t}=1, \;&&\loclet{a_i^2b^2}=0, \;&&\loclet{{\hat\tcmq}^2b^2}\geq 1,
\;&&\locsetgezero{\Locok\setminus \{\loclet{{\hat\tcmq}^2b^2}, \loclet{a_i^2b^2}\}} ]\bigr\}~ && \text{ if } \op = \zeroc{i}
\end{alignat*}}\\[-1ex]
\midrule
\textbf{(b)} 
For all $\tcmt = (\tcminit,\op,q') \in \tcmQ$ such that $\op \in \{\incc{i},\zeroc{i}\}$: \\[-2ex]
\parbox{\linewidth}{
\begin{alignat*}{8}
\accCond{\tcmt}&= \bigl\{[\loclet{\tcminit}=1, \;&&\loclet{\tcmt}=1, \;&&\loclet{a_i} = 1, \;&&\loc_0 \ge 0]\bigr\}~ && \text{ if } \op = \incc{i}\\
\accCond{\tcmt}&=\bigl\{[\loclet{\tcminit}=1, \;&&\loclet{\tcmt} = 1, \;&&\loc_0 \ge 0]\bigr\}~ && &&\text{ if } \op = \zeroc{i}
\end{alignat*}}\\[-1ex]
\midrule
\textbf{(c)}
For all $\tcmt = (q,\op,q') \in \tcmQ$:\\[-2ex]
\parbox{\textwidth}{
\begin{alignat*}{8}
\accCond{(\tcmq, \tcmt, \tcmq')}&=\bigl\{[\loclet{q^2b}=1, \;&&\loclet{t^2b}=1, \;&&\loclet{a_i^2b}=1, \;&&\loclet{\tcmq'}=1, \;&&\wlocsetgezero{\Locok}]\bigr\} ~&& \text{ if } \op = \incc{i}\\
\accCond{(\tcmq, \tcmt, \tcmq')}&=\bigl\{[\loclet{q^2b}=1,\;&&\loclet{t^2b}=1, \;&&\loclet{a_i^4b^3}=1, \;&&\loclet{\tcmq'}= 1, \;&&\wlocsetgezero{\Locok}]\bigr\} && \text{ if } \op = \decc{i}\\
\accCond{(\tcmq, \tcmt,\tcmq')}&=\bigl\{[\loclet{q^2b}=1,\;&&\loclet{t^2b}=1, \;&&\wlocsetgezero{\Locok} ]\bigr\} && && && \text{ if } \op = \zeroc{i}
\end{alignat*}}\\[-1ex]
\midrule
\midrule
\textbf{Requirements for Environment}\\
\midrule
\textbf{(d)}
Let $\Loc_{\sys < \env} = \bigl\{\loc \in \Loc \mid \bigl(\sum_{\alpha \in \Alpha_\sys} \loc(\alpha)\bigr) < \loc(b)\bigr\}$. For all $\loc\in\Loc_{\sys < \env}$:~
$
\accCond{\loc}=[\loc\geq 1, \locsetgezero{\Loc \setminus \{\loc\}}]
$
\\
\midrule
\textbf{(e)}
For all $\tcmt = (q,\op,q') \in \tcmQ$:\\[1ex]
$\accConde{(\tcmq, \tcmt)}=
\left\{
\parbox{10em}{
\vspace{-2.8ex}
\begin{alignat*}{8}
&[\loclet{\tcmq b}=1, \;&&\loclet{\tcmt}=1, \;&&\loclet{a_i}=1, \;&&\wlocsetgezero{\Locok}],~
&&[\loclet{\tcmq}=1, \;&&\loclet{\tcmt b}=1, \;&&\loclet{a_i}=1, \;&&\wlocsetgezero{\Locok}],\\
&[\loclet{\tcmq}=1, \;&&\loclet{\tcmt}=1, \;&&\loclet{a_ib}=1, \;&&\wlocsetgezero{\Locok}],~
&&[\loclet{\tcmq b}=1, \;&&\loclet{\tcmt b}=1,\, \;&&\loclet{a_i}=1, \;&&\wlocsetgezero{\Locok}],\\
&[\loclet{\tcmq b}=1, \;&&\loclet{\tcmt}=1, \;&&\loclet{a_ib}=1, \;&&\wlocsetgezero{\Locok}],~
&&[\loclet{\tcmq}=1, \;&&\loclet{\tcmt b}=1, \;&&\loclet{a_ib}=1, \;&&\wlocsetgezero{\Locok}]
\end{alignat*}
\vspace{-3.5ex}
}
\right\} \hspace{0.95em}\text{ if } \op = \incc{i}$
\\[5.5ex]
$\accConde{(\tcmq, \tcmt)}=
\left\{
\parbox{10em}{
\vspace{-2.8ex}
\begin{alignat*}{8}
&[\loclet{qb}=1, \;&&\loclet{t}=1, \;&&\loclet{a_i^3b^2}=1, \;&&\wlocsetgezero{\Locok} ],~
&&[\loclet{q}=1, \;&&\loclet{tb}=1, \;&&\loclet{a_i^3b^2}=1, \;&&\wlocsetgezero{\Locok} ],\\
&[\loclet{q}=1, \;&&\loclet{t}=1, \;&&\loclet{a_i^3b^3}=1, \;&&\wlocsetgezero{\Locok} ],~
&&[\loclet{qb}=1, \;&&\loclet{tb}=1, \;&&\loclet{a_i^3b^2}=1, \;&&\wlocsetgezero{\Locok} ],\\
&[\loclet{qb}=1, \;&&\loclet{t}=1, \;&&\loclet{a_i^3b^3}=1, \;&&\wlocsetgezero{\Locok} ],~
&&[\loclet{q}=1, \;&&\loclet{tb}=1, \;&&\loclet{a_i^3b^3}=1, \;&&\wlocsetgezero{\Locok} ]
\end{alignat*}
\vspace{-3.5ex}
}
\right\} \text{ if } \op = \decc{i}$
\\[5.8ex]
$\accConde{(\tcmq, \tcmt)}=
\bigl\{[\loclet{qb}=1,\; \loclet{t}=1,\; \wlocsetgezero{\Locok} ],~ [\loclet{q}=1,\;\loclet{tb}=1,\; \wlocsetgezero{\Locok} ]\bigl\} \hspace{11.5em}\text{ if } \op = \zeroc{i}$\\[1ex]
\midrule
\textbf{(f)}
For all $\tcmt = (q,\op,q') \in \tcmQ$:\\[1ex]
$\accConde{(\tcmq, \tcmt, \tcmq')}=
\left\{
\parbox{10em}{
\vspace{-2.8ex}
\begin{alignat*}{8}
&[\loclet{\tcmq'}=1, \;&&\loclet{\tcmq^2b}=1, \;&&\loclet{\tcmt^2b}\geq 0, \;&&\loclet{a_i^2b}\geq 0, &&\wlocsetgezero{\Locok}] ,\\
&[\loclet{\tcmq'}=1,  \;&&\loclet{\tcmq^2b}\geq 0, \;&&\loclet{\tcmt^2b}=1, \;&&\loclet{a_i^2b}\geq 0,  \;&&\wlocsetgezero{\Locok}] ,\\
&[\loclet{\tcmq'}=1, \;&&\loclet{\tcmq^2b}\geq 0, \;&&\loclet{\tcmt^2b}\geq 0, \;&&\loclet{a_i^2b}=1, \;&&\wlocsetgezero{\Locok}] ,\\
&[\loclet{\tcmq'b}=1, \;&&\loclet{\tcmq^2b}\geq 0, \;&&\loclet{\tcmt^2b}\geq 0, \;&&\loclet{a_i^2b}\geq 0, \;&&\wlocsetgezero{\Locok}]
\end{alignat*}
\vspace{-3.5ex}
}
\right\} \hspace{0.5em}\text{ if } \op = \incc{i}$
\\[8ex]
$\accConde{(\tcmq, \tcmt, \tcmq')}=
\left\{
\parbox{10em}{
\vspace{-2.8ex}
\begin{alignat*}{8}
&[\loclet{\tcmq'}=1,\;&&\loclet{\tcmq^2b}=1,  \;&&\loclet{\tcmt^2b}\geq 0, \;&&\loclet{a_i^4b^3}\geq 0, \;&&\wlocsetgezero{\Locok}] ,\\
&[\loclet{\tcmq'}=1, \;&&\loclet{\tcmq^2b}\geq 0, \;&&\loclet{\tcmt^2b}=1, \;&&\loclet{a_i^4b^3}\geq 0,  \;&&\wlocsetgezero{\Locok}] ,\\
&[\loclet{\tcmq'}=1, \;&&\loclet{\tcmq^2b}\geq 0, \;&&\loclet{\tcmt^2b}\geq 0,\;&&\loclet{a_i^4b^3}=1,\;&&\wlocsetgezero{\Locok}] ,\\
&[\loclet{\tcmq'b}=1, \;&&\loclet{\tcmq^2b}\geq 0, \;&&\loclet{\tcmt^2b}\geq 0, \;&&\loclet{a_i^4b^3}\geq 0, \;&&\wlocsetgezero{\Locok}] 
\end{alignat*}\vspace{-3.5ex}
}
\right\} \text{ if } \op = \decc{i}$
\\[8ex]
$\accConde{(\tcmq, \tcmt, \tcmq')}=
\left\{
\parbox{10em}{
\vspace{-2.8ex}
\begin{alignat*}{8}
&[\loclet{\tcmq'}=1, \;&&\loclet{\tcmq^2b}=1, \;&&\loclet{\tcmt^2b}\geq 0, \;&&\wlocsetgezero{\Locok}] ,\\
&[\loclet{\tcmq'}=1, \;&&\loclet{\tcmq^2b}\geq 0, \;&&\loclet{\tcmt^2b}=1, \;&&\wlocsetgezero{\Locok}] ,\\
&[\loclet{\tcmq'b}=1, \;&&\loclet{\tcmq^2b}\geq 0, \;&&\loclet{\tcmt^2b}\geq 0, \;&&\loclet{a_i^4b^3}\geq 0, \;&&\wlocsetgezero{\Locok}]
\end{alignat*}
\vspace{-3.5ex}
}
\right\} \hspace{0.25em}\text{ if } \op = \zeroc{i}$
\\[5ex]
\bottomrule
\end{tabular}
}
}
\end{center}
\end{table}

First, we force Environment to play only in response to System by making System win as soon as there is a process where Environment has played more letters than System (see Condition (d) in Table~\ref{tab:game00N}).

If $\gamma$ is not halting, the configurations in $\confset(\gamma)$
will not be winning for System. Hence,
System will have to move to win (Condition (a)). 

The first transition chosen by System must start from the initial state of $\tcm$. This is enforced by
Condition (b).

Once System has moved, Environment will move other processes to leave accepting configurations. The only possible move for her is to add $b$ on a process in locations $\loclet{q}$, 
$\loclet{t}$, and $\loclet{a_i}$, if $t$ is a transition incrementing counter $\tcmcounter_i$ (respectively $\loclet{a_i^3b^2}$ if $t$ is a transition decrementing counter $\tcmcounter_i$). All other $\game$-configurations
accessible by Environment from already defined accepting configurations are winning for System, as established in Condition (e).

System can now encode the successor configuration of $\tcm$, according to the chosen transition, by moving a process to the destination state of the transition (see Condition (c)).

Finally, Environment makes the necessary transitions for the configuration to be a valid $\game$-configuration. If she deviates, System wins (see Condition (f)).

If Environment reaches a configuration in $\confset(\gamma)$ for $\gamma\in F$, System can win by moving the process in $\loclet{\tcmfinal}$ to
$\loclet{\tcmfinal^2}$. From there, all the configurations reachable by Environment are also winning for System:
\[\accCond{F}=\bigl\{[\loclet{\tcmfinal^2}= 1, \wlocsetgezero{\Locok}]\;,\; [\loclet{\tcmfinal^2b}= 1, \wlocsetgezero{\Locok}] \;,\;
[\loclet{\tcmfinal^2b^2}= 1, \wlocsetgezero{\Locok}]\bigr\}\,.\]

Finally, the acceptance condition is given by
\begin{displaymath}\Acc=\bigcup_{\loc\in\Loc_{\sys < \env}}\accCond{\loc}\mathrel{\cup}\!\!\!\bigcup_{\tcmt=(\tcminit,\op,q')\in\tcmT}\!\!\!\!\!\!\!\accCond{\tcmt}~\mathrel{\cup}\!\!\!\!\!\!\bigcup_{\tcmt=(q,\op,q')\in\tcmT} \!\!\!\!\!(\accCond{(\tcmq,\tcmt)}\mathrel{\cup}
\accConde{(\tcmq,\tcmt)}\mathrel{\cup} \accCond{(\tcmq,\tcmt,\tcmq')}\cup
\accConde{(\tcmq,\tcmt,\tcmq')})\cup\accCond{F}\,.
\end{displaymath}

Note that a correct play can end in three different ways: either there is a process in  $\loclet{\tcmfinal}$ and System moves it to $\loclet{\tcmfinal^2}$, 
or System has no transition to pick, or there are not enough processes in $\loc_0$ for System to simulate a new transition. Only the first kind is winning for System.

We can show that there is an accepting run in $\tcm$ iff there is some $k$ such that System has a winning $\conf_{(0,0,k)}$-strategy for $\game$
(cf.\ Appendix~\ref{app:undec-00N} for details).
\qed
\end{proof}

\section{Conclusion}\label{sec:conclusion}

There are several questions that we left open and that are interesting in their own right  due to their fundamental character. Moreover, in the decidable cases, it will be worthwhile to provide tight bounds on cutoffs and the algorithmic complexity of the decision problem. A first step would be a direct transformation into a normal form without the detour through the (potentially non-elementary and triply exponential) normal-form constructions due to Schwentick \& Barthelmann and Hanf, respectively.
Like in \cite{BrutschT16,FigueiraP18,khalimov_et_al:concur:2019,exibard_et_al:concur:2019,KhalimovMB18}, our strategies allow the system to have a global view of the whole program run executed so far. However, it is also perfectly natural to consider uniform local strategies where each process only sees its own actions and possibly those that are revealed according to some causal dependencies. This is, e.g., the setting considered in \cite{FinkbeinerO17,beutner_et_al:concur:2019} for a fixed number of processes and in \cite{JacobsB14} for parameterized systems over ring architectures. Another interesting direction would be to look at guarded fragments of FO logic \cite{AndrekaNB98} or extensions of $\wFOdata$ with counting beyond thresholds that also come with normal-form constructions \cite{KuskeS17,kuske_et_al:LIPIcs:2018:9137}.
Finally, we would like to study a parameterized version of the control problem \cite{Muscholl15} where, in addition to a specification, a program in terms of an arena is already given but has to be controlled in a way such that the specification is satisfied.


\bibliographystyle{abbrv}
\bibliography{lit}


\clearpage

\appendix

\section*{Appendix}

\section{Proof of Theorem~\ref{thm:undecFOtwo} ({\normalfont{$\Synthesis{\wFOtwo}{\Zero}{\Zero}{\N}$}})}\label{app:undecFOtwo}
We adapt the proof from \cite{figueira2018playing,FigueiraP18} reducing the halting problem for 2CM.
We call a 2CM $\tcm=(\tcmQ,\tcmT,\tcmcounter_1,\tcmcounter_2,\tcminit,\tcmfinal)$ \emph{deterministic} if, from a given configuration $(\tcmq,\nu_1,\nu_2)$, at most one transition $t$ is fireable. 
It is known that the halting problem is undecidable even for deterministic 2CM
\cite{Minsky67}.

Given a deterministic 2CM $\tcm$, we write a specification formula $\varphi$ such that a satisfying execution encodes a correct run $\tcm$ that ends in a halting configuration. 
Environment has to choose the sequence of transitions of $\tcm$, while System's job is to ensure that Environment does not make an illegal transition. 
The valuation of counter $\tcmcounter_1$ at any point in the run is encoded by the number of different processes  on which only System has executed actions (and not Environment). Conversely, the valuation of $\tcmcounter_2$ is the number of different processes on which only Environment has executed actions (and not System). Both players will cooperate to ensure the valuation is correct with respect to the sequence of transitions taken. 
If a transition increments $\tcmcounter_1$, then first Environment plays on a process on which both System
and herself have already executed an action, then System executes an action on a fresh process so that there is one more process unique to System while the value of $\tcmcounter_2$ stays unchanged. If a transition decrements $\tcmcounter_1$, then Environment executes an action on a process that was unique to System,
 and System replies on the
same process, making one process not unique for System anymore thus decrementing $\tcmcounter_1$ (while $\tcmcounter_2$ is unchanged). 
If a transition tests that $\tcmcounter_1$ is zero, Environment executes an action on a process already shared by both System and Environment,
 and System either plays on the same process if the transition is legal with this valuation (that is $\tcmcounter_1$ is actually zero), otherwise she plays on
 a process that was unique to herself so far (proving that $\tcmcounter_1$ was not zero) and instantly wins the game.
Transitions involving the other counter are encoded in a similar fashion. The formula ensures that Environment starts the execution, 
and that System and Environment
alternate their actions until a halting configuration is reached. The actions of Environment are the different transitions taken along the simulated run of 
$\tcm$. 
The objective of System is that a halting configuration is reached, while the objective of Environment is that the run ends before reaching
such a configuration (because there are no more fresh processes anymore), or that the run continues forever without reaching a halting configuration. One can show that there exists a halting run of $\tcm$
if and only if there is some $k\in\N$ such that there exists a winning $\conf_{(0,0,k)}$-strategy for System for $\varphi$.
The alphabet is partitioned in $\Alpha_\env = \{s\} \cup \{t~|~t$ is a transition of the 2CM$\}$ and $\Alpha_\sys = \{a\}$. 
Let us fix $x$ and $y$ the two variables used. The formula will check that at every position of the run, two consecutive transitions played by the environment
can actually be taken consecutively, with respect to their starting and ending states. Moreover, one has to make sure that the first transition played by Environment can be
taken from the initial state. The first two positions will be dummies actions, to ensure that Environment and System share at least one process, 
while the simulation of the actual 2CM execution will start from the third position.
We use shorthands $\first(x)$, $\second(x)$, and $\third(x)$ for the $\fod(\succ)$-definable formulas that say that $x$ is the first, second, and third position in the $\Procs$-execution respectively.

We define the following useful formula:
\[\old_\theta(x) \equiv \exists y. (y < x \land y \dataeq x \land \bigvee_{b \in \Alpha_\theta} b(x))\] 
for $\theta \in \{\sys,\env\}$ that says that the value at position $x$ was already seen at an anterior position of player $\theta$. 
The specification will force Environment to play first, and then System and Environment to play in turn until
eventually a halting configuration is reached. We give first the set of constraints $\Phi^\env$ that Environment must satisfy, then the set of constraints related to 
System. 

The following formulas make up the set $\Phi^\env$:
\begin{itemize}
\item Environment does not play twice in a row:
\[\forall x.\forall y.(\bigvee_{b\in\Alpha_\env} b(x) \wedge +1(x,y))\implies \neg \bigvee_{b\in\Alpha_\env}b(y)\]
\item Environment always executes an action when it is its turn, unless the halting configuration is reached. We let $\tcmT_h$ be the set of transitions in $\tcm$ whose
ending state is $\tcmfinal$:
\[\forall x. \bigvee_{b\in\Alpha_\env\setminus \tcmT_h} b(x) 
\implies \exists y. (x<y \wedge \bigvee_{b\in\Alpha_\env} b(y))\]
\item Environment starts with an $s$:
\[\exists x.\first(x)\wedge s(x)\]
\item There is an $s$ only in the first position:
\[\forall x. s(x)\implies  \first(x) \]
\item The first transition is an initial transition (i.e starts from an initial state):
\[\forall x. \third(x) \Rightarrow \bigvee_{t\text{ is initial}} t(x)\]
\item Consecutive transitions are compatible (i.e the ending state of the $n$-th one is the starting state of the $n+1$-th):
\[\forall x. \bigwedge_{t} \left(t(x) \Rightarrow  \forall y. (\succ(x,y) \Rightarrow \forall x. (\succ(y,x) \Rightarrow \bigvee_{t' \text{ compatible with } t}t'(x)))\right)\]
\item If $t$ increments $c_1$, Environment plays a on a process already shared by System and Environment:
\[\forall x. \bigwedge_{t \text{ increments $c_1$}} \left(t(x) \Rightarrow \old_\env(x) \land \old_\sys(x)\right))\]
\item If $t$ decrements $c_1$, Environment plays on a process that was unique to System:
\[\forall x. \bigwedge_{t \text{ decrements $c_1$}} \left(t(x) \Rightarrow \neg \old_\env(x) \land \old_\sys(x)\right)\]
\item If $t$ increments $c_2$, Environment must play on a fresh process:
\[\forall x. \bigwedge_{t \text{ increments $c_2$}} \left(t(x) \Rightarrow \neg \old_\env(x) \land \neg \old_\sys(x)\right)\]
\item If $t$ decrements $c_2$, Environment must play on a process that was unique to herself:
\[\forall x. \bigwedge_{t \text{ decrements $c_2$}} \left(t(x) \Rightarrow \old_\env(x) \land \neg \old_\sys(x)\right)\]

\item If $t$ checks that $c_1$ is zero, Environment plays a shared value and System does not reply with a value unique to herself:
\[\forall x. \bigwedge_{t \text{ zero-tests $c_1$}} (t(x) \Rightarrow \old_\sys(x) \land \old_\env(x) \land \forall y. (\succ(x,y) \Rightarrow \neg \old_\sys(y) \lor \old_\env(y)))\]
\item If $t$ checks that $c_2$ is zero, Environment plays a shared value and System does not reply with a value unique to Environment:
\[\forall x. \bigwedge_{t \text{ zero-tests $c_2$}} (t(x) \Rightarrow \old_\sys(x) \land \old_\env(x) \land \forall y. (\succ(x,y) \Rightarrow \old_\sys(y) \lor \neg \old_\env(y)))\]
\end{itemize}
And now we construct the set $\Phi^\sys$ of environment constraints:
\begin{itemize}
\item System does not play twice in a row:
\[\forall x.\forall y.(a(x) \wedge +1(x,y))\implies \neg a(y)\]
\item The first move must be played on the same process than Environment:
\[\forall x. \second(x) \Rightarrow \exists y. (\succ(y,x) \land x \dataeq y)\]
\item If $t$ increments $c_1$, System must reply on a fresh process:
\[\forall x. \bigwedge_{t \text{ increments $c_1$}} ((t(x) \Rightarrow \exists y. (\succ(x,y) \land \neg \old_\sys(y) \land \neg \old_\env(y)))\]
\item If $t$ increments $c_2$, System replies on an already shared process:
\[\forall x. \bigwedge_{t \text{ increments $c_2$}} (t(x) \Rightarrow \exists y. (\succ(x,y) \land \neg ( x\dataeq y)\land \old_\sys(y) \land \old_\env(y)))\]
\item If $t$ decrements either counter, System replies on the same process:
\[\forall x. \bigwedge_{t \text{ decrements $c_1$ or $c_2$}} (t(x) \Rightarrow \exists y. (\succ(x,y) \land y \dataeq x))\]
\item If $t$ zero-tests either counter, System replies on the same process 
\[\forall x. \bigwedge_{t \text{ zero-tests $\tcmcounter_1$ or $\tcmcounter_2$}} (t(x) \Rightarrow \exists y. (\succ(x,y) \land (y \dataeq x )))\]

\end{itemize}

Put together, this gives the formula $\varphi = \Phi^\env\implies (\Phi^\sys\wedge\exists x. \bigvee_{t \in \tcmT_h} t(x)) $.
Consider the strategy $\strat$ for System such that $\strat(s,i)=(a,i)$ and for each prefix $w\cdot(\tcmt,i)$, $\strat(w\cdot(\tcmt,i))=(a,j)$ with $j$ a process such that, for all event $(x,k)$ in $w$, $k\neq j$ if 
$\tcmt$ is a transition incrementing $\tcmcounter_1$, $j=1$ if $\tcmt$ increments $\tcmcounter_2$, $j$ a process such that $w=w'\cdot(a,j)\cdot w''$ and for all event
$(x,k)$ appearing in $w'$, $w''$, if $x\in\Alpha_\env$, $k\neq j$ if $\tcmt$ zero-tests $\tcmcounter_1$ and if such a $j$ exists, $j$ a process such that $w=w'\cdot(\tcmt',j)\cdot w''$ and for all event
$(a,k)$ appearing in $w'$, $w''$, $k\neq j$ if $\tcmt$ zero-tests $\tcmcounter_1$ and if such a $j$ exists, $j=i$ otherwise.
For all other prefixes, $\strat$ returns $\varepsilon$.

Consider an execution $\strat$-compatible. Then if $\Phi^\env$ is violated, it means that either Environment has not respected the encoding, 
or that it has played $\tcmt$ that zero-tests $\tcmcounter_1$ or $\tcmcounter_2$. In that case, by $\strat$, System replies with a value unique to herself, or unique 
to Environment respectively, proving that Environment has chosen a non-valid transition for the current configuration. If $\Phi^\env$ is satisfied, then by
$\strat$,  either $\Phi^\sys$ is satisfied, or at some point System is unable to provide a fresh process when $\tcmt$ increments $\tcmcounter_1$. 

Suppose that, for some $k\in\N$, $\strat$ is $(0,0,k)$-winning.  Then, for any execution $w$ that is $\strat$-compatible, if $\Phi^\env$ is satisfied,
then $\Phi^\sys$ is satisfied, meaning that the sequence of transitions chosen by Environment correctly encodes the run of $\tcm$. Moreover, at some
point a final transition is reached, encoding the halting run of $\tcm$.

Conversely, suppose that $\tcm$ is halting. Then, the run being finite, we can compute $k_\sys$ and $k_\env$ the number of increments of respectively
$\tcmcounter_1$ and $\tcmcounter_2$ during the run.
Consider $\strat$ as a $\conf_{(0,0,k_\env+k_\sys)}$-strategy, and take $w$ an $\strat$-compatible execution that satisfies $\Phi^\env$. Necessarily, 
the sequence of transitions selected in $w$ is the unique possible sequence of transitions in an execution of $\tcm$, and the number of processes
is enough to provide as many fresh processes as needed. Hence the run can be encoded entirely and a final transition is reached. Thus, $w$ satisfies the
specification formula, and $\strat$ is $(0,0,k_\env+k_\sys)$-winning.
Therefore the synthesis problem for $\fod[\dataeq,\succ,<]$ is undecidable.

\section{Proof of Theorem~\ref{thm:normalform} (normal form for $\wFOdata$)}\label{app:normalform}

Let $\Phi$ be an $\FOdata$ formula. Using the Schwentick-Barthelmann normal
form \cite{schwentick1998local}, we know that $\Phi$ is equivalent to a formula of the form
\[\Phi_1 = \exists x_1 \dots \exists x_n \forall y. \varphi(x_1, \dots, x_n,y)\]
where, in $\phi(x_1,\ldots,x_n,y)$, quantification is always of the form
$\exists z. (z \sim y \wedge \ldots)$ or
$\forall z. (z \sim y \Longrightarrow \ldots)$.
Since $\phi$ essentially talks about the class of $y$,
we call it a \emph{class formula} (wrt.\ $y$).
Let $\X = \{x_1,\ldots,x_n\}$.
Wlog., we assume that none of the variables in
$\X \cup \{y\}$ is quantified in $\phi$.

\paragraph{\textup{\textbf{Class Abstraction.}}}

Note that, due to the variables in $\X$, the formula $\varphi$
may reason about elements that are \emph{outside} the class of $y$.
Our aim is to get rid of these variables so as to end up with
formulas that talk about classes only.
As the variables in $\X$ are quantified existentially, we can
basically guess the relation between them.
This is done in terms of a \emph{class abstraction},
which  is given by a triple
$\Class = (\Part,\approx,\lambda)$ where
$\Part \subseteq 2^\X$ is a partition of $\X$ (if $n=0$, then $\Part = \emptyset$),
$\lambda: \X \to \Alpha \uplus \Types$
(recall that $\Types = \{\sys,\env,\both\}$), and
${\approx}$ is an equivalence relation over $\X$ such that,
for all $x_i,x_j \in \X$
\begin{itemize}\itemsep=1ex
\item if $x_i \approx x_j$, then $x_i \sim_\Class x_j$ and $\lambda(x_i) = \lambda(x_j)$, and

\item if $x_i \sim_\Class x_j$ and $x_i \not\approx x_j$, then $\{\lambda(x_i),\lambda(x_j)\} \cap \Alpha \neq \emptyset$,
\end{itemize}
where we write $x_i \sim_\Class x_j$ if $\{x_i,x_j\} \subseteq X$ for some $X \in \Part$.
Let $\Abstractions{\X}$ be the set of all class abstractions.

\begin{example}
Figure~\ref{fig:classabstr} depicts a class abstraction
$\Class = (\Part,\approx,\lambda)$ for $\Alpha = \{a,b,c,d\}$.
The red areas represent the partition $P$, the blue ones
represent the equivalence classes of $\approx$, which refine
$P$. Moreover, we have $\lambda(x_4) = \lambda(x_7) = \sys$,
$\lambda(x_3) = \lambda(x_5) = b$, $\lambda(x_9) = c$, etc.
The meaning of $\Class$ is that
$x_1, x_4, x_7, x_8$ are equivalent wrt.\ $\sim$, i.e., they
belong to the same process. In particular, formulas
such as $x_1 \sim x_4$ and $x_3 \sim x_{10}$ are true under this
assumption. Moreover, $x_4 \approx x_7$ means that $x_4$ and $x_7$ denote
\emph{identical} elements. That is, the formula $x_4 = x_7$
would be true, whereas $x_2 = x_3$ does not hold.
\exend
\end{example}

\begin{figure}[t]
\centering
\includegraphics[width=0.6\textwidth]{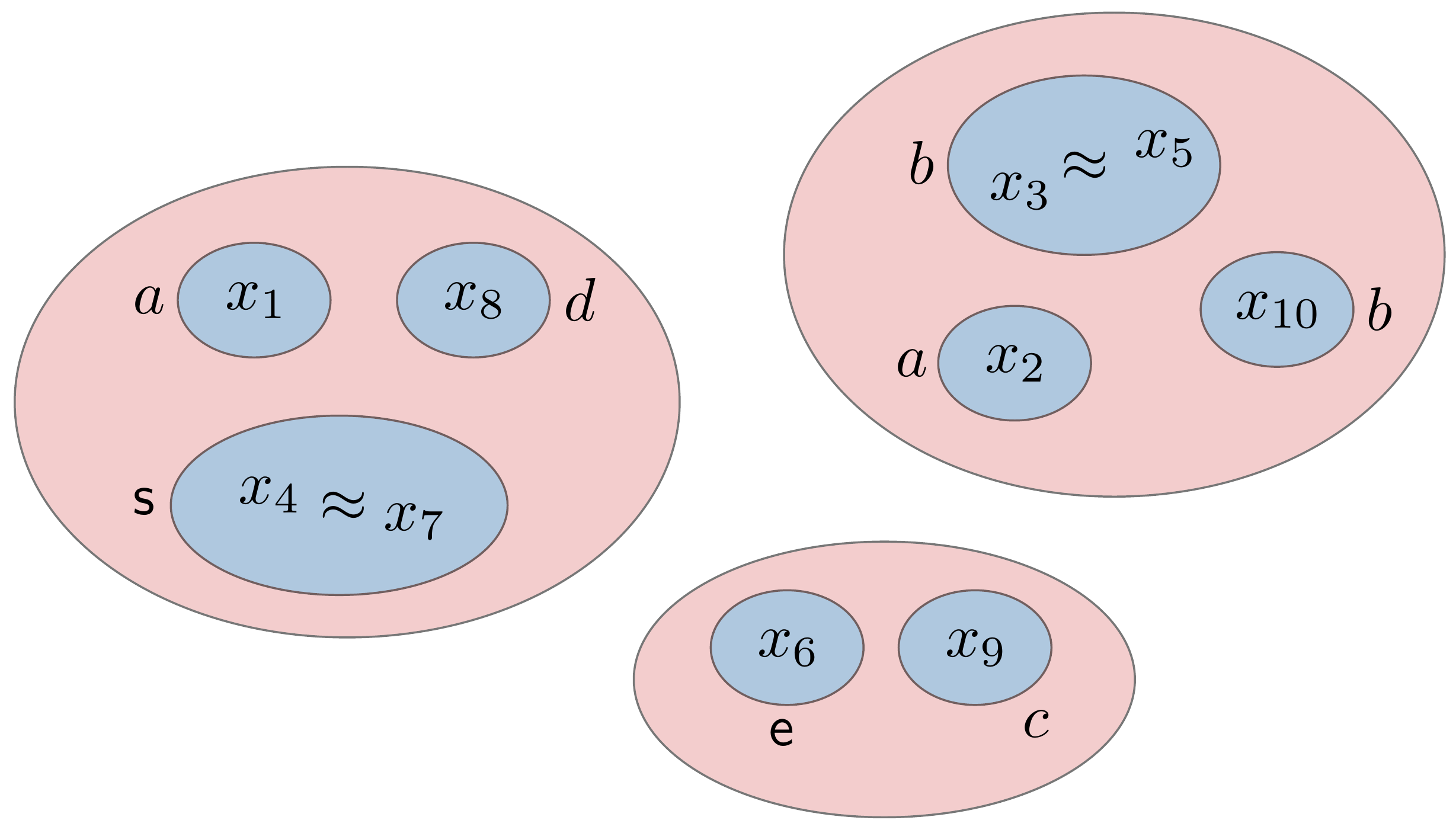}
\caption{A class abstraction $\Class = (\Part,\approx,\lambda)$\label{fig:classabstr}}
\end{figure}

\newcommand{\cmodels}[2]{\mathsf{sat}(#1,#2)}

Given $\Class = (\Part,\approx,\lambda)$ and $X \subseteq \X$, we can define the formula $\cmodels{X}{\Class}$,
which checks whether the class abstraction $\Class$ is consistent with a given
execution as far as variables from $\class$ are concerned:
\[\cmodels{X}{\Class} ~~=~~
\left(
\begin{array}{rl}
& \bigwedge_{x_i \in X} (\lambda(x_i))(x_i)\\[1ex]
\wedge &
\displaystyle\bigwedge_{(x_i,x_j) \in {\approx} \cap X^2} x_i = x_j \wedge
\bigwedge_{(x_i,x_j) \in X^2 \setminus{\approx}} x_i \neq x_j\\[3ex]
\wedge &
\displaystyle\bigwedge_{(x_i,x_j) \in {\sim_\Class} \cap X^2} x_i \sim x_j \wedge
\bigwedge_{(x_i,x_j) \in X^2 \setminus{\sim_\Class}} x_i \not\sim x_j
\end{array}
\right)
\]

Moreover, by fixing $\Class = (\Part,\approx,\lambda)$ and $X \in P \cup \{\emptyset\}$,
we can transform $\phi$ into a class formula (wrt.\ $y$)
\[\ftrans{\phi}{\Class}{\class}((x_i)_{x_i \in X},y)\]
without variables from $\X \setminus \class$
that ``evaluates'' $\phi$ based on the assumption
that $\cmodels{\X}{\Class} \wedge (y \sim \class)$ holds
(in particular, $\cmodels{\X}{\Class} \wedge \neg (y \sim \X)$ if $\class = \emptyset$)
where $y \sim \class$ is a shorthand for $\bigvee_{x_i \in X} y \sim x_i$.
We obtain it from $\phi$ inductively as follows (let $\theta \in \Alpha \cup \Types$):
\begin{align*}
\ftrans{\theta(z)}{\Class}{\class} &=
\begin{cases}
\theta(z) & \textup{if } z \in \PosVar \setminus \X\\ 
\ttrue & \textup{if } z \in \X \textup{ and } \lambda(z) = \theta\\
\ffalse & \textup{if } z \in \X \textup{ and } \lambda(z) \neq \theta
\end{cases}\\
\ftrans{z = z'}{\Class}{\class} &=
\begin{cases}
z = z' & \textup{if } \{z,z'\} \subseteq \PosVar \setminus (\X \setminus X)\\
\ttrue & \textup{if } \{z,z'\} \subseteq \X \setminus X \textup{ and } z \approx z'\\
\ffalse & \textup{otherwise}
\end{cases}\\
\ftrans{z \sim z'}{\Class}{\class} &=
\begin{cases}
z \sim z' & \textup{if } \{z,z'\} \subseteq \PosVar \setminus (\X \setminus X)\\
\ttrue & \textup{if } \{z,z'\} \subseteq \X \setminus X \textup{ and } z \sim_\Class z'\\
\ffalse & \textup{otherwise}
\end{cases}\\
\ftrans{\psi \vee \psi'}{\Class}{\class} &= \ftrans{\psi}{\Class}{\class} \vee \ftrans{\psi'}{\Class}{\class}\\
\ftrans{\neg\psi}{\Class}{\class} &= \neg\ftrans{\psi}{\Class}{\class}\\
\ftrans{\exists z. \psi}{\Class}{\class} &= \exists z. \ftrans{\psi}{\Class}{\class}
\end{align*}

\paragraph{\textup{\textbf{The Transformation.}}}

Given the above definitions, we can now rephrase $\Phi_1$ as follows.
Along with $x_1,\ldots,x_n$, we also guess a class abstraction, 
which then allows us to reason about each class separately,
without looking at the variables outside a class:
\[
\Phi_2 =
\bigvee_{\Class = (\Part,\approx,\lambda) \in \Abstractions{\X}}
\exists x_1 \dots \exists x_n.
\left(
\begin{array}{rl}
& \cmodels{\X}{\Class}
\hfill ~~(\subf_1)
\\[1ex]
\wedge &
\displaystyle\bigwedge_{\class \in \Part}
\forall y.
\Bigl(
y \sim \class
\implies
\ftrans{\phi}{\Class}{\class}((x_i)_{x_i \in X},y)
\Bigr)
\hfill ~~(\subf_2)
\\[3ex]
\wedge &
\forall y.
\Bigl(
\neg(y \sim \X)
\implies
\ftrans{\phi}{\Class}{\emptyset}(y)
\Bigr)
\hfill ~~(\subf_3)
\end{array}
\right).
\]

In fact, we can push the quantifiers $\exists x_1 \ldots \exists x_n$
further inwards by replacing them with $\exists (z_\class)_{\class \in \Part}$,
which chooses one canonical representative $z_\class$ per class~$\class$:
\[
\Phi_3 = \!\!\!
\bigvee_{\Class = (\Part,\approx,\lambda) \in \Abstractions{\X}}
\!\!\!
\exists (z_\class)_{\class \in \Part}.
\left(
\begin{array}{rl}
& \displaystyle
\bigwedge_{\class \in P} \procform(z_\class)
\wedge \bigwedge_{\substack{\class,Y \in P\\\class \neq Y}} z_\class \not\sim z_Y
\hfill ~~(\subf_4)
\\[4ex]
\wedge &
\displaystyle\bigwedge_{\class \in \Part}
\phi_{\Class,\class}(z_\class)
\hfill ~~(\subf_5)
\\[3ex]
\wedge &
\forall z_\emptyset.
\Bigl(
\Bigl(
\procform(z_\emptyset)
\wedge
\displaystyle\bigwedge_{\class \in \Part} \neg(z_\emptyset = z_\class)
\Bigr)
\implies
\phi_{\Class,\emptyset}(z_\emptyset)
\Bigr)
\hfill ~~(\subf_6)
\end{array}
\right)
\]
where
$\procform(z_\class) = \bigvee_{\theta \in \Types} \theta(z_\class)$
and
\begin{align*}
\phi_{\Class,\class}(z_\class) ~=~&
\exists (x_i)_{x_i \in \class}.
\left(
\begin{array}{rl}
& z_\class \sim \class ~\wedge~ \cmodels{X}{\Class}\\[0.5ex]
\wedge & \forall y.\bigl(y \sim z_\class ~\Longrightarrow~ \ftrans{\phi}{\Class}{\class}((x_i)_{x_i \in X},y)\bigr)
\end{array}
\right)
\\[1ex]
\phi_{\Class,\class}(z_\emptyset) ~=~&
\forall y.\bigl(y \sim z_\emptyset ~\Longrightarrow~ \ftrans{\phi}{\Class}{\emptyset}(y)\bigr)
\end{align*}
are \emph{class formulas}\footnote{in fact, they can be easily rewritten as a class formula} wrt.\ $z_\class$.

\begin{claim}
Formulas $\Phi_2$ and $\Phi_3$ are logically equivalent.
\end{claim}

\begin{proof}
Suppose $(\Procs,w) \models \Phi_2$,
say, witnessed by class abstraction $\Class$ and
interpretation $\Inter_\X = \{x_1 \mapsto e_1,\ldots,x_n \mapsto e_n\}$.
That is, $(\Procs,w),\Inter_\X \models \subf_1 \wedge \subf_2 \wedge \subf_3$.
Let $\Inter_\class$ denote the restriction of $\Inter_\X$ to $\class \in P$.
For $\class \in P$, consider the unique process $e_\class \in \Procs$
such that $e_\class \sim e_i$ for some $x_i \in \class$.
Consider the interpretation $\Inter_\mathit{repr} = \{z_\class \mapsto e_\class \mid \class \in P\}$.
Let us show $(\Procs,w),\Inter_\mathit{repr} \models \subf_4 \wedge \subf_5 \wedge \subf_6$.

\begin{itemize}\itemsep=0.5ex
\item[($\subf_4$)] Clearly, we have $(\Procs,w),\Inter_\mathit{repr} \models \subf_4$.

\item[($\subf_5$)] Let $\class \in P$. We have $(\Procs,w),\Inter_\class \models z_\class \sim \class ~\wedge~ \cmodels{X}{\Class}$.
Take any $e \in \Procs \mathrel{\cup} \Pos{w}$ such that $e \sim e_\class$.
By satisfaction of $\subf_2$, we have
$(\Procs,w),\Inter_\class,\{y \mapsto e\} \models \ftrans{\phi}{\Class}{\class}((x_i)_{x_i \in X},y)$.
Therefore, $(\Procs,w),\Inter_\mathit{repr} \models \subf_5$. 

\item[($\subf_6$)]
Let $e_\emptyset \in \Procs$ such that
$e_\emptyset \neq e_\class$ for all $\class \in P$.
Moreover, let $e \in \Procs \cup \Pos{w}$ such that
$e \sim e_\emptyset$. Then, $e \not\sim e_\class$ for all $\class \in P$
so that, by satisfaction of $\Phi_2$,
we have $(\Procs,w),\{y \mapsto e\} \models \ftrans{\phi}{\Class}{\emptyset}(y)$.
We obtain $(\Procs,w),\Inter_\mathit{repr} \models \subf_6$.
\end{itemize}
We conclude that $(\Procs,w) \models \Phi_3$.

\medskip

Conversely, suppose that $(\Procs,w) \models \Phi_3$, witnessed
by $\Class$ and an interpretation $\Inter_\mathit{repr} = \{z_\class \mapsto e_\class\}_{\class \in P}$.
That is, $(\Procs,w),\Inter_\mathit{repr} \models \subf_4 \wedge \subf_5 \wedge \subf_6$.
As $(\Procs,w),\{z_\class \mapsto e_\class\} \models \phi_{\Class,\class}(z_\class)$ for every $\class \in P$,
there is $\Inter_\X = \{x_1 \mapsto e_1,\ldots,x_n \mapsto e_n\}$
(with restrictions $\Inter_\class$) such that
\[
(\Procs,w),\{z_\class \mapsto e_\class\},\Inter_\class \models
\left(
\begin{array}{rl}
& z_\class \sim \class ~\wedge~ \cmodels{X}{\Class}\\[0.5ex]
\wedge & \forall y.\bigl(y \sim z_\class ~\Longrightarrow~ \ftrans{\phi}{\Class}{\class}((x_i)_{x_i \in X},y)\bigr)
\end{array}
\right)\]
for every $\class \in P$.
Let us show $(\Procs,w),\Inter_\X \models \subf_1 \wedge \subf_2 \wedge \subf_3$.

\begin{itemize}\itemsep=0.5ex
\item[($\subf_1$)]
As $e_i \not\sim e_j$ whenever $e_i \in X$ and $e_i \in Y$
for distinct sets $X,Y \in P$,
we get $(\Procs,w),\Inter_\X \models \cmodels{\X}{\Class}$.

\item[($\subf_2$)]
Take any $\class \in P$ and $e \in \Procs \mathrel{\cup} \Pos{w}$
such that $e \sim e_i$ for some $i \in \{1,\ldots,n\}$ such that $x_i \in \class$.
As $e \sim e_\class$,
we obtain $(\Procs,w),\Inter_\class,\{y \mapsto e\}\models \ftrans{\phi}{\Class}{\class}((x_i)_{x_i \in X},y)$.

\item[($\subf_3$)]
Let $e \in \Procs \cup \Pos{w}$ such that
$e \not\sim e_i$ for all $i \in \{1,\ldots,n\}$.
Let $e_\emptyset \in \Procs$ such that $e_\emptyset \sim e$.
That is, $e_\emptyset \not\sim e_\class$ for all $\class \in P$.
We get $(\Procs,w),\{z_\emptyset \mapsto e_\emptyset\} \models
\phi_{\Class,\emptyset}(z_\emptyset)$.
Since $e \sim e_\emptyset$, we also have that
$(\Procs,w),\{y \mapsto e\} \models \phi_{\Class,\emptyset}(y)$.
Therefore, $(\Procs,w),\Inter_\X \models \subf_3$
\end{itemize}
We conclude that $(\Procs,w) \models \Phi_2$.
\qed
\end{proof}

The formulas $\phi_{\Class,\class}(z_\class)$, including the case $X = \emptyset$,
are interesting, because they only reason about the class determined by
$z_\class$. As, wrt.\ $\sim$, any two elements from a class are equivalent anyway,
we can actually ignore $\sim$. A class can then be seen as a simple multiset,
or as a logical structure of degree $0$ (there is no binary relation that connects two
elements from a class). By Hanf's theorem \cite{Hanf1965,BolligK12}, we can find $B \in \N$ such that every formula
$\phi_{\Class,\class}(z_\class)$, including the case $X = \emptyset$,
is equivalent to a formula
\[
\phi_{\Class,\class}(z_\class) ~\equiv \bigvee_{(\theta,\loc) \in V_{\Class,\class}} \bigl(\theta(z_\class) \wedge \ctype{\bound,\loc}(z_\class)\bigr)
\]
for some sets $V_{\Class,\class} \subseteq \Types \times \{0,\ldots,B\}^\Alpha$.
Note that, for $V_{\Class,\class} = \emptyset$, we get $\ffalse$.
Recall that we had defined:
\[
\ctype{\bound,\loc}(y) ~=~
\bigwedge_{\substack{a \in \Alpha\\[0.3ex]\loc(a) < \bound}} \exists^{=\loc(a)} z.\bigl(y \sim z \wedge a(z)\bigr)
\wedge
\bigwedge_{\substack{a \in \Alpha\\[0.3ex]\loc(a) = \bound}} \exists^{\ge \loc(a)} z.\bigl(y \sim z \wedge a(z)\bigr)
\]

Thus, $\Phi_3$ is equivalent to the following formula
(note that the conjunct $\bigwedge_{\class \in P} \procform(z_\class)$ is not needed anymore, as its
satisfaction is guaranteed
by the second line; other changes wrt.\ $\Phi_3$ are highlighted in \textcolor{red}{red}):
\[
\Phi_4 =
\bigvee_{\Class = (\Part,\approx,\lambda) \in \Abstractions{\X}}
\exists (z_\class)_{\class \in \Part}.
\left(
\begin{array}{rl}
& \displaystyle
\bigwedge_{\substack{\class,Y \in P\\\class \neq Y}} z_\class \not\sim z_Y\\[5ex]
\wedge &
\displaystyle\bigwedge_{\class \in \Part}\;
\textcolor{red}{\bigvee_{(\theta,\loc) \in V_{\Class,\class}} \bigl(\theta(z_\class) \wedge \ctype{\bound,\loc}(z_\class)\bigr)}
\\[4ex]
\wedge &
\forall z_\emptyset.
\Bigl(
\Bigl(
\procform(z_\emptyset)
\wedge
\displaystyle\bigwedge_{\class \in \Part} \neg(z_\emptyset = z_\class)
\Bigr)\\[3ex]
&~~~~~~~~~~~~\implies
\displaystyle\textcolor{red}{\bigvee_{(\theta,\loc) \in V_{\Class,\emptyset}} \bigl(\theta(z_\emptyset) \wedge \ctype{\bound,\loc}(z_\emptyset)\bigr)}
\Bigr)
\end{array}
\right)
\]

Expanding the expression, we obtain that $\Phi_4$ is equivalent to:
\[
\Phi_5 =
\bigvee_{\substack{\Class = (\Part,\approx,\lambda) \in \Abstractions{\X}\\[0.8ex]
\textcolor{red}{((\theta_\class,\loc_\class))_{\class \in \Part}}\\[0.2ex]
\textcolor{red}{\;\in\; \prod_{\class \in \Part} V_{\Class,\class}}
}
}
\hspace{-0.5em}
\exists (z_\class)_{\class \in \Part}.
\left(
\begin{array}{rl}
& \displaystyle\bigwedge_{\substack{\class,Y \in P\\\class \neq Y}} z_\class \not\sim z_Y\\[5ex]
\wedge &
\displaystyle\bigwedge_{\class \in \Part}
\textcolor{red}{\bigl(\theta_\class(z_\class) \wedge \ctype{\bound,\loc}(z_\class)\bigr)}
\\[3ex]
\wedge &
\forall z_\emptyset.
\Bigl(
\Bigl(
\procform(z_\emptyset)
\wedge
\displaystyle\bigwedge_{\class \in \Part} \neg(z_\emptyset = z_\class)
\Bigr)\\[3ex]
&~~~~~~~~~~~~\implies
\displaystyle\bigvee_{(\theta,\loc) \in V_{\Class,\emptyset}} \bigl(\theta(z_\emptyset) \wedge \ctype{\bound,\loc}(z_\emptyset)\bigr)
\Bigr)
\end{array}
\right)
\]

Finally, $\Phi_5$ is equivalent to a formula of the desired form:
\[
\Phi_6 =
\bigvee_{\substack{\Class = (\Part,\approx,\lambda) \in \Abstractions{\X}\\[0.8ex]
\vect = ((\theta_\class,\loc_\class))_{\class \in \Part}\\[0.2ex]
\;\in\; \prod_{\class \in \Part} V_{\Class,\class}
}
}
\left(
\begin{array}{rl}
& \displaystyle \bigwedge_{(\theta,\loc) \textcolor{blue}{\,\in\, V_{\Class,\emptyset}}}
\exists^{\ge |\vect|_{(\theta,\loc)}} y.
\bigl(\theta(y) \wedge \ctype{\bound,\loc}(y)\bigr)
\\[6ex]
\wedge &
 \displaystyle \bigwedge_{\substack{(\theta,\loc)\\[0.2ex] \,\in\, (\Types \times \{0,\ldots,B\}^\Alpha) \textcolor{blue}{\,\setminus\, V_{\Class,\emptyset}}}}
\hspace{-2em}
\exists^{= |\vect|_{(\theta,\loc)}} y.
\bigl(\theta(y) \wedge \ctype{\bound,\loc}(y)\bigr)
\end{array}
\right)
\]
where $|\vect|_{(\theta,\loc)}$ is the number of occurrences of $(\theta,\loc)$
in $\vect = ((\theta_\class,\loc_\class))_{\class \in \Part}$, i.e., \[|\vect|_{(\theta,\loc)} = |\{\class \in P \mid (\theta,\loc) = (\theta_\class,\loc_\class)\}|\,.\]

\section{Proof of Lemma~\ref{lemma:equifogame} (equivalence of synthesis and parameterized vector games)}
\label{app:equifogame}

As an intermediate step in the translation of the synthesis problem into games,
we first consider a \emph{normalized} version of the former.
In a second step, we show equivalence between the normalized synthesis
problem and games.

\subsection{Normalized Synthesis Problem for $\wFOdata$}

In the normalized synthesis problem, instead of being fully asynchronous, both players will alternately give a sequence of events instead of a single one. Moreover, after every move from System, the partial word created up to that point should satisfy the formula, whereas after every move from Environment, the word should falsify the formula.

Let us fix, for the rest of the definitions, a sentence $\varphi \in \FOdata$.
We call a finite $\Procs$-execution $w \in \DS^\ast$ \emph{normalized}
if it is of the form $w = w_1 \dots w_n$ with $n \ge 1$ such that
\begin{itemize}[itemsep=0.5ex]
\item for all odd $i$ such that $1 \le i \le n$, $w_i \in \DS_\sys^\ast$ and $(\Procs,w_1 \dots w_i) \models \varphi$,
\item for all even $i$ such that $1 \le i \le n$, $w_i \in \DS_\env^\ast$ and $(\Procs,w_1 \dots w_i) \not\models \varphi$.
\end{itemize}
Note that the decomposition into the $w_i$, if it exists, is uniquely determined.

A \emph{normalized} $\Procs$-strategy (for System) is a partial mapping $\strat: \DS^\ast \to \DS_\sys^\ast$ such that $\strat(\varepsilon)$ is defined and,
if $\strat(w)$ is defined, then $(\Procs,w \cdot \strat(w)) \models \varphi$.
A normalized $\Procs$-execution $w = w_1 \dots w_n$ is
\begin{itemize}\itemsep=0.5ex
\item $\strat$\emph{-compatible} if, for all odd $1 \le i \le n$, we have $w_i = \strat(w_1 \dots w_{i-1})$,
\item $\strat$\emph{-maximal} if it is not the strict prefix of an $\strat$-compatible normalized $\Procs$-execution,
\item \emph{winning} if $(\Procs, w) \models \varphi$.
\end{itemize}
Finally, a normalized strategy is $\Procs$-\emph{winning} if all $\strat$-compatible
$\strat$-maximal normalized $\Procs$-executions are winning.

Similarly to the initial synthesis problem, we define the normalized winning set
$\NWin{\Alpha_\sys}{\Alpha_\env}{\varphi}$ as the set of
triples $(\pconsts,\pconsta,\pconstb) \in \N^\Types$
for which there is $\Procs = (\sProcs,\eProcs,\seProcs)$ such that
\begin{itemize}\itemsep=0.5ex
\item $|\Procs_\theta| = \pconst_\theta$ for all $\theta \in \Types$, and
\item there is a normalized $\Procs$-strategy that is $\Procs$-winning.
\end{itemize}

\begin{example}
Consider $\varphi_4 = \forall x. \bigl(\bigl(\exists^{= 2} y.(x \sim y \wedge a(y))\bigr)
\Longleftrightarrow \bigl(\exists^{= 2} y.(x \sim y \wedge d(y))\bigr)\bigr) \in \FOdata$ from Example~\ref{ex:formulas},
where $\Alpha_\sys = \{a,b\}$ and $\Alpha_\env = \{c,d\}$.
Let $\strat$ be the strategy defined for $\varphi_4$ (and $\varphi_3$)
in Example~\ref{ex:synthesis}. We can consider it as a strategy $\strat: \DS^\ast \to \DS_\sys^\ast$.
However, $\strat$ is not normalized:
For $\Procs=(\{1,2,3\},\emptyset,\{6,7,8\})$, $\strat((d,7)(d,7)) = (a,7)$, but
$(\Procs,(d,7)(d,7)(a,7)) \not\models \phi_4$.
Consider any $\Procs$ such that $\eProcs = \emptyset$.
Towards a winning normalized $\Procs$-strategy $\strat_\mathsf{norm}$,
suppose $w = (a_1,p_1) \ldots (a_n,p_n)\in \DS^\ast$.
For $p \in P = \{p_1,\ldots,p_n\}$,
we let $\textup{diff}(w,p)$ denote $|\{i \in \{1,\ldots,n\} \mid (a_i,p_i) =
(d,p)\}| - |\{i \in \{1,\ldots,n\} \mid (a_i,p_i) = (a,p)\}|$.
With this, we set
\[
\strat_\mathsf{norm}(w) =
\begin{cases}
\prod_{p \in P} (a,p)^{\textup{diff}(w,p)} & \textup{ if } \textup{diff}(w,p) \ge 0 \textup{ for all } p \in P\\
\textup{undefined} & \textup{ otherwise.}
\end{cases}
\]
That is, $\strat_\mathsf{norm}(\varepsilon) = \varepsilon$ and, if $\strat_\mathsf{norm}(w)$ is defined,
then it is the concatenation of words $(a,p)^{\textup{diff}(w,p)}$
in any order. Therefore, if $w$ is $\strat_\mathsf{norm}$-compatible,
then $w \cdot \strat_\mathsf{norm}(w)$
contains as many letters $(a,p)$ as letters $(d,p)$, for every $p$.
We deduce $w \cdot \strat_\mathsf{norm}(w) \models \varphi_4$.
\exend
\end{example}

Now, the original and the normalized synthesis problem are equivalent in the following sense:
\begin{lemma}\label{lem:normalized}
$\Win{\Alpha_\sys}{\Alpha_\env}{\varphi} = \NWin{\Alpha_\sys}{\Alpha_\env}{\varphi}$.
\end{lemma}

\begin{proof}
We say that two executions $w$ and $w'$ are \emph{similar}, noted $w \sim w'$, if $w'$ is $w$ with the position of its events rearranged in any combination, i.e. $w \sim w'$ if and only if there exists a letter-preserving bijection from $\Pos{w}$ to $\Pos{w'}$.
Note that in $\FOdata$, there is no way to write constraints on the relative order of positions.
In other words, for any $\varphi \in \FOdata$, if $w \models \varphi$ and $w \sim w'$, then $w' \models \varphi$ too.
This is the property that we use to prove that the Synthesis Problem is equivalent to the normalized one.

For the remainder of this proof, let us fix $\Procs = (\sProcs,\eProcs,\seProcs)$ and its corresponding triple $(\pconsts,\pconsta,\pconstb)$.
$\Procs$-executions and $\Procs$-strategies will simply be referred as executions and strategies respectively.

\underline{$\Win{\Alpha_\sys}{\Alpha_\env}{\varphi} \supseteq \NWin{\Alpha_\sys}{\Alpha_\env}{\varphi}$:}
Suppose that $(\pconsts,\pconsta,\pconstb) \in \NWin{\Alpha_\sys}{\Alpha_\env}{\varphi}$ and let $\stratnorm$ be a normalized winning strategy for System. 
We want to build $\strat$ a winning strategy in the Synthesis Problem.

The idea is to simulate $\stratnorm$ by memorizing the word of actions given by $\stratnorm$ and playing it one action at a time. 
Meanwhile, the actions played by Environment are stored and then processed as if they happened all at once after System finishes playing its word, thus simulating a corresponding normalized run.

We define a function $\mem$ such that for all executions $w = \ev_1 \ev_2 \dots$, $\mem(w) = (w_N, w_\sys, w_\env)$ where $w_N$ is the corresponding normalized run, $w_\sys$ is the word that System must play to simulate the choice of $\stratnorm$, and $w_\env$ stores the actions played by Environment in the meantime. It is defined as follows:
\begin{itemize}
\item $\mem(\varepsilon) = (\varepsilon, \stratnorm(\varepsilon), \varepsilon)$
\item If $\ev \in \DS_\sys$ and $\mem(w) = (w_N, w_\sys, w_\env)$, then
\begin{align*}
\mem(w \cdot \ev) =
\begin{cases}
(w_N \cdot \ev, w'_\sys, w_\env) &\textup{ if } w_\sys = \ev \cdot w'_\sys,\\
(w_N \cdot w_\env \cdot \ev, w'_\sys, \varepsilon) &\textup{ if } w_\sys = \varepsilon \textup{ and } \stratnorm(w_N \cdot w_\env) = \ev \cdot w'_\sys,\\
\textup{undefined otherwise.}
\end{cases}
\end{align*}
\item If $\ev \in \DS_\env$ and $\mem(w) = (w_N, w_\sys, w_\env)$, then
\begin{align*}
\mem(w \cdot \ev) = (w_N, w_\sys, w_\env \cdot \ev)
\end{align*}
\end{itemize}

Then we define an auxiliary function $\strat_\textup{aux}$:
\begin{align*}
\strataux(w_N, w_\sys, w_\env) =
\begin{cases}
\ev \textup{ if } w_\sys = \ev \cdot w'_\sys,\\
\ev \textup{ if } w_\sys = \varepsilon \textup{ and } \stratnorm(w_N \cdot w_\env) = \ev \cdot w'_\sys,\\
\textup{undefined otherwise.}
\end{cases}
\end{align*}

Finally, we define the strategy $\strat$ as $\strat(w) = \strataux(\mem(w))$ when both $\strataux$ and $\mem$ are defined, otherwise $\strat(w) = \varepsilon$.

From these definitions, we can immediately state the following properties describing the workings of $\mem$ and $\strat$:
\begin{enumerate}
\item If $w = w' \ev$ is a $\strat$-compatible execution such that $\ev \in \DS_\sys$ and $\mem(w') = (w_N, \varepsilon, w_\env)$, then $\mem(w) = (w_N w_\env \ev, \ev_2 \dots \ev_n, \varepsilon)$ with $\ev \ev_2 \dots \ev_n = \stratnorm(w_N w_\env)$.
\item If $w = w' w_0 \ev_1 w_1 \dots \ev_m w_m$ is a $\strat$-compatible execution such that 
for all $i \le m$ $w_i \in \DS_\env^\ast$ and $\ev_i \in \DS_\sys$, 
$\mem(w') = (w_N, \ev'_1 \dots \ev'_n, w_\env)$, 
for all $j < m$ $\mem(w' w_0 \dots \ev_j) \neq (\ast, \varepsilon, \ast)$, 
and $\mem(w) = (\ast, \varepsilon, \ast)$, then 
$n = m$, for all $i \le n$ $\ev_i = \ev'_i$, 
and $\mem(w) = (w_N \ev_1 \dots \ev_n, \varepsilon, w_\env w_0 \dots w_n)$. 
\item If $\strat(w) = \varepsilon$ for some $\strat$-compatible execution $w$ then either $w = \varepsilon$ and $\stratnorm(w) = \varepsilon$, or 
$\mem(w) = (w_N, \varepsilon, w_\env)$ with $w_N$ and $w_\env$ such that $\stratnorm(w_N \cdot w_\env)$ is undefined.
\end{enumerate}

Let $w$ be a finite $\strat$-compatible execution such that $\mem(w) = (w_N, \varepsilon, w_\env)$.
This execution can always be decomposed as $w = w_0 \ev_0 w_1 \ldots \ev_m w_m$ with $w_i \in \DS_\env^\ast$ and $\ev_i \in \DS_\sys$ for all $i \le m$.
We show that $w$ can also be written as:
\[
w = w_0^1 \ev_1^1 \dots \ev_{n_1}^1 w_{n_1}^1 \cdot \ev_0^2 w_0^2 \ev_1^2 \dots \ev_{n_2}^2 w_{n_2}^2 \cdot \ldots \cdot \ev_0^k w_0^k \ev_1^k \dots \ev_{n_k}^k w_{n_k}^k
\]
where $k, n_1, \dots, n_k \in \N$
and such that if we define $\sys^j = \ev_0^j \dots \ev_{n_j}^j$ and $\env^j = w_0^j \dots w_{n_j}^j$ for all $j \le k$ then:
\begin{itemize}
\item $w_N = \sys^1 \env^1 \sys^2 \dots \env^{k-1} \sys^k$,
\item $w_\env = \env^k$, and
\item for all $j \le k$, $\sys^j = \stratnorm(\sys^1 \env^1 \dots \sys^{j-1} \env^{j-1})$.
\end{itemize}
We prove this by recursion on the number of prefixes $w'$ of $w$ ending in an action of System such that $\mem(w') = (\ast, \varepsilon, \ast)$, that number being $k$ in the decomposition above.

Let $w = w_0 \ev_0 \ldots \ev_m w_m$ a $\strat$-compatible execution with $\mem(w) = (w_N, \varepsilon, w_\env)$.
We note $(w_N^i, w_\sys^i, w_\env^i) = \mem(w_0 \ev_0 \ldots w_i \ev_i)$ for all $i \le m$.

\paragraph{Base case ($k = 1$).}
Suppose that $w_\sys^i \neq \varepsilon$ for all $i < m$ and $w_\sys^m = \varepsilon$.
As $\mem(\varepsilon) = (\varepsilon, \stratnorm(\varepsilon), \varepsilon)$, if we let $\stratnorm(\varepsilon) = \ev'_1 \ldots \ev'_n$ then by the second property of $\mem$ we have that $n = m$, $\ev_i = \ev'_i$ for all $i \le m$, and that $\mem(w) = (\ev_1 \ldots \ev_n, \varepsilon, w_0 \ldots w_m)$.
Then if we let $\sys^1 = \ev_1 \ldots \ev_n$ and $\env^1 = w_0 \ldots w_m$, the decomposition holds (with $k = 1$).

\paragraph{Induction step.}
Suppose $w = w' \cdot \ev_i w_{i+1} \ldots \ev_m w_m$ where $w' = w_0^1 \ev_1^1 \ldots \ev_{n_k}^k w_{n_k}^k$ with $\mem(w') = (w_N, \varepsilon, w_\env)$ satisfying the conditions above.
Suppose also that $w_\sys^j \neq \varepsilon$ for all $i \le j < m$ and $w_\sys^m = \varepsilon$.
By the first property of $\mem$, we know that $\mem(w' \ev_i) = (w_N \cdot w_\env \cdot \ev_i, \ev'_2 \ldots \ev'_n, \varepsilon)$ with $\ev_i \ev'_2 \ldots \ev'_n = \stratnorm(w_N \cdot w_\env) = \stratnorm(\sys^1 \env^1 \ldots \sys^k \cdot \env^k)$.
Then using the second property, we deduce that $n = m - i$, $\ev'_{j+1} = \ev_{i+j}$ for all $0 < j \le n$, and that $\mem((w' \ev_i) \cdot w_{i+1} \ev_{i+1} \ldots \ev_m w_m) = (w'_N, \varepsilon, w'_\env)$ where $w'_N = w_N \cdot w_\env \cdot \ev_i \ev_{i+1} \ldots \ev_m$ and $w'_\env = w_{i+1} \ldots w_m$.
So we let $\sys^{k+1} = \ev_i \ldots \ev_m$ and $\env^{k+1} = w_{i+1} \ldots w_m$, and all conditions of the decomposition have been satisfied.\\

Thanks to the decomposition we just proved, if $w$ is a finite $\strat$-compatible execution such that $\mem(w) =  (w_N, \varepsilon, w_\env)$ then we can deduce two facts:
that $w \sim w_N \cdot w_\env$ and that $w_N$ is a $\stratnorm$-compatible normalized execution.

Let $w$ be an arbitrary fair $\strat$-compatible execution. We distinguish two different cases:

If $w$ is finite and $\mem(w) = (w_N, w_\sys, w_\env)$, then $\strat(w) = \varepsilon$ because $w$ is fair, 
so either $w = \varepsilon = \stratnorm(w)$ in which case $(\Procs, \varepsilon) \models \varphi$ because $\stratnorm$ is a normalized strategy, 
or $w_\sys = \varepsilon$ and $\stratnorm$ is undefined on $w_N \cdot w_\env$.
Moreover, we get that $w \sim w_N \cdot w_\env$ and that $w_N$ is a $\stratnorm$-compatible normalized execution.
Since $w_N$ is a $\stratnorm$-compatible normalized execution and $\stratnorm$ is undefined on $w_N \cdot w_\env$, then necessarily $w_N \cdot w_\env$ satisfies $\varphi$, otherwise $w_N \cdot w_\env$ would be a $\stratnorm$-compatible maximal normalized execution that is not winning which would contradict that $\stratnorm$ is winning.
Therefore, since $w_N \cdot w_\env$ satisfies $\varphi$ and $w \sim w_N \cdot w_\env$, we have that $w$ satisfies $\varphi$.

If $w$ is infinite, let $w_i$ be the prefix of size $i$ of $w$ and $(w_N^i, w_\sys^i, w_\env^i) = \mem(w_i)$.
We again distinguish two cases.
If there are an infinite number of actions from System, then there is an infinite sequence $i_1 < i_2 < \ldots$ such that $w_\sys^{i_j} = \varepsilon$, which in turn means that there is an increasing sequence of $\stratnorm$-compatible normalized executions $w_N^{i_1}, w_N^{i_2}, \ldots$, so one can find a normalized execution of arbitrary size.
This is impossible, as by Theorem~\ref{thm:normalform} there is a bound on the number of letter that can be played on a single process before the satisfiability of $\varphi$ remains stable, that bound being $\bound.|\Alpha_\theta|$ for processes of type $\theta \in \Types$.
Since the number of processes is fixed that means there is a bound on the total number of times that an execution can go from satisfying $\varphi$ to not satisfying it and vice-versa, which in turn limits the size of normalized executions.

Therefore there is a finite number of actions from System, i.e. $w = w' \cdot w_\env^\infty$ where $w'$ is a finite execution ending with an action from System and $w_\env^\infty$ is an infinite execution with only actions from Environment. 
Let $n = |w'|$. 
By fairness of $w$ necessarily there is some point $K > n$ such that $\strat(w_i) = \varepsilon$ for all $i \ge K$.
Since there are no actions from System in $w_\env^\infty$, we also know that $w_N^i = w_N^n$ and $w_\sys^i = w_\sys^n = \varepsilon$ for all $i > n$, and that $w_N^n$ is a $\stratnorm$-compatible normalized execution.
Thus for all $i \ge K$, $w_i \sim w_N^n \cdot w_\env^i$ and $w_N^n \cdot w_\env^i$ satisfies $\varphi$ otherwise $\stratnorm$ would not be winning, therefore $w_i$ satisfies $\varphi$ for all $i \ge K$.
We conclude that $w$ is winning, and therefore that $\strat$ is a winning strategy in the Synthesis Problem.\\

\underline{$\Win{\Alpha_\sys}{\Alpha_\env}{\varphi} \subseteq \NWin{\Alpha_\sys}{\Alpha_\env}{\varphi}$:}

Suppose that $(\pconsts,\pconsta,\pconstb) \in \Win{\Alpha_\sys}{\Alpha_\env}{\varphi}$ and let $\strat$ be a winning strategy for System. 
We will define $\stratnorm$ a normalized strategy.
Let $w$ be a finite normalized execution; note that $w$ can also be seen as a (regular) execution. 
Suppose that $w$ is a $\strat$-compatible execution that does not satisfy $\varphi$.
Let $\ev_1 = \strat(w)$, $\ev_2 = \strat(w \ev_1)$, $\ev_3 = \strat(w \ev_1 \ev_2)$, and so on.
As $\strat$ is winning, necessarily there exists $i \in \N$ such that $\ev_1, \ldots, \ev_i$ are all not $\varepsilon$ and such that $w \ev_1 \ldots \ev_i$ satisfies $\varphi$, otherwise $w \ev_1 \ev_2 \ldots$ would be an infinite $\strat$-compatible fair execution that is not winning.
We then take the minimal $i$ satisfying those conditions and we define $\stratnorm(w) = \ev_1 \ldots \ev_i$.
Remark that in that case, $w \cdot \stratnorm(w)$ is still a $\strat$-compatible execution.
If $(\Procs, \varepsilon) \models \varphi$, we let $\stratnorm(\varepsilon) = \varepsilon$, and the remark above still holds.
If $w \neq \varepsilon$ either satisfies $\varphi$ or is not $\strat$-compatible, then $\stratnorm$ is undefined.

Let $w = w_\sys^1 w_\env^1 w_\sys^2 \ldots w_\theta^i$ be a $\stratnorm$-compatible normalized execution with $\theta \in \{\sys, \env\}$. 
Then $w$ is also a $\strat$-compatible execution: this is true if $w = \varepsilon$, and if $w'$ is a $\strat$-compatible execution then $w' \cdot \stratnorm(w')$ is also one as we remarked earlier, and $w' \cdot \stratnorm(w') \cdot w_\env$ as well for any $w_\env \in \DS_\env^+$.

Now suppose that $w$ is also $\stratnorm$-maximal. If $\theta = \sys$ then $w$ is winning.
Otherwise, by definition of maximal $\stratnorm(w)$ must be undefined. 
$\stratnorm$ is undefined when $w$ is not $\strat$-compatible or satisfies $\varphi$.
As $w$ is $\strat$-compatible, this means that $w$ must satisfy $\varphi$.
Thus all maximal $\stratnorm$-compatible normalized executions are winning. Thus, $\stratnorm$ is a winning strategy.
\qed
\end{proof}

\subsection{Proof of Lemma~\ref{lemma:equifogame}}

We split the lemma into two, one for each direction.

\begin{lemma}\label{lemma:equifogame1}
For every sentence $\varphi \in \FOdata$, there is a parameterized vector game $\G$ such that
$\Win{}{}{\varphi} = \Win{}{}{\game}$.
\end{lemma}

\begin{proof}
We actually show that parameterized vector games are equivalent to the normalized synthesis problem.
Let $\varphi$ be a sentence in $\FOdata$.
With the normal form from Theorem~\ref{thm:normalform}, we suppose that there is $\bound \in \N$ and that
\[
\varphi = \bigvee_{i = 1}^{n} \left[ \left( \bigwedge_{j = 1}^{m_i} \exists^{= k_j^i} y. (\theta_j^i(y) \wedge \ctype{\bound,\loc_j^i}(y)) \right) 
\land \left( \bigwedge_{j = 1}^{\hat{m}'_i} \exists^{\ge \hat{k}_j^i} y. (\hat{\theta}_j^i(y) \wedge \ctype{\bound,\hat{\loc}_j^i}(y)) \right) \right]
\]
with $k_j^i, \hat{k}_j^i \in \N$, $\theta_j^i, \hat{\theta}_j^i \in \Types$, $\loc_j^i, \hat{\loc}_j^i \in \{0 \ldots \bound\}^\Alpha$ for all $i, j \in \N$.
Let $\Loc = \{0 \ldots \bound\}^\Alpha$, we can also assume that for all $i \in \N$, any pair $(\theta,\loc) \in \Types \times \Loc$ appears at most once in $\cup_{j} \{(\theta_j^i, \loc_j^i), (\hat{\theta}_j^i, \hat{\loc}_j^i)\}$.

We define the parameterized vector game $\G = (\Alpha, \bound, \Acc)$ where $\Alpha$ and $\bound$ are given by $\varphi$, 
and $\Acc = \{\accfunction_i \mid 1 \le i \le n\}$ such that for all $1 \le i \le n$ and $\loc \in \Loc$,
$\accfunction_i(\loc) = (\bowtie_\sys^i n_\sys^i, \bowtie_\env^i n_\env^i, \bowtie_\both^i n_\both^i)$ where
\begin{align*}
\bowtie_\theta^i n_\theta^i = 
\begin{cases}
= k_j^i &\textup{ if } \exists j. (\theta, \loc) = (\theta_j^i, \loc_j^i),\\
\ge \hat{k}_j^i &\textup{ if } \exists j. (\theta, \loc) = (\hat{\theta}_j^i, \hat{\loc}_j^i),\\
\ge 0 &\textup{ otherwise.}
\end{cases}
\end{align*}

\underline{$\NWin{}{}{\varphi} \subseteq \Win{}{}{\game}$:}

First we show how to obtain a play from a normalized execution.
Let $w$ be a normalized $(\sProcs, \eProcs, \seProcs)$-execution and $k_\theta = |\Procs_\theta|$.
By abuse of notation, we note $\play(w)$ the play corresponding to $w$ that we are building.
Let $w = w_1 \dots w_\alpha$ with $\alpha \ge 1$ and let $\conf_0 = \conf_{(k_\sys, k_\env, k_\both)}$.
For all $\proc \in \sProcs \cup \eProcs \cup \seProcs$ and $\beta \in \{1,\dots,\alpha\}$, we define $\loc_\proc^\beta$ that track the state of each process $\proc$ after $w_1 \dots w_\beta$ as:
\[
\loc_\proc^\beta(a) = |\{j \in \Pos{w_1 \dots w_\beta} \mid w[j] = (a, \proc)\}|
\]
By convention, we also let $\loc_\proc^0 = \loc_0$ for all $\proc$.
Then let $\Procs_{\theta,\loc}^\beta$ be the set of processes of type $\theta$ in state $\loc$ after $w_1 \dots w_\beta$, defined as
\[
\Procs_{\theta,\loc}^\beta = \{\proc \in \Procs_\theta \mid \loc = \loc_\proc^\beta\}
\]
Then we define $\play(w) = \conf_0 \upd_1 \conf_1 \dots \upd_\alpha \conf_\alpha$ where 
for all $\beta \in \{1, \dots, \alpha\}$ and $(\loc, \loc') \in \Loc^2$, 
$\upd_\beta(\loc,\loc') = (\upd_\sys, \upd_\env, \upd_\both)$ with 
$\upd_\theta = |\Procs_{\theta,\loc}^{\beta-1} \cap \Procs_{\theta,\loc'}^\beta|$ 
and $\conf_\beta = \upd_\beta(\conf_{\beta-1})$.

Let $w = w_1 \dots w_\alpha$ be a normalized execution and let $\play(w) = \conf_0 \dots \conf_\alpha$.
We prove that for all $\loc \in \Loc$, $\beta \le \alpha$, and $\theta \in \Types$ we have $\conf_\beta(\loc,\theta) = |\Procs_{\theta,\loc}^\beta|$.
If $\beta = 0$ then $\conf_\beta(\loc_0,\theta) = k_\theta = |\Procs_{\theta,\loc}^0|$.
If the property is true for $\beta < \alpha$, then 
$\conf_{\beta+1}(\loc) = \conf_\beta(\loc)  - \outm{\upd_{\beta+1}}(\loc) + \inm{\upd_{\beta+1}}(\loc) = \conf_\beta(\loc) - \sum_{\loc' \in \Loc} \upd_{\beta+1}(\loc, \loc') + \sum_{\loc' \in \Loc} \upd_{\beta+1}(\loc', \loc)$.
For a given $\theta \in \Types$, this simplifies into 
$\conf_{\beta+1}(\loc,\theta) = 
|\Procs_{\theta,\loc}^\beta| 
- \sum_{\loc' \in \Loc}\left(|\Procs_{\theta,\loc}^{\beta} \cap \Procs_{\theta,\loc'}^{\beta+1}|\right) 
+ \sum_{\loc' \in \Loc}\left(|\Procs_{\theta,\loc'}^{\beta} \cap \Procs_{\theta,\loc}^{\beta+1}|\right)
= |\Procs_{\theta,\loc}^\beta| 
- |\{\proc \in \Procs_\theta \mid \proc \in \Procs_{\theta,\loc}^\beta$ and $\proc \notin \Procs_{\theta,\loc}^{\beta+1}\}| 
+ |\{\proc \in \Procs_\theta \mid \proc \notin \Procs_{\theta,\loc}^\beta$ and $\proc \in \Procs_{\theta,\loc}^{\beta+1}\}| 
= |\Procs_{\theta,\loc}^{\beta+1}|$.

Consequently, we can prove that $w$ is winning iff $\play(w)$ is winning.
Suppose there is $i \le n$ such that $(\Procs, w_1 \dots w_\alpha) \models \varphi_i$ with
\[\varphi_i = \left( \bigwedge_{j = 1}^{m_i} \exists^{= k_j^i} y. (\theta_j^i(y) \wedge \ctype{\bound,\loc_j^i}(y)) \right) 
\land \left( \bigwedge_{j = 1}^{\hat{m}'_i} \exists^{\ge \hat{k}_j^i} y. (\hat{\theta}_j^i(y) \wedge \ctype{\bound,\hat{\loc}_j^i}(y)) \right)\]
Then by definition of the subformulas $\ctype{\bound,\loc_j^i}(y)$, 
for all $(\theta,\loc)$ and $j$ such that $(\theta,\loc) = (\theta_j^i, \loc_j^i)$ (respectively $= (\hat\theta_j^i, \hat\loc_j^i)$), 
there must be at exactly $k_j^i$ (resp. at least $\hat{k}_j^i$) processes of type $\theta$ in state $\loc$
i.e. $|\Procs_{\theta,\loc}^\alpha| = k_j^i$ (resp. $\ge \hat{k}_j^i$).
Therefore $\conf_\alpha(\loc,\theta) = k_j^i$ (resp. $\ge \hat{k}_j^i$) for all $(\theta,\loc)$, so $\conf_\alpha$ satisfies $\accfunction_i$ thus $\play(w)$ is winning.
The other direction is similar.

Now suppose there is a winning normalized $\Procs$-strategy $\strat$.
We define a strategy $\strat_\game$ in $\game$ as $\strat_\game(\play) = \upd_\alpha$ if there is a $\strat$-compatible play $w$ such that $\play = \play(w)$, $\strat(w)$ is defined and $\play(w \cdot \strat(w)) = \conf_0 \dots \upd_\alpha \conf_\alpha$.
Moreover, let $\strat_\game(\varepsilon) = \upd_1$ where $\play(\strat(\varepsilon)) = \conf_0 \upd_1 \conf_1$.
In all other cases, $\strat_\game$ is undefined.

Finally, we show that $\strat_\game$ is winning.
If $\play = \conf_0 \upd_1 \conf_1 \dots \conf_\alpha$ is a $\strat_\game$-compatible play, then inductively by definition of $\strat_\game$ we know that there is $w = w_1 \dots w_\alpha$ such that for all $\beta \le \alpha$, $\conf_0 \upd_1 \conf_1 \dots \conf_\beta = \play(w_1 \dots w_\beta)$.
Furthermore, if $\play$ is $\strat_\game$-maximal, then there it is not the prefix of a longer $\strat_\game$-compatible play.
If $w$ was not $\strat$-maximal, there would be an execution $w' = w w_{\alpha+1}$ that is $\strat$-compatible, but in that case $\play(w')$ would be a $\strat_\game$-compatible play which contradicts the maximality of $\play$.
Therefore $w$ must be $\strat$-maximal, and thus winning as $\strat$ is a winning strategy.
Since we proved that $w$ is winning iff $\play(w)$ is winning, then $\play$ is a winning play, therefore $\strat_\game$ is a winning strategy.

\underline{$\NWin{}{}{\varphi} \supseteq \Win{}{}{\game}$:}

Let $\conf_0 = \conf_{(k_\sys, k_\env, k_\both)}$ for some $k_\sys, k_\env, k_\both \in \N$.
We define $\Procs_\theta = \{1, \dots, k_\theta\}$ for all $\theta \in \Types$.
For all $\conf_0$-plays $\play$, let us build a normalized $(\sProcs, \eProcs, \seProcs)$-execution that we note $w(\play)$ again by abuse of notation.
Since processes in $\G$ do not have identities, we will need to arbitrarily assign one to each of them.
To that end, we define a function $\memG$ such that for all $\conf_0$-plays $\play \in \Plays$, locations $\loc \in \Loc$, and $\theta \in \Types$, 
$\memG(\play, \loc, \theta) = S$ with $S \subseteq \Procs_\theta$ storing the identities of all processes in location $\loc$ at the end of play $\play$.
First we fix an arbitrary total order $<$ on $\Loc^2$.
Then $\memG$ is defined as follows:
\begin{align*}
\memG(\conf_0, \loc, \theta) =
\begin{cases}
\Procs_\theta &\textup{ if } \loc = \loc_0,\\
\emptyset &\textup{ otherwise.}
\end{cases}
\end{align*}
and for all $\play = \conf_0 \upd_1 \dots \conf_\alpha$ such that $\memG(\play, \loc, \theta)$ is defined for all $(\loc, \theta) \in \Loc \times \Types$,
for all $\upd$ applicable at $\conf_\alpha$ and $\conf = \upd(\conf_\alpha)$, 
for all $\loc, \loc' \in \Loc$ such that $\upd(\loc, \loc') = (n_\sys, n_\env, n_\both)$, 
for all $\theta \in \Types$, we define 
$S_{\loc, \loc'}^\theta$ as the $n_\theta$ lowest (w.r.t. the natural order on $\N$) elements of 
\[
\memG(\play, \loc, \theta) \setminus \left( \bigcup_{(\hat\loc, \hat\loc') < (\loc, \loc')} S_{\hat\loc, \hat\loc'}^\theta \right)
\]
which is always well-defined if $\upd$ is applicable as we supposed. Then we let 
\[
\memG(\play \upd \conf, \loc, \theta) = \memG(\play, \loc, \theta) 
\cup \left( \bigcup_{\loc' \neq \loc} S_{\loc', \loc}^\theta \right)
\setminus \left( \bigcup_{\loc' \neq \loc} S_{\loc, \loc'}^\theta \right) 
\]
With that being done, we define $w(\play)$ recursively. 
Let $w(\conf_0) = \varepsilon$.
For all $\play = \conf_0 \upd_1 \conf_1 \dots \conf_\alpha$ such that $w(\play)$ is defined,
for all $\upd$ applicable at $\conf_\alpha$ and $\conf = \upd(\conf_\alpha)$, 
we let
\[
w(\play \upd \conf) = w(\play) \cdot 
\prod_{\substack{\loc' = \loc + a_1 \dots a_j,\\ \theta \in \Types\\ \proc \in \memG(\play, \loc, \theta) \cap\\ \memG(\play \upd \conf, \loc', \theta)}}
(a_1, \proc) \cdot \ldots \cdot (a_j, \proc)
\]

We prove that for all $\conf_0$-plays $\play = \conf_0 \upd_1 \dots \conf_\alpha$ and $w(\play) = w_1 \dots w_\alpha$, 
for all $\loc \in \Loc$, $\beta \le \alpha$, and $\theta \in \Types$ we have 
$\memG(\conf_0 \upd_1 \dots \conf_\beta, \loc, \theta) = \Procs_{\theta,\loc}^\beta$ with $\Procs_{\theta,\loc}^\beta$ defined as before.
If $\beta = 0$, for all processes $\proc$ we have that $\loc_\proc^0 = \loc_0$, so $\Procs_{\theta,\loc}^0 = \Procs_\theta$ if $\loc = \loc_0$ and $\emptyset$ otherwise, therefore $\Procs_{\theta,\loc}^0 = \memG(\conf_0, \loc, \theta)$.
If the property holds for some $\beta < \alpha$, then $\memG(\conf_0 \upd_1 \dots \upd_{\beta+1} \conf_{\beta+1}, \loc, \theta) = 
\Procs_{\theta,\loc}^\beta 
\cup \left( \bigcup_{\loc' \neq \loc} S_{\loc', \loc}^\theta \right)
\setminus \left( \bigcup_{\loc' \neq \loc} S_{\loc, \loc'}^\theta \right)$.
Moreover, $\Procs_{\theta,\loc}^{\beta+1} = \Procs_{\theta,\loc}^\beta 
\cup \{\proc \in \Procs_\theta \mid \proc \notin \Procs_{\theta,\loc}^\beta$ and $\proc \in \Procs_{\theta,\loc}^{\beta+1}\}
\setminus \{\proc \in \Procs_\theta \mid \proc \in \Procs_{\theta,\loc}^\beta$ and $\proc \notin \Procs_{\theta,\loc}^{\beta+1}\}
= \Procs_{\theta,\loc}^\beta 
\cup \left( \bigcup_{\loc' \neq \loc} \{\proc \in \Procs_\theta \mid \proc \in \Procs_{\theta,\loc'}^\beta \textup{ and } \proc \in \Procs_{\theta,\loc}^{\beta+1}\}\right)
\setminus \left( \bigcup_{\loc' \neq \loc} \{\proc \in \Procs_\theta \mid \proc \in \Procs_{\theta,\loc}^\beta \textup{ and } \proc \in \Procs_{\theta,\loc'}^{\beta+1}\} \right)$.
If $\proc \in \Procs_\theta$ is such that $\proc \in \Procs_{\theta,\loc'}^\beta$ and $\proc \in \Procs_{\theta,\loc}^{\beta+1}$ for some $\loc' \neq \loc$, then by definition of $w(\play)$ necessarily $\loc = \loc' + a_1 \dots a_j$ and $\proc \in \memG(\conf_0 \dots \conf_\beta, \loc', \theta) \cap \memG(\conf_0 \dots \conf_{\beta+1}, \loc, \theta)$, and therefore $\proc \in S_{\loc',\loc}^\theta$. 
The reverse is also true.
Therefore, $\Procs_{\theta,\loc}^{\beta+1} =
\memG(\conf_0 \dots \conf_\beta, \loc, \theta)
\cup \left( \bigcup_{\loc' \neq \loc} S_{\loc', \loc}^\theta \right)
\setminus \left( \bigcup_{\loc' \neq \loc} S_{\loc, \loc'}^\theta \right)
= \memG(\conf_0 \dots \conf_\beta \upd_{\beta+1} \conf_{\beta+1}, \loc, \theta)$.

Furthermore, it is easy to see that $\conf_\beta(\loc,\theta) = |\memG(\conf_0 \dots \conf_\beta, \loc, \theta)|$ for all $\beta, \loc, \theta$.
Therefore, as in the other direction, we have that $\conf_\beta(\loc,\theta) = |\Procs_{\theta,\loc}^\beta|$, which in turn gives us that $\play$ is winning iff $w(\play)$ is winning.

Now suppose there is a winning strategy $\strat_\game$ in $\game$.
We define a normalized strategy $\strat$ as $\strat(\varepsilon) = w(\conf_0 \upd \conf_1)$ with $\upd = \strat_\game(\conf_0)$ and $\conf_1 = \upd(\conf_0)$, and for all $w$ we define $\strat(w) = w'$ if there is a play $\play$ ending in $\conf$ such that $w = w(\play)$, $\strat_\game(\play) = \upd$ is defined and $w(\play \upd \conf') = w w'$ where $\conf' = \upd(\conf)$. 
In all other cases, $\strat(w)$ is undefined.
The proof that $\strat$ is a winning strategy is the same as the other direction, with the roles of $\strat$ and $\strat_\game$ as well as $w$ and $\play$ swapped, since the definitions of compatibility and maximality are the same for the normalized synthesis and the parameterized vector games.
\qed
\end{proof}

\begin{lemma}\label{lemma:equifogame2}
For every parameterized vector game $\game$, there is a sentence $\varphi \in \FOdata$
such that $\Win{}{}{\game} = \Win{}{}{\varphi}$.
\end{lemma}

\begin{proof}

Let $\G = (\Alpha, \bound, \Acc)$ be a parameterized vector game, and let $\Acc = \{\accfunction_i \mid 1 \le i \le n\}$.
As usual, let $\Loc = \{0 \dots \bound\}^\Alpha$.
For all $\loc \in \Loc$ and $1 \le i \le n$, if $\accfunction_i(\loc) = (\bowtie_\sys^i n_\sys^i, \bowtie_\env^i n_\env^i, \bowtie_\both^i n_\both^i)$, then for all $\theta \in \Types$ we let:
\[
\varphi_{i, \loc, \theta} = \exists^{\bowtie_\theta^i n_\theta^i} y. (\theta(y) \land \ctype{\bound,\loc}(y))
\]
which is a $\FOdata$ formula and then we define:
\[
\varphi = \bigvee_{i = 1}^n\; \bigwedge_{\substack{\loc \in \Loc\\\theta \in \Types}} \varphi_{i, \loc, \theta}
\]

The proof that $\Win{}{}{\game} = \Win{}{}{\varphi}$ is similar to the one from Lemma~\ref{lemma:equifogame1}.
\end{proof}

\section{Proof of Lemma~\ref{lem:cutoffgameNN0} (no cutoff for {\normalfont{$\Gameproblem{\N}{\N}{0})$}}}
\label{app:cutoffgameNN0}
Starting from configuration $\conf_0 = (k_\sys, k_\env, 0)$, System has a winning strategy if $k_\sys \ge k_\env$: first send one System process from $\loc_0$ to $\loclet{a}$ (to satisfy $\accfunction_1$), wait until Environment sends one token from $\loc_0$ to $\loclet{b}$ (but not more than one otherwise $\accfunction_2$ would be satisfied), then send the first process to $\loclet{a^2}$ (to satisfy $\accfunction_3$), and wait until Environment does the same to $\loclet{b^2}$ (the only way to falsify $\accfunction_3$). 
Then repeat from the beginning until Environment has all her tokens in $\loclet{b^2}$, which will happen since one process of each type is moved from $\loc_0$ at each time and because we supposed that there are more System processes than Environment processes.
Finally, send all the remaining System tokens in $\loc_0$ directly to $\loclet{a^2}$, which satisfies $\accfunction_4$ and Environment has no more possible move.

It is also easy to see that the strategy described here is the only possible winning strategy.
Therefore, if $k_\env > k_\sys$, a configuration will occur with 0 System processes and at least 1 Environment process in $\loc_0$, which is losing for System.
Note that this game could easily be adapted to give a game where System wins when she has at least $k$ times the number of processes of Environment, for any $k \in \N$.

\section{Proof of Theorem~\ref{thm:game00N} ({\normalfont{$\Gameproblem{\Zero}{\Zero}{\N}$}} is undecidable)}
\label{app:undec-00N}

\fbox{$\Rightarrow$}
Suppose there is an accepting run $\rho : (\tcminit, 0, 0) \vdash_{\tcmt_1} (\tcmq_1, \nu_1^1, \nu_2^1) \vdash_{\tcmt_2} \dots \vdash_{\tcmt_n} (\tcmq_n, \nu_1^n, \nu_2^n)$ for $\tcm$, and fix some $k \ge 3n+1$.
Without loss of generality, we can assume that the configurations $\gamma_0, \dots, \gamma_n$ visited in $\rho$ are pairwise different.
The positional strategy $\strat$ for System that faithfully simulates $\rho$ is formally defined as follows.
In the following, a transition $(q,\op,q') \in \tcmT$ is written $q \xrightarrow{\op} q'$.
\paragraph{Initialization.}
Let $\conf_0$ be the initial $\game$-configuration.

\begin{itemize}
\item If $\tcmt_1 = \tcminit \xrightarrow{\tcmcounter_i++} \tcmq_1$, we let $\upd_1$ be defined by $\upd_1(\loc_0,\loclet{\tcmq_0})=1$, $\upd_1(\loc_0,\loclet{\tcmt_1})=1$,
$\upd_1(\loc_0,\loclet{a_i})=1$, and $\upd_1(\loc,\loc')=0$ for all other $\loc,\loc'\in\Loc$.

\item If $\tcmt_1 = \tcminit \xrightarrow{\tcmcounter_i == 0} \tcmq_1$, we let $\upd_1$ be defined by $\upd_1(\loc_0,\loclet{\tcmq_0})=1$, $\upd_1(\loc_0,\loclet{\tcmt_1})=1$,
and $\upd_1(\loc,\loc')=0$ for all other $\loc,\loc'\in\Loc$.
\end{itemize}

We then let $f(\conf_0)=\upd_1$.

\paragraph{Simulation of a new transition.}
For $0< j<n$, for any $k$-configuration $\conf\in \confset((\tcmq_j,\nu^j_1,\nu^j_2))$, we let $f(\conf)=\upd_{j+1}$, with $\upd_{j+1}$ defined as follows.
\begin{itemize}
\item If $\tcmt_{j+1} = \tcmq_j \xrightarrow{\tcmcounter_i++} \tcmq_{j+1}$, then $\upd_{j+1}(\loc_0,\loclet{\tcmt_{j+1}})=1$, $\upd_{j+1}(\loc_0, \loclet{a_i})=1$, and $\upd_{j+1}(\loc,\loc')=0$
for all other $\loc,\loc'\in\Loc$.
\item If $\tcmt_{j+1} = \tcmq_j \xrightarrow{\tcmcounter_i--} \tcmq_{j+1}$, then $\upd_{j+1}(\loc_0, \loclet{\tcmt_{j+1}})= 1$, $\upd_{j+1}(\loclet{a_i^2 b^2}, \loclet{a_i^3 b^2}) =1$, and $\upd_{j+1}(\loc,\loc')=0$
for all other $\loc,\loc'\in\Loc$.

\item If $\tcmt_{j+1} = \tcmq_j \xrightarrow{\tcmcounter_i==0} \tcmq_{j+1}$, then $\upd_{j+1}(\loc_0, \loclet{\tcmt_{j+1}}) = 1$, and $\upd_{j+1}(\loc,\loc')=0$
for all other $\loc,\loc'\in\Loc$.
\end{itemize}
Note that if $\conf \in \confset((\tcmq_j,\nu^j_1,\nu^j_2))$ then $\conf \notin \confset((\tcmq_{j'},\nu^{j'}_1,\nu^{j'}_2))$ for $j' \neq j$ as we assumed that $\gamma_j \neq \gamma_{j'}$, therefore $\strat(\conf)$ is well-defined.

\paragraph{Second step of the simulation of a transition.}
Let $\conf$ be a $k$-configuration such that $\conf(\loclet{\tcmq b})=1$ for some $\tcmq\in\tcmQ$, $\conf(\loclet{\tcmt b})=1$ for some $\tcmt\in\tcmT$, $\conf(\loclet{a_i^2b^2})\geq 0$, 
$\conf(\loclet{a_i^4b^4})\geq 0$ for $i=1,2$, 
$\conf(\loclet{\tcmt^2b^2})\geq 0$ for all $\tcmt\in\tcmT$, $\conf(\loclet{\tcmq^2b^2})\geq 0$ for all $\tcmq\in\tcmQ$, $\conf(\loc_0)>0$
and $\conf(\loc)=0$ for all other $\loc\in\Loc$. We define $f(\conf)=\upd$ with $\upd$ defined as follows.
\begin{itemize}
\item If $\tcmt=\tcmq\xrightarrow{c_i++} \tcmq'$  and $\conf$ is such that $\conf(\loclet{a_ib})=1$. Then  $\upd(\loclet{\tcmq b},\loclet{\tcmq^2b})=1$, 
$\upd(\loclet{\tcmt b}, \loclet{\tcmt^2b})=1$, 
$\upd(\loclet{a_ib},\loclet{a_i^2b})=1$, 
$\upd(\loc_0,\loclet{q'})=1$ and $\upd(\loc,\loc')=0$ for all other $\loc,\loc'\in\Loc$.

\item If $\tcmt=\tcmq\xrightarrow{c_i--} \tcmq'$ $\conf$ is such that $\conf(\loclet{a_i^3b^3})=1$, then $\upd$ is defined by $\upd(\loclet{\tcmq b}, \loclet{\tcmq^2b})=1$, $\upd(\loclet{\tcmt b}, \loclet{\tcmt^2b})=1$, $\upd(\loclet{a_i^3b^3},\loclet{a_i^4b^3})=1$, 
$\upd(\loc_0,\loclet{q'})=1$ and $\upd(\loc,\loc')=0$ for all other $\loc,\loc'\in\Loc$.

\item If $\tcmt =\tcmq \xrightarrow{\tcmcounter_i==0} \tcmq'$, $\upd(\loclet{\tcmq b}, \loclet{\tcmq^2 b})=1$, $\upd(\loclet{\tcmt b}, \loclet{\tcmt^2 b})=1$, 
$\upd(\loc_0, \loclet{\tcmq'}) = 1$ and $\upd(\loc,\loc')=0$ for all other $\loc,\loc'\in\Loc$.
\end{itemize}

If $\play$ is a partial play ending in $\conf \in \confset(\gamma_n)$, then $\strat(\conf)$ is the update function $\upd$ such that $\upd(\loclet{\tcmq_n}, \loclet{\tcmq_n^2}) =1$ and $\upd(\loc,\loc')=0$ for all other $\loc,\loc'\in\Loc$.

And for any other configuration, $\strat$ is undefined.

We show that $\strat$ is winning by contradiction: suppose there is a winning strategy $\strat_\env$ for Environment. 
Let $\pi = \conf_0\upd'_0 \conf'_0 \upd_1\conf_1 \dots$ be the maximal play compatible with $\strat$ and $\strat_\env$. We show the following by recursion: 
\begin{lemma}
For all $0 < j \le n$, $\conf_{2j} \in \confset(\gamma_j)$ and $\conf_{2j}(\loc_0) \ge k - (3j+1)$
\end{lemma}
Intuitively, this lemma states that $\play$ correctly simulates $\rho$ and that there are always enough process in $\loc_0$ for System to do his transitions. 
\begin{proof}
We prove it by induction on $j$.

\underline{Base step ($\conf_2 \in \confset(\gamma_1)$):}
$\conf_0$ is the initial ($k$-)configuration. Suppose that $\tcmt_1 = \tcminit \xrightarrow{\tcmcounter_1++} \tcmq_1$, then by definition of $\strat$, 
$\conf'_0(\loclet{\tcminit})=1$, $\conf'_0(\loclet{\tcmt_1})=1$, $\conf'_0(\loclet{a_1})=1$, $\conf'(\loc_0)=k-3> 1$, and $\conf'_0(\loc)=0$ for all other $\loc\in\Loc$. 
Since $\conf'_0\models \accCond{\tcmt_1}$, $\strat_\env(\conf'_0)=\upd_1$ is defined, otherwise it is not winning. Then,
\begin{itemize}
\item If $\upd_1(\loc_0,\loclet{b^m})\geq 1$ for $m\geq 1$, then $\conf_1\models\accCond{\loclet{b^m}}$. 
\item If $\upd_1(\loclet{\tcmq_0},\loclet{\tcmq_0b})=0$, or if $\upd_1(\loclet{\tcmt_1}, \loclet{\tcmt_1b})=0$ or if $\upd_1(\loclet{a_1}, \loclet{a_1b})=0$, then 
$\conf_1\models\accCond{(\tcmq_0,\tcmt_1,\env)}$.
\item If $\upd_1(\loclet{\tcmq_0},\loclet{\tcmq_0b^m})=1$, (respectively if $\upd_1(\loclet{\tcmt_1}, \loclet{\tcmt_1b^m})=1$ or if $\upd_1(\loclet{a_1},\\ \loclet{a_1b^m})=1$, 
for $m>1$), then $\conf_1\models\accCond{\loclet{\tcmq_0b^m}}$ (respectively $\conf_1\models\accCond{\loclet{\tcmt_1b^m}}$, $\conf_1\models\accCond{\loclet{a_1b^m}}$).
\end{itemize}
Hence $\upd_1(\loclet{\tcminit}, \loclet{\tcminit b})=\upd_1(\loclet{\tcmt_1},\loclet{\tcmt_1b})=\upd_1(\loclet{a_1},\loclet{a_1b})=1$, and for all other $\loc\in\Loc$,
$\upd_1(\loc_0,\loc)=0$ and
$\conf_1(\loclet{\tcminit b}) =\conf_1( \loclet{\tcmt_1 b})=\conf_1(\loclet{a_1 b})= 1$, $\conf_1(\loc_0) =k-3$ and $\conf_1(\loc)=0$ for all other $\loc\in\Loc$.

Following the definition of $\strat$, $\conf'_1(\loclet{\tcminit^2 b})=\conf'_1(\loclet{\tcmt_1^2 b})=\conf'_1(\loclet{a_1^2b})=\conf'_1(\loclet{\tcmq_1})\\ =1$, 
$\conf'_1(\loc_0)=k-4> 0$, and $\conf'_1(\loc)=0$ for all other $\loc\in\Loc$. 
Since, $\conf'_1\models\accCond{\tcminit, \tcmt_1,\tcmq_1}$, $\strat_\env=\upd_2$ is defined.

Again we look at all possible transitions for Environment:
\begin{itemize}
\item As before, if $\upd_2(\loc_0,\loclet{b^m})\geq 1$ for $m\geq 1$, then $\conf_2\models\accCond{\loclet{b^m}}$. 
\item If $\upd_2(\loclet{\tcmq_0^2b},\loclet{\tcmq_0^2b^2})=0$, or if $\upd_2(\loclet{\tcmt_1^2b}, \loclet{\tcmt_1^2b^2})=0$ or if $\upd_2(\loclet{a_1^2b}, \loclet{a_1^2b^2})=0$, then 
$\conf_2\models\accCond{(\tcmq_0,\tcmt_1,\tcmq_1,\env)}$.
\item If $\upd_2(\loclet{\tcmq_0^2b},\loclet{\tcmq_0^2b^m})=1$, (respectively if $\upd_2(\loclet{\tcmt_1^2b}, \loclet{\tcmt_1^2b^m})=1$ or if $\upd_2(\loclet{a_1^2b},\\ \loclet{a_1^2b^m})=1$, 
for $m>2$), then $\conf_2\models\accCond{\loclet{\tcmq_0b^m}}$ (respectively $\conf_2\models\accCond{\loclet{\tcmt_1b^m}}$, $\conf_2\models\accCond{\loclet{a_1b^m}}$).
\item Finally, if $\upd_2(\loclet{\tcmq_1},\loclet{\tcmq_1b})=1$ then $\conf_2\models\accCond{(\tcminit,\tcmt_1,\tcmq_1,\env)}$ and if $\upd_2(\loclet{\tcmq_1},\loclet{\tcmq_1b^m})\\=1$, for $m>2$,
then $\conf_2\models\accCond{\loclet{\tcmq_1b^m}}$.
\end{itemize}
Thus, necessarily, $\conf_2(\loclet{\tcminit^2 b^2})=\conf_2(\loclet{\tcmt_1^2 b^2})=\conf_2(\loclet{a_1^2 b^2})=\conf_2(\loclet{\tcmq_1})=1$, $\conf_2(\loc_0)=k-4$, and
$\conf_2(\loc')=0$ for all other $\loc'\in\Loc$. Hence, $\conf_2\in \confset(\tcmq_1, 1, 0)$, and $\conf_2(\loc_0) \ge k-(3*1+1) = k-4$.

If $\tcmt_1=\tcminit\xrightarrow{c_2++}\tcmq_1$, the proof is identical, but with $a_2$ replacing $a_1$. If now 
$\tcmt_1=\tcminit\xrightarrow{c_i==0}\tcmq_1$, the proof goes along the same lines, without difficulty.

\underline{Induction step:}
Let $0< j<n$ and $\gamma_j = (\tcmq_j, \nu_1^j, \nu_2^j)$, and suppose that $\conf_{2j} \in \confset(\gamma_j)$ and $\conf_{2j}(\loc_0) \ge k - (3j + 1)\geq 3$. 
There are six cases depending on the type of $\tcmt_{j+1}$. Without loss of generality, we consider here only the three cases involving $\tcmcounter_1$.

If $\tcmt_{j+1} = \tcmq_j \xrightarrow{\tcmcounter_1++} \tcmq_{j+1}$ then $\gamma_{j+1} = (\tcmq_{j+1}, \nu_1^j+1, \nu_2^j)$. Following $\strat$, 
we obtain  that $\conf_{2j}'(\loclet{\tcmt_{j+1}})=1$, $\conf_{2j}'(\loclet{a_1})=1$, $\conf_{2j}'(\loc_0)=\conf_{2j}(\loc_0)-2$
and $\conf_{2j}'(\loc)=(\conf_{2j})(\loc)$ for all other $\loc\in\Loc$. 
With the same arguments as in the base case, the only possibility is that $\strat_\env(\conf_{2j}')=\upd_{2j+1}$ with 
$\upd_{2j+1}(\loclet{\tcmq_j}, \loclet{\tcmq_j b}) =1$, $\upd_{2j+1}(\loclet{\tcmt_{j+1}},\\ \loclet{\tcmt_{j+1} b})=1$, $\upd_{2j+1}(\loclet{a_1}, \loclet{a_1 b}) =1$ and $\upd_{2j+1}(\loc,\loc')=0$ for all other $\loc,\loc'\in\Loc$, yielding the configuration $\conf_{2j+1}(\loclet{\tcmq_j b})=1$, $\conf_{2j+1}(\loclet{a_1 b})=1$,
$\conf_{2j+1}(\loclet{\tcmt_{j+1} b})=1$, $\conf_{2j+1}(\loclet{\tcmq_j})=\conf_{2j+1}(\loclet{a_1})=\conf_{2j+1}(\loclet{\tcmt_{j+1}})=0$ and $\conf_{2j+1}(\loc)=
\conf'_{2j}(\loc)$ for all other $\loc\in\Loc$.
By definition of $\strat$, the action of System leads to $\conf_{2j+1}'$ defined by $\conf_{2j+1}'(\loclet{\tcmq_j^2 b})=\conf_{2j+1}'(\loclet{\tcmt_{j+1}^2 b})=\conf_{2j+1}'(\loclet{a_1^2 b}=\conf_{2j+1}'(\loclet{\tcmq_{j+1}})=1$, $\conf_{2j+1}'(\loclet{\tcmq_j b})=\conf_{2j+1}'(\loclet{\tcmt_{j+1} b})=\conf_{2j+1}'(\loclet{a_1 b}=0$,
$\conf_{2j+1}'(\loc_0)=\conf_{2j+1}(\loc_0)-1=\conf_{2j}(\loc_0)-3$, and $\conf_{2j+1}'(\loc)=0$ for all other $\loc\in\Loc$.
Finally, again as in the base case, we necessarily have $\strat_\env=\upd_{2j+2}$ with $\upd_{2j+2}(\loclet{\tcmq_j^2 b}, \loclet{\tcmq_j^2 b^2}) =1$, 
$\upd_{2j+2}(\loclet{\tcmt_{j+1}^2 b}, \loclet{\tcmt_{j+1}^2 b^2}) =1$, $\upd_{2j+2}(\loclet{a_1^2 b}, \loclet{a_1^2 b^2}) = 1$.
Hence, $\conf_{2j+2}(\loclet{\tcmq_j^2 b^2})=\conf_{2j}(\loclet{\tcmq_j^2 b^2})+1$, $\conf_{2j+2}(\loclet{\tcmt_j^2 b^2}) =\conf_{2j}(\loclet{\tcmt_j^2 b^2})+1$, $\conf_{2j+2}(\loclet{a_1^2 b^2}) =\conf_{2j}(\loclet{a_1^2 b^2})+1$, $\conf_{2j+2}(\loclet{\tcmq_{j+1}}) =1$, 
$\conf_{2j+2}(\loc_0)=\conf_{2j}(\loc_0)-3$ and $\conf_{2j+1}(\loc)=\conf_{2j}(\loc)$ for all other $\loc\in\Loc$. Since $\conf_{2j} \in \confset(\gamma_j)$ and $\gamma_{j+1} = (\tcmq_{j+1}, \nu_1^j + 1, \nu_2^j)$ it is easy to verify that $\conf_{2j+2}$ is indeed in $\confset(\gamma_{j+1})$. Moreover, $\conf_{2j+2}(\loc_0) = \conf_{2j}(\loc_0) - 3 \ge k - (3(j+1) + 1)$. 

If $\tcmt_{j+1} = \tcmq_j \xrightarrow{\tcmcounter_1--} \tcmq_{j+1}$, then we know that $\nu_1^j \ge 1$. Since $\conf_{2j} \in \confset(\gamma_j)$, we deduce that $\conf_{2j}(\loclet{a_1^2 b^2}) \ge 1$.
Following $\strat$, $\conf'_{2j}(\loclet{\tcmt_{j+1}})=\conf'_{2j}(\loclet{a_1^3b^2})=1$, $\conf'_{2j}(\loc_0)=\conf_{2j}(\loc_0)-1$, $\conf'_{2j}(\loclet{a_1^2b^2})=\conf_{2j}(\loclet{a_1^2b^2})-1$ and $\conf'_{2j}(\loc)=\conf_{2j}(\loc)$ for all other $\loc\in\Loc$. To avoid reaching configurations in $\accCond{\loc}$ for some $\loc\in\Loc$, or in $\accCond{(\tcmq_j,\tcmt_{j+1},\env)}$, 
Environment necessarily updates the configuration to $\conf_{2j+1}(\loclet{\tcmt_{j+1}b})\\=\conf_{2j+1}(\loclet{a_1^3b^3})=\conf_{2j+1}(\loclet{\tcmq_{j}b})=1$, and $\conf_{2j+1}(\loc)=
\conf'_{2j}(\loc)$ for all other $\loc\in\Loc$. Again, the strategy defined for System leads to the configuration $\conf'_{2j+1}(\loclet{\tcmq_j^2b})=\conf'_{2j+1}(\loclet{\tcmt_{j+1}^2b})=
\conf'_{2j+1}(\loclet{a_1^4b^3})=\conf'_{2j+1}(\loclet{\tcmq_{j+1}})=1$, $\conf'_{2j+1}(\loclet{\tcmt_{j+1}b})=\conf'_{2j+1}\\(\loclet{a_1^3b^3})=\conf'_{2j+1}(\loclet{\tcmq_{j}b})=0$,
$\conf'_{2j+1}(\loc_0)=\conf_{2j+1}(\loc_0)-1$, and $\conf'_{2j+1}(\loc)=\conf_{2j+1}(\loc)$ for all other $\loc\in\Loc$. Finally, the only possible move for Environment is 
$\conf_{2j+2}(\loclet{\tcmq_j^2b^2})=\conf'_{2j+1}(\loclet{\tcmq_j^2b^2})+1$, $\conf_{2j+2}(\loclet{\tcmt_{j+1}^2b^2})=\conf'_{2j+1}(\loclet{\tcmt_{j+1}^2b^2})+1$, $\conf_{2j+2}(\loclet{a_1^4b^4})=\conf'_{2j+1}(\loclet{a_1^4b^4})+1$, $\conf_{2j+2}(\loclet{\tcmq_j^2b})=\conf_{2j+2}(\loclet{\tcmt_{j+1}^2b})=
\conf_{2j+2}(\\\loclet{a_1^4b^3})=0$ and $\conf_{2j+2}(\loc)=\conf'_{2j+1}(\loc)$ for all other $\loc\in\Loc$. From this, we deduce that 
$\conf_{2j+2}(\loclet{\tcmq_{j+1}})=\conf'_{2j}(\loclet{\tcmq_{j+1}})=1$, $\conf_{2j+2}(\loclet{a_2^2b^2})=\conf_{2j}(\loclet{a_2^2b^2})$, and $\conf_{2j+2}(\loclet{a_1^2b^1})=\conf_{2j}(\loclet{a_1^2b^2})-1$. Moreover, $\conf_{2j+2}(\loclet{\tcmq_j^2b^2})\geq 0$, $\conf_{2j+2}(\loclet{\tcmt_{j+1}^2b^2})\\\geq 0$, $\conf_{2j+2}(\loclet{a^4b^4})\geq 0$, $\conf_{2j+2}(\loc_0)=\conf_{2j}(\loc_0)-2\geq 0$ by induction hypothesis, and for all other 
$\loc\in\Loc$, $\conf_{2j+2}(\loc)=\conf_{2j}(\loc)$. Since $\conf_{2j}\in \confset(\gamma_j)$, this implies that $\conf_{2j+2}\in \confset(\gamma_{j+1})$, as expected. Also, 
$\conf_{2j+2}(\loc_0)\geq k-(3j+1)-2\geq k- (3(j+1)+1)$.

If $\tcmt_{j+1} = \tcmq_j \xrightarrow{\tcmcounter_1==0} \tcmq_{j+1}$, then $\nu_1^j = 0=\conf_{2j}(a_1^2 b^2)$. The proof goes along the same lines as before. Now the sequence of 
configurations is necessarily: $\conf'_{2j}(\loclet{\tcmt_{j+1}})=1$, $\conf'_{2j}(\loc_0)=\conf_{2j}(\loc_0)-1$, and $\conf'_{2j}(\loc)=\conf_{2j}(\loc)$ for all other $\loc\in\Loc$,
$\conf_{2j+1}(\loclet{\tcmt_{j+1}b})=\conf_{2j+1}(\loclet{\tcmq_jb})=1$, $\conf_{2j+1}(\loclet{\tcmt_{j+1}})=\conf_{2j+1}(\loclet{\tcmq_{j}})=0$, and $\conf_{2j+1}(\loc)=\conf'_{2j}(\loc)$
for all other $\loc\in\Loc$. Then we have $\conf'_{2j+1}(\loclet{\tcmq_{j+1}})=\conf'_{2j+1}(\loclet{\tcmt_{j+1}^2b})=\conf'_{2j+1}(\loclet{\tcmq_j^2b})=1$, $\conf'_{2j+1}(\loclet{\tcmt_{j+1}b})=
\conf'_{2j+1}(\loclet{\tcmq_{j}b})=0$, $\conf'_{2j+1}(\loc_0)=\conf_{2j+1}(\loc_0)-1$, and $\conf'_{2j+1}(\loc)=\conf_{2j+1}(\loc)$ for all other $\loc\in\Loc$. Finally, we have that $\conf_{2j+2}(\loclet{\tcmt_{j+1}^2b^2})=\conf'_{2j+1}(\loclet{\tcmt_{j+1}^2b^2})+1$, $\conf_{2j+2}(\loclet{\tcmq_j^2b^2})=\conf'_{2j+1}(\loclet{\tcmq_j^2b^2})+1$, $\conf_{2j+2}(\loclet{\tcmt_{j+1}^2b})=\conf_{2j+2}(\loclet{\tcmq_j^2b})=0$, and 
$\conf_{2j+2}(\loc)=\conf'_{2j+1}(\loc)$ for all other $\loc\in\Loc$. One can check that 
$\conf_{2j+1}(\loc_0)=\conf_{2j}(\loc_0)-2\geq k-(3j+1)-2\geq k- (3(j+1)+1)\geq 0$ and that $\conf_{2j+2}\in \confset(\gamma_{j+1})$.
\qed
\end{proof}

With this lemma, we know that $\conf_{2n}$ exists and $\conf_{2n} \in \confset(\gamma_n)$. By definition of $\strat$, $\conf_{2n}'(\loclet{\tcmq_n^2})=1$, $\conf_{2n}'(\loclet{\tcmq_n})=0$ and,
for all other $\loc\in\Loc$, $\conf_{2n}'(\loc)=\conf_{2n}(\loc)$.
But this time, there are no more possible transition for Environment:
\begin{itemize}
\item Moving the process in $\loclet{\tcmq_n^2}$ to $\loclet{\tcmq_n^2b^m}$ for some $m\geq 1$ leads to a configuration either in $\accCond{F}$ or in $\accCond{\loclet{\tcmq_n^2b^m}}$ 
if $m\geq 3$.
\item Moving any other process leads to a configuration in $\accCond{\loc}$ for some $\loc$ where $\left(\sum_{\alpha \in \Alpha_\sys} \loc(\alpha)\right) < \loc(b)$.
\end{itemize}
Therefore the play is winning for System and we get a contradiction, there is no winning strategy for Environment. Thus $\strat$ is a winning $k$-strategy for System. The same strategy $\strat$ also work for any $k' > k$, which completes the first direction of the proof.\\

\fbox{$\Leftarrow$}
Suppose that there is a constant $k\in\N$ and $\strat$ a winning $k$-strategy for System.

\begin{lemma}\label{lemma:00N-right-to-left}
For any $\strat$-compatible play $\play=\conf_0 \upd_0 \conf'_0 \upd'_0 \conf_1 \upd_1 \conf'_1 \upd'_1 \conf_2 \dots \upd_{2n} \conf_{2n}$,
there exists a run 
$\gamma_0\vdash_{t_1}\gamma_1\vdash_{t_2}\dots\gamma_n$ of  $\tcm$ such that $\conf_{2i}\in \confset(\gamma_i)$, for all $1\leq i\leq n$.
\end{lemma}

\begin{proof}
Let $\play=\conf_0 \upd_0 \conf'_0 \upd'_0 \conf_1 \upd_1 \conf'_1 \upd'_1 \conf_2 \dots \upd_{2n} \conf_{2n}$ be a $f$-compatible play, not necessarily maximal. 
From $\conf_0$, the only winning configurations reachable for System, without any past action of Environment are the ones in $\accCond{t}$ for $t\in \tcmT$ of the
form $\tcmt:\tcminit\xrightarrow{c_i++}\tcmq_1$ or $\tcmt:\tcminit\xrightarrow{c_i==0}\tcmq_1$. Let $\tcmt_1$ be the transition such that $\conf'_0\in\accCond{\tcmt_1}$.  For simplicity, assume that $\tcmt_1:\tcminit\xrightarrow{c_1++}\tcmq_1$, but the other cases are similar. From $\conf'_0$, there is only one configuration reachable by Environment which is not winning for System: the one where there is exactly
one process in locations $\loclet{\tcmt_1b}$, $\loclet{\tcminit b}$, $\loclet{a_1b}$. Now that the transition $\tcmt_1$ has been selected, the only winning 
configuration reachable by System is $\conf'_1$ such that $\conf'_1(\loclet{\tcmt_1^2b})=\conf'_1(\loclet{\tcminit^2b})=\conf'_1(\loclet{a_1^2b})=\conf'_1(\loclet{\tcmq_1})=1$,
$\conf'_1(\loc_0)\geq 0$, and $\conf'_1(\loc)=0$ for all other $\loc\in\Loc$. Indeed, all other winning configurations require moves of Environment to be reached, or require that Environment has never played. 
Now, the first accepting condition prevents Environment to play $b$ on any new process, or to play several $b$s on processes that have already played. 
Moreover, if she plays $b$ on the process already in the location $\loclet{\tcmq_1}$, the configuration reached is in $\accCond{(\tcminit, \tcmt_1,\tcmq_1,\env})$.
The
only possibility to leave the set of winning configurations is then to reach $\conf_2$ defined by $\conf_2(\loclet{\tcmq_1})=\conf_2(\loclet{a_1^2b^2})=
\conf_2(\loclet{\tcmt_1^2b^2})=\conf_2(\loclet{\tcminit^2b^2})=1$, $\conf_2(\loc_0)=k-3$ and $\conf_2(\loc)=0$ for all other $\loc\in\Loc$. Hence,
$\conf_2$ is valid and $\mconf(\conf_2)=(\tcmq_1,1,0)=\gamma_1$. Moreover, $\gamma_0\vdash_{\tcmt_1}\gamma_1$.

Let now $j<n$ and suppose that we have built $\gamma_0\vdash_{\tcmt_1}\gamma_1\dots\vdash_{\tcmt_j}\gamma_j$ with $\gamma_i=(q_i,\nu_1^i, \nu_2^i)=\mconf(\conf_{2i})$ for
all $1\leq i\leq j$. From the valid configuration $\conf_{2i}$ such that $\conf_{2i}(\loclet{{\tcminit^2}b^2})>0$, the only winning configurations reachable by System 
are the ones in $\accCond{(\tcmq_j,\tcmt)}$ for some $\tcmt$ starting in $\tcmq_j$. Let $\tcmt_{j+1}$ be the transition such that 
$\conf'_{2i}\in\accCond{(\tcmq_j,\tcmt_{j+1})}$. Assume for example that $\tcmt_{j+1}:\tcmq_j\xrightarrow{c_1--}\tcmq_{j+1}$. Then $\conf'_{2j}(\loclet{\tcmt_{j+1}})=
\conf'_{2j}(\loclet{a_1^3b^2})=1$, $\conf'_{2j}(\loc_0)=\conf_{2j}(\loc_0)-1$, $\conf'_{2j}(\loclet{a_1^2b^2})=\conf_{2j}(\loclet{a_1^2b^2})-1$, and $\conf'_{2j}(\loc)=\conf_{2j}(\loc)$ for all other $\loc\in\Loc$, with $\conf_{2j}(\loc)\neq 0$ implies that $\loc\in\Locok$ since $\conf_{2j}$ is valid. As before, in order to reach a non winning configuration, the only possibility for Environment is to go to $\conf_{2j+1}$ such that $\conf_{2j+1}(\loclet{\tcmt_{j+1}b})=\conf_{2j+1}(\loclet{a_1^3b^3})=
\conf_{2j+1}(\loclet{\tcmq_jb})=1$, $\conf_{2j+1}(\loclet{\tcmt_{j+1}})=\conf_{2j+1}(\loclet{\tcmq_j})=\conf_{2j+1}(\loclet{a_1^3b^2})=0$, and $\conf_{2j+1}(\loc)=
\conf'_{2j}(\loc)$ for all other $\loc\in\Loc$. Again, the only winning configuration System can reach without the help of Environment is $\conf'_{2j+1}$ such that
$\conf'_{2j+1}(\loclet{\tcmt_{j+1}^2b})=\conf_{2j+1}(\loclet{a_1^4b^3})=
\conf_{2j+1}(\loclet{\tcmq_j^2b})=1$, $\conf_{2j+1}(\loclet{\tcmq_{j+1}})=1$, $\conf_{2j+1}(\loclet{\tcmq_{j+1}b})=\conf_{2j+1}(\loclet{a_1^3b^3})=0$, and $\conf_{2j+1}(\loc)=
\conf'_{2j}(\loc)$ for all other $\loc\in\Loc$. Finally, the analysis of all the winning conditions shows that Environment cannot play anything else that
$\conf_{2j+2}(\loclet{\tcmt_{j+1}^2b^2})=\conf_{2j}(\loclet{\tcmt_{j+1}^2b^2})+1$, $\conf_{2j+2}(\loclet{a_1^4b^4})=\conf_{2j}(\loclet{a_1^4b^4})+1$,
$\conf_{2j+2}(\loclet{\tcmq_j^2b^2})=\conf_{2j}(\loclet{\tcmq_j^2b^2}\\)+1$, and $\conf_{2j+2}(\loc)=\conf'_{2j+1}(\loc)$ for all other $\loc\in\Loc$. In particular,
$\conf_{2j+2}(\loclet{\tcmq_{j+1}})=1$, $\conf_{2j+2}(\loclet{a_1^2b^2})=\conf_{2j}(\loclet{a_1^2b^2})-1$, $\conf_{2j+2}(\loclet{a_2^2b^2})=\conf_{2j}(\loclet{a_2^2b^2})$. Hence, $\conf_{2j+2}$ is valid, and $\mconf(\conf_{2j+2})=(\tcmq_{j+1}, \nu_1^j-1, \nu_2^j)=\gamma_{j+1}$, 
with $\gamma_j\vdash_{\tcmt_{j+1}}\gamma_{j+1}$.
\qed
\end{proof}

Assume now that there is no accepting run of $\tcm$ and consider a maximal $\strat$-compatible play $\pi$. Since $\pi$ is winning, it ends in a configuration
reached by System, so it is of the form $\pi=\conf_0\upd_0\dots\conf_{2n}\dots\conf'_{m}$ with 
$m\in\{2n,2n+1\}$, for some $n\in\N$, and $\conf'_{m}\models\Acc$. By Lemma~\ref{lemma:00N-right-to-left}, we have the corresponding run 
$\gamma_0\vdash_{\tcmt_1}\gamma_1\vdash_{\tcmt_2}\dots\vdash_{\tcmt_n}\gamma_n$, with $\conf_{2n}\in \confset(\gamma_n)$ and $\gamma_n$ not a halting
configuration. An analysis of the possible moves of System in that case shows that $\conf'_{2n}(\loclet{\tcmt})=1$ for some transition $\tcmt\in\tcmT$.
But from such a configuration, Environment can easily reach a non winning configuration, by playing $b$ on every location where the number of $b$ is 
strictly smaller
than the number of letters of System. Again, System moves to configuration $\conf'_{2n+1}$, which is winning. According to the precise definition of 
$\conf'_{2n+1}$, there is only one possibility for System: $\conf'_{2n+1}(\loclet{t}^2b)=1$, $\conf'_{2n+1}(\loclet{\tcmq'})=1$ for some $\tcmq'\in\tcmQ$,
 and possibly $\conf'_{2n+1}(\loclet{a_i^2b})=1$ or $\conf'_{2n+1}(\loclet{a_i^4b^3})=1$. In any case, Environment can still reach a non winning configuration
 by playing $b$ on all these locations. Then, either the play is not maximal, or it is not winning, both of which contradict our hypotheses. Hence, there is an 
 accepting run of $\tcm$.

This concludes the proof that $\Gameproblem{\Zero}{\Zero}{\N}$ is undecidable.

\end{document}